\documentclass[reqno,11pt]{article}
\usepackage[utf8]{inputenc}   
\usepackage[T1]{fontenc}      

\usepackage{amsmath,amsthm} 
\usepackage{amssymb,mathrsfs} 
\usepackage{amssymb}
\usepackage{amsfonts}
\usepackage{upgreek}
\usepackage{nicefrac}
\usepackage{geometry}
\usepackage{authblk}
\usepackage{graphicx}
\usepackage[hypertexnames=false]{hyperref}
\hypersetup{
  colorlinks = true,
  citecolor = blue,
}
 \usepackage{color}
 \usepackage{algorithm2e}
\usepackage{stmaryrd}
\usepackage[caption = false]{subfig} 
\usepackage[shortlabels,inline]{enumitem}
\usepackage{url}
\usepackage{graphicx} 
\usepackage{lmodern} 
\usepackage{xcolor}
\usepackage{bbm}

\usepackage{ifthen}
\usepackage{xargs}

\usepackage[disable]{todonotes}

\usepackage{aliascnt}
\usepackage{cleveref}
\usepackage{autonum}
\makeatletter
\newtheorem{theorem}{Theorem}
\crefname{theorem}{theorem}{Theorems}
\Crefname{Theorem}{Theorem}{Theorems}

\newtheorem*{lemma_nonumber*}{Lemma}

\newaliascnt{lemma}{theorem}
\newtheorem{lemma}[lemma]{Lemma}
\aliascntresetthe{lemma}
\crefname{lemma}{lemma}{lemmas}
\Crefname{Lemma}{Lemma}{Lemmas}

\newaliascnt{corollary}{theorem}
\newtheorem{corollary}[corollary]{Corollary}
\aliascntresetthe{corollary}
\crefname{corollary}{corollary}{corollaries}
\Crefname{Corollary}{Corollary}{Corollaries}

\newaliascnt{proposition}{theorem}
\newtheorem{proposition}[proposition]{Proposition}
\aliascntresetthe{proposition}
\crefname{proposition}{proposition}{propositions}
\Crefname{Proposition}{Proposition}{Propositions}

\newaliascnt{definition}{theorem}

\aliascntresetthe{definition}
\crefname{definition}{definition}{definitions}
\Crefname{Definition}{Definition}{Definitions}

\newaliascnt{remark}{theorem}

\aliascntresetthe{remark}
\crefname{remark}{remark}{remarks}
\Crefname{Remark}{Remark}{Remarks}

\crefname{example}{example}{examples}
\Crefname{Example}{Example}{Examples}

\crefname{figure}{figure}{figures}
\Crefname{Figure}{Figure}{Figures}

\newtheorem{assumption}{\textbf{A}\hspace{-3pt}}
\crefformat{assumption}{{\textbf{A}}#2#1#3}

\Crefname{assumptionB}{\textbf{B}\hspace{-3pt}}{\textbf{B}\hspace{-3pt}}
\crefname{assumptionB}{\textbf{B}}{\textbf{B}}

\Crefname{assumptionC}{\textbf{C}\hspace{-3pt}}{\textbf{C}\hspace{-3pt}}
\crefname{assumptionC}{\textbf{C}}{\textbf{C}}

\Crefname{assumptionH}{\textbf{H}\hspace{-3pt}}{\textbf{H}\hspace{-3pt}}
\crefname{assumptionH}{\textbf{H}}{\textbf{H}}

\Crefname{assumptionT}{\textbf{T}\hspace{-3pt}}{\textbf{T}\hspace{-3pt}}
\crefname{assumptionT}{\textbf{T}}{\textbf{T}}

\Crefname{assumptionT}{\textbf{T}\hspace{-3pt}}{\textbf{T}\hspace{-3pt}}
\crefname{assumptionT}{\textbf{T}}{\textbf{T}}

\Crefname{assumptionL}{\textbf{L}\hspace{-3pt}}{\textbf{L}\hspace{-3pt}}
\crefname{assumptionL}{\textbf{L}}{\textbf{L}}

\newtheorem{assumptionG}{\textbf{G}\hspace{-3pt}}
\Crefname{assumptionG}{\textbf{G}\hspace{-3pt}}{\textbf{G}\hspace{-3pt}}
\crefname{assumptionG}{\textbf{G}}{\textbf{G}}

\Crefname{assumptionAR}{\textbf{AR}\hspace{-3pt}}{\textbf{AR}\hspace{-3pt}}
\crefname{assumptionAR}{\textbf{AR}}{\textbf{AR}}

\newtheoremstyle{named}%
    {}{}{\itshape}{}{\bfseries}{.}{.5em}{\thmnote{#3}}
\theoremstyle{named}
\newtheorem{assumptionD*}{Theorem}

\Crefname{assumptionQ}{\textbf{Q}\hspace{-3pt}}{\textbf{Q}\hspace{-3pt}}
\crefname{assumptionQ}{\textbf{Q}}{\textbf{Q}}

\def\msa{\mathsf{A}}
\def\msb{\mathsf{B}}
\def\msc{\mathsf{C}}

\newcommandx{\functionspace}[2][1=+]{\mathbb{F}_{#1}(#2)}
\def\supp{\operatorname{supp}}

\newcommand{\1}{\mathbbm{1}}
\newcommand{\indi}[2]{\mathbbm{1}_{#1}\left( #2 \right)}

\newcommand{\borelSet}{\mathcal{B}}

\def\Tr{\operatorname{T}}
\def\BB{\operatorname{B}}

\newcommand{\LeftEqNo}{\let\veqno\@@leqno}

\newcommand{\PE}{\mathbb{E}}
\newcommand{\PP}{\mathbb{P}}

\newcommand{\abs}[1]{\left\vert #1 \right\vert}
\newcommand{\absLigne}[1]{\vert #1 \vert}
\newcommand{\tvnorm}[1]{\| #1 \|_{\mathrm{TV}}}

\newcommandx{\Vnorm}[2][1=V]{\| #2 \|_{#1}}
\newcommandx{\VnormEq}[2][1=V]{\left\| #2 \right\|_{#1}}
\newcommandx{\norm}[2][1=]{\ifthenelse{\equal{#1}{}}{\left\Vert #2 \right\Vert}{\left\Vert #2 \right\Vert^{#1}}}
\newcommandx{\normLigne}[2][1=]{\ifthenelse{\equal{#1}{}}{\Vert #2 \Vert}{\Vert #2 \Vert^{#1}}}
\newcommandx{\normsup}[2][1=]{\ifthenelse{\equal{#1}{}}{\left\Vert #2 \right\Vert_{\infty}}{\left\Vert #2 \right\Vert^{#1}_{\infty}}}
\newcommandx{\normLine}[2][1=]{\ifthenelse{\equal{#1}{}}{\Vert #2 \Vert}{\Vert #2\Vert^{#1}}}
\newcommand{\normop}[1]{\left\Vert #1 \right\Vert}

\newcommand{\parenthese}[1]{\left(#1 \right)}

\newcommand{\parentheseDeux}[1]{\left[ #1 \right]}
\newcommand{\defEns}[1]{\left\lbrace #1 \right\rbrace }
\newcommand{\defEnsLigne}[1]{\lbrace #1 \rbrace }

\newcommand{\ps}[2]{\left\langle#1,#2 \right\rangle}
\newcommand{\eqdef}{\overset{\text{\tiny def}} =}

\newcommandx\probaMarkovTilde[2][2=]
{\ifthenelse{\equal{#2}{}}{{\widetilde{\mathbb{P}}_{#1}}}{\widetilde{\mathbb{P}}_{#1}\left[ #2\right]}}

\newcommandx{\expe}[2][1=]{\ifthenelse{\equal{#1}{}}{\PE \left[ #1 \right]}{\PE^{#1}\left[ #1 \right]}}

\newcommand{\bigO}{\ensuremath{\mathcal O}}

\newcommand{\plusinfty}{+\infty}

\def\ie{i.e.}
\def\eqsp{\;}
\newcommand{\coint}[1]{\left[#1\right)}
\newcommand{\ocint}[1]{\left(#1\right]}
\newcommand{\ooint}[1]{\left(#1\right)}
\newcommand{\ccint}[1]{\left[#1\right]}

\newcommand{\oointLigne}[1]{(#1)}

\newcommandx{\weight}[2][2=n]{\omega_{#1,#2}^N}

\newcommand{\boule}[2]{\operatorname{B}(#1,#2)}
\newcommand{\ball}[2]{\operatorname{B}(#1,#2)}

\def\rmd{\mathrm{d}}
\newcommandx\sequence[3][2=,3=]
{\ifthenelse{\equal{#3}{}}{\ensuremath{( #1_{#2})}}{\ensuremath{( #1_{#2})_{ #2 \in #3}}}}
\newcommandx{\sequencen}[2][2=n\in\N]{\ensuremath{(#1)_{#2}}}
\newcommandx\sequenceDouble[4][3=,4=]
{\ifthenelse{\equal{#3}{}}{\ensuremath{\{ (#1_{#3},#2_{#3}) \}}}{\ensuremath{\{  (#1_{#3},#2_{#3}), \eqsp #3 \in #4 \}}}}
\newcommandx{\sequencenDouble}[3][3=n\in\N]{\ensuremath{\{ (#1_{n},#2_{n}), \eqsp #3 \}}}
\newcommand{\wrt}{w.r.t.}

\def\rme{\mathrm{e}}

\def\eg{e.g.}

\def\rset{\mathbb{R}}
\def\rsetep{\rset^*_+}

\def\rsetp{\rset_+}

\def\nset{\mathbb{N}}
\def\nsets{\nset^{*}}

\newcommandx{\CPE}[3][1=]{{\mathbb E}^{#3}_{#1}\left[#2 \right]} 
\newcommandx{\CPVar}[3][1=]{\mathrm{Var}^{#3}_{#1}\left\{ #2 \right\}}
\newcommand{\CPP}[3][]
{\ifthenelse{\equal{#1}{}}{{\mathbb P}\left(\left. #2 \, \right| #3 \right)}{{\mathbb P}_{#1}\left(\left. #2 \, \right | #3 \right)}}

\def\mcb{\mathcal{B}}

\newcommand{\chunk}[4][]%
{\ifthenelse{\equal{#1}{}}{\ensuremath{{#2}_{#3:#4}}}{\ensuremath{#2^#1}_{#3:#4}}
}

\def\Id{\operatorname{Id}}

\def\eventA{\mathsf{A}}
\def\eventB{\mathsf{B}}

\newcommand{\N}{\mathbb{N}}

\def\complementary{\operatorname{c}}

\def\Ham{H}

\def\Leb{\operatorname{Leb}}
\def\Lip{\operatorname{Lip}}

\def\constone{\operatorname{A}_3}
\def\consttwo{\operatorname{A}_4}
\def\constthree{\operatorname{A}_1}
\def\constfour{\operatorname{A}_2}
\def\constfive{\operatorname{A}_5}
\def\rhtwo{\operatorname{R_U}}
\def\expozero{\beta}

\def\alphaacc{\alpha_{\Ham}}
\def\balphaacc{\bar{\alpha}_{\Ham}}
\def\alphagen{\alpha}

\def\constzero{\mathrm{L}_1}
\def\constzeroT{\mathrm{M}_1}

\def\Xset{\mathsf{X}}
\def\Xtribu{\mathcal{X}}

\def\Vgeo{V}
\def\lambdageo{\lambda}
\def\bgeo{b}
\def\lambdageotilde{\tilde{\lambda}}
\def\bgeotilde{\tilde{b}}
\def\rejectregion{\mathscr{R}}
\def\ballV{\mathscr{B}}

\def\ouvert{\mathsf{U}}
\def\open{\ouvert}

\def\hfunb{f}
\def\phib{\phi}
\def\Pkerb{\mathrm{K}}
\def\constLiphx{L_{\hfunb}}

\def\rassG{R}

\def\Pker{\mathrm{P}}
\def\Kker{\mathrm{K}}

\def\Psiverlet{\Psi}
\newcommandx{\gperthmc}[2][1=,2=]{\ifthenelse{\equal{#1}{}}{\Xi}{\ifthenelse{\equal{#2}{}}{\Xi_{h,#1}}{\Xi_{#2,#1}}}}
\newcommandx{\Phiverlet}[2][1=,2=]{\ifthenelse{\equal{#1}{}}{\Phi}{\Phi_{#1}^{\circ (#2)}}}
\newcommandx{\gpertub}[2][1=,2=]{\ifthenelse{\equal{#1}{}}{g}{g_{#1}^{#2}}}
\newcommandx{\Phiverletq}[2][1=,2=]{\ifthenelse{\equal{#1}{}}{\widetilde{\Phi}}{\widetilde{\Phi}_{#1}^{\circ (#2)}}}
\newcommandx{\Phiverletqi}[2][1=,2=]{\ifthenelse{\equal{#1}{}}{\bar{\Psi}}{\bar{\Psi}_{#1}^{(#2)}}}
\newcommandx{\Pkerhmc}[2][1=,2=]{\ifthenelse{\equal{#1}{}}{\mathrm{P}}{\mathrm{P}_{#1, #2}}}
\newcommandx{\tPkerhmc}[2][1=,2=]{\ifthenelse{\equal{#1}{}}{\tilde{\mathrm{P}}}{\tilde{\mathrm{P}}_{#1, #2}}}
\newcommandx{\PkerhmcD}[2][1=,2=]{\ifthenelse{\equal{#1}{}}{\mathrm{K}}{\mathrm{K}_{#1, #2}}}

\def\projq{\mathrm{proj}}

\def\phia{\phi}

\def\detj{\mathrm{D}}
\def\det{\operatorname{det}}
\def\Jac{\mathbf{J}}
\def\sign{\operatorname{sign}}

\def\ga{g}

\def\ra{R}
\def\Cga{C}

\def\bg{b}
\def\hga{f^{\ga}}
\def\hog{\mathrm{H}^g}

\def\Dset{\mathsf{D}}
\def\Dsetc{\overline{\mathsf{D}}}
\def\dist{\operatorname{dist}}
\def\hpy{\mathrm{H}}
\def\matrix{\mathrm{M}}
\def\linearmap{\mathrm{L}}
\def\harmonic{\phi}

\def\lipgr{L_{\ga,\ra}}
\def\kernel{\mathrm{K}}
\def\qker{\mathrm{k}}
\def\MassG{M}
\def\tMassG{\tilde{M}}
\def\q{q}
\def\Q{Q}
\def\p{p}
\def\P{P}
\def\x{x}

\def\Rexp{R_1}
\def\m{m}
\def\lset{\mathsf{L}}
\newcommand{\clos}[1]{\overline{#1}}
\newcommand{\interior}[1]{#1^{\circ}}
\newcommand{\boundary}[1]{\partial #1}

\def\F{U}
\def\U{\F}
\def\a{a}
\def\b{b}

\newcommandx{\Vdrifta}[1][1=]{V_{#1}}
\newcommandx{\Eproof}[2][1=h,2=T]{\operatorname{E}_{#1,#2}}

\newcommand{\ensemble}[2]{\left\{#1\,:\eqsp #2\right\}}
\newcommand{\set}[2]{\ensemble{#1}{#2}}

\def\iff{if and only if}

\def\randomkerhmc{\overline{\mathrm{P}}}

\def\Csf{\mathsf{C}}
\def\Asf{\mathsf{A}}
\def\Bsf{\mathsf{B}}
\def\tildeAlphaacc{\tilde{\alpha}_{\Ham}}
\def\bfvarpi{\pmb{\varpi}}
\def\Tkernel{$\operatorname{T}$-kernel}
\def\Tker{\mathrm{T}}

\def\Sigmabf{\boldsymbol{\Sigma}}

\def\rmD{\mathrm{D}}

\title{On the convergence of Hamiltonian Monte Carlo}

\author[1]{Alain Durmus \footnote{Email: alain.durmus@cmla.ens-cachan.fr} }
\author[2]{\'Eric Moulines \footnote{Email: eric.moulines@polytechnique.edu} }
\author[3]{Eero Saksman \footnote{Email: eero.saksman@helsinki.fi} }

\affil[1]{CMLA - \'Ecole normale supérieure Paris-Saclay, CNRS, Université Paris-Saclay, 94235 Cachan, France.}
\affil[2]{Centre de Math\'ematiques Appliqu\'ees,\\ UMR 7641, Ecole Polytechnique}
\affil[3]{University of Helsinki, Department of Mathematics and Statistics}

\begin{document}

\maketitle

\begin{abstract}
This paper discusses the irreducibility and geometric ergodicity of the Hamiltonian Monte Carlo (HMC) algorithm.
We consider cases where the number of steps of the symplectic integrator is either  fixed or random. Under mild conditions on the potential $\F$ associated with target distribution $\pi$, we first show that the Markov kernel associated to the HMC algorithm is irreducible and recurrent.
Under more stringent conditions, we then establish that the Markov kernel is Harris recurrent. Finally, we provide verifiable  conditions on $\F$  under which the HMC sampler is geometrically ergodic. We compare our assumptions with those recently presented in  \cite{livingstone:betancourt:byrne:girolami:2016} and \cite{bou:sanz:2017}.
\end{abstract}

\section{Introduction }

We consider in this paper the Hamiltonian Monte
Carlo (HMC),  a Metropolis-Hastings
algorithm designed to sample  target probability density
$\pi$ on $\rset^d$. This method was first proposed by \cite{duane:1987}  in computational physics. It has later been introduced in the  statistics community in the early paper of \cite{neal:1993} and quickly  gained popularity; see for example \cite[chapter~9]{liu:2008}, \cite{neal:2011} and \cite{girolami:calderhead:2011}. The most attractive feature of the HMC algorithm is to allow the possibility of generating proposals -  obtained by integrating a system of Hamiltonian equations - that are far away from the current position but still having a high probability of being accepted. The HMC algorithm therefore offer promise for eliminating the random walk behavior of most classical Monte Carlo algorithms. The distance between the current state and the proposal is controlled by length of the time interval along which the Hamiltonian equations are integrated; see \cite[chapter~9]{liu:2008} and \cite{sanz-serna:2014}.

HMC algorithms have achieved many empirical successes. Recently, the theory on HMC have been addressed by many authors 
; see
\cite{byrne:girolami:2013,tang:srivastava:salakhutdinov:2014,
schofield:barker:gelman:cook:briffa:2016,betancourt-bernoulli:2017,livingstone:betancourt:byrne:girolami:2016}. An in depth discussion of the HMC method and a survey of the existing results are given in \cite{bou-rabee:sanz-serna:2018}.

Consider a target probability density  $\pi$ on $\rset^d$ with respect to the Lebesgue measure,  defined for all $\q \in \rset^d$ by
\begin{equation}
\label{eq:def_density_pi}
\pi(\q) = \left. \rme^{-\F(\q)} \middle / \int_{\rset^d}\rme^{-\F(\tilde{\q})} \rmd \tilde{\q} \right.\eqsp,
\end{equation}
where $\F: \rset^d \to
\rset$ is a continuously differentiable function. Note that this representation implies that the density is nonzero everywhere (this can be relaxed; see \cite[Section~5.5.1]{neal:2011}).

The properties of Hamiltonian dynamics have been discussed in numerous
papers. We provide here only a brief outlook mainly aimed at
introducing the notations and the essence of the main ideas. We refer
the interested readers to the monograph \cite{leimkuhler:reich:2004}
and the surveys given in \cite[Chapter~9]{liu:2008}, \cite{neal:2011},
\cite{betancourt-bernoulli:2017} and
\cite{bou-rabee:sanz-serna:2018}. The key idea behind HMC is to
exploit the measure-preserving properties of Hamiltonian flow over an
extended phase space. For simplicity, we restrict our study to the
phase space $\rset^{2d}$. Hamiltonian dynamics describes the evolution
of a physical system which consists in the \emph{position}
$\q \in \rset^d$ and the \emph{momentum} $\p \in \rset^d$.
The total energy of the system is given by the Hamiltonian function $\Ham$ defined for  $(\q,\p) \in \rset^d \times \rset^d$ by
\begin{equation}
  \label{eq:def_ham}
\Ham(\q,\p) = \F(\q) + \norm[2]{\p}/2 \eqsp,
\end{equation}
which is the sum of a potential energy $U : \rset^d \to \rset$, a
function solely of the position, and the kinetic energy
$p \mapsto \norm[2]{p}/2$ (note that other choices of kinetic energy
have proposed recently, see \eg~\cite{livingstone:faulkner:roberts:2017} and \cite{lu:et:al:2016}).
The system then evolves in time $(\q(t),\p(t))_{t \geq 0}$ according to Hamilton's equations on $\rset^d \times \rset^d$,
\begin{equation}
  \label{eq:hamil_ode}
  \begin{aligned}
\frac{\rmd }{\rmd t}
\left[
  \begin{array}{c}
    q(t) \\
    p(t) \\
  \end{array}
\right]
& = J^{-1} \nabla \Ham(\q(t),\p(t)) \eqsp, \\
\quad \text{where} \quad  J&= \left[
                                         \begin{array}{cc}
                                           0_{d \times d} & -\Id_{d \times d} \\
                                           \Id_{d \times d} & 0_{d \times d} \\
                                         \end{array}
                                       \right]
                                       \eqsp,
                                       \quad \nabla \Ham(\q,\p) = \begin{bmatrix} \nabla U(\q) \\ \p \end{bmatrix}           \eqsp.
                                     \end{aligned}
                                   \end{equation}
We denote by $(\varphi_t)_{t \geq 0}$ the differential flow associated to the system \eqref{eq:hamil_ode}. For each $t \in \rset$, $\varphi_t: \rset^{2d} \to \rset^{2d}$ is the map that associates to each $(\p_0,\q_0)$ the value at time $t$ of the (unique) solution of \eqref{eq:hamil_ode} that takes the value $(\p_0,\q_0)$ at time $t=0$. We shall assume hereafter that $\varphi_t(\p_0,\q_0)$ is defined for any $(\p_0,\q_0) \in \rset^{2d}$ and $t \in \rset$.

A mapping $\Phi: \rset^{2d} \to \rset^{2d}$ is said to be \emph{symplectic} if, at each point $(p,q) \in \rset^{2d}$, $\Jac_{\Phi}(p,q)^T J \Jac_{\Phi}(p,q)= J$, where $\Jac_{\Phi}(p,q)$ denotes the $2d \times 2d$ Jacobian of $\Phi$. Note that in particular, 
symplectic transformations are volume preserving on $\rset^{2d}$.
An important property of Hamiltonian systems \eqref{eq:hamil_ode} is that, for each $t \in \rset$, $\varphi_t$ is a symplectic mapping; see \cite[Theorem~2.1]{bou-rabee:sanz-serna:2018}.

Another important property of Hamiltonian flow is the conservation of energy.
Since $J^{-1}$ is skew-symmetric, for any solution $(\q(t),\p(t))$ of \eqref{eq:hamil_ode}
\[
\frac{\rmd }{\rmd t} \Ham(\q(t),\p(t)) =\nabla \Ham(\q(t),\p(t))^T J^{-1} \nabla \Ham(\q(t),\p(t))= 0.
\]
Then, the value of the Hamiltonian function is preserved by the flow of the corresponding Hamiltonian system, $\Ham \circ \varphi_t= \Ham$ for each $t \in \rset$.

Denote by $S$ the momentum flip involution, $S(\q,\p)= (\q,-\p)$, $(\q,\p) \in \rset^{2d}$. A mapping $\Phi: \rset^{2d} \to \rset^{2d}$ is said to be \emph{reversible} with respect to $S$ (or $S$-reversible for short) if $ S \circ \Phi= \Phi^{-1} \circ S$. If the mapping $\Phi$ is differentiable, the $S$-reversibility implies that $|\det \Jac_{\Phi} (S \circ \Phi)|= |\det \Jac_{\Phi}|^{-1}$, where $\Jac_{\Phi}$ is the Jacobian of $\Phi$.
By uniqueness of solutions of
\eqref{eq:hamil_ode}, for all $t \in \rset$, the flow $\varphi_t$ is  a
$S$-reversible mapping. More precisely, if $(\q_0,\p_0)$ is the initial state and
$(\q(t),\p(t))= \varphi_t(\q_0,\p_0)$ is the state of the system after
$t$ units of times, then
\[
\varphi_t(\q(t),-p(t))= \varphi_t( S \circ \varphi_t(\q_0,\p_0)= (\q_0,-\p_0)=S (\q_0,\p_0) \eqsp.
\]

Consider the extended target distribution with density given for any $(\q,\p) \in \rset^{2d}$ by 
\begin{equation}
\label{eq:def_ext_pi}
 \tilde{\pi}(\q,\p) = Z^{-1}
\exp(-\Ham(\q,\p)) \eqsp,   \, \, Z= \int_{\rset^{2d}} \exp(-\Ham(\q,\p)) \rmd \q \rmd \p \eqsp.
\end{equation}
Since the flow $\varphi_t$ preserves the oriented volumes and the Hamiltonian, the probability measure with density $\tilde{\pi}$ is preserved by the flow $\varphi_t$, for each $t \in \rset$ and $A \in \borelSet(\rset^{2d})$ (the Borel sets of $\rset^{2d}$), $\tilde{\pi}(\varphi_t(A))= \tilde{\pi}(A)$ where, with a slight abuse in notations,
$\tilde{\pi}(A)= \int_{\rset^{2d}} \indi{A}{q,\p} \tilde{\pi}(\q,\p) \rmd\q \rmd \p$.

For the Hamiltonian function \eqref{eq:def_ham}, the density $\tilde{\pi}$ may be factorized
\[
\exp(-\Ham(\q,\p)) = \exp(-\F(\q)) \exp(- \norm[2]{\p}/2)
\]
and then, under the distribution $\tilde{\pi}$, the position $\q$ and the momentum $\p$ are independent,
the marginal distribution of the position has a probability density function proportional to the target distribution \eqref{eq:def_density_pi} and the momentum $\p$ is Gaussian with zero mean and identity covariance matrix.

In most cases, it is not possible to compute explicitly the solutions of \eqref{eq:hamil_ode}; discretization must be used instead.
 A crucial point in the construction of HMC sampler is that  symplectiness and $S$-reversibility can be preserved exactly by discretization, provided that we use a \emph{symplectic integrator} like the Störmer-Verlet (referred to as \emph{leap-frog}) integrator. The Hamiltonian is not exactly preserved in the discretization, but it is expected that a sensible integrator conserves this quantity at least "approximately".
 Given a time step $h \in \rset^*_+$ and a number of iterations $T \in \nset^*$, the Störmer-Verlet integrator proceeds as follows: starting from an initial point $(\q_0,\p_0) \in
\rset^d \times \rset^d$, $T$ leap frog steps are performed, where for each $\ell \in \{0,\dots,T-1\}$ the $\ell$-th leap frog step is defined as
\begin{equation}
\label{eq:iteration_verlet}
  \begin{cases}
\p_{\ell+1/2} &= \p_\ell - (h/2) \nabla \F(\q_\ell)\\
\q_{\ell+1} &= \q_\ell +h\p_{\ell+1/2}\\
\p_{\ell+1} &= \p_{\ell+1/2}-(h/2) \nabla \F(\q_{\ell+1}) \eqsp.
  \end{cases}
\end{equation}
The sequence $(\q_\ell,\p_\ell)_{\ell \in \{0,\dots,T\}}$ is an
approximation of the solution of \eqref{eq:hamil_ode} at times $\{\ell h
\, : \, \ell \in \{0,\ldots, T\}\}$ started at $(\q_0,\p_0)$. This sequence defines a discrete dynamical system given for $\ell \in \{0,\ldots,T-1\}$ by
\begin{equation}
  (\q_{\ell+1},\p_{\ell+1}) =  \Psiverlet^{(1)}_{h/2} \circ \Psiverlet^{(2)}_h \circ \Psiverlet^{(1)}_{h/2} (\q_\ell,\p_\ell) = \Phiverlet^{(1)}_h(\q_\ell,\p_\ell) \eqsp,
\end{equation}
where for each $t \in \rset_+$, $\Psiverlet^{(1)}_t, \Psiverlet^{(2)}_t :\rset^{2d} \to \rset^{2d}$ are given for all $(\q,\p) \in \rset^{2d}$ by
\begin{equation}
  \label{eq:def_Psiverlet_0}
\Psiverlet^{(1)}_t(\q,\p) = (\q, \p-t\nabla \F(\q)) \eqsp, \quad\Psiverlet^{(2)}_t(\q,\p) = (\q+t\p, \p)
\end{equation}
These mappings in molecular dynamics are called the \emph{kick} (the system stays in its current configuration and the momentum is incremented by action of the force $\nabla U(\q)$) and the \emph{drift} (the position $q$ advances at constant speed while the momentum $\p$ remains constant). One iteration of the Störmer-Verlet formula comprises two kicks of duration $h/2$ separated by a drift of duration $h$.
Define the sequence of iterates $\{\Phiverlet[h][\ell] : \rset^d \times \rset^d \to \rset^d \times \rset^d \, : \, \ell \in \nset^*\} $ for $\ell \geq 1$ by induction
\begin{equation}
  \label{eq:def_Phiverlet}
\Phiverlet[h][\ell+1] = \Phiverlet[h][\ell] \circ \Phiverlet[h][1] \eqsp, \quad
\end{equation}
Set for all $\ell \geq 1$,  
\begin{equation}
  \label{eq:def_Phiverletq}
  \Phiverletq[h][\ell] = \projq \circ \Phiverlet[h][\ell] \eqsp,
\end{equation}
where $ \projq : \rset^d \times \rset^d \to \rset^d$ is the projection on the first $d$ coordinates, for all $(\q,\p) \in \rset^d \times \rset^d$, $\projq(\q,\p) = \q$.
Thus, with our notation for all $\ell \in \{1,\ldots, T\}$, $ (\q_\ell,\p_\ell) =\Phiverlet[h][\ell](\q_0,\p_0)$ and $ \q_\ell = \Phiverletq[h][\ell](\q_0,\p_0)$

 Because each inner step in the leap-frog step are \emph{shear} transformations of the phase variable (only the position or the momentum are updated by a quantity that depends only on the variable that do not change), it is clear that this transformation is volume preserving (the Jacobian of each individual transformation is equal to $1$). Each inner leap frog step due of its symmetry is also $S$-reversible: starting from $(\q_{\ell+1},-\p_{\ell+1})$, applying the leap-frog step forward and then negating the momentum variable again, we obtain again $(\q_\ell,\p_\ell)$.

We now have all the background required to describe the HMC algorithm,
which is a special instance of Metropolis-Hastings algorithm aimed at sampling the  target distribution  $\pi$.
It is similar to most classical MCMC algorithm in that we propose a new point based on the current position and then either accept or reject it.
Denote by $(\Q_k,\P_k)$ the value of the position and momentum at the $k$-th iteration of the algorithm. Each iteration of the algorithm may be decomposed into two steps, which
are constructed to leave the extended
distribution $\tilde{\pi}$ invariant; see \cite{neal:2011}, \cite{fang:sanz-serna:skeel:2016} and \cite[Theorem 5.2]{bou-rabee:sanz-serna:2018}.
In the first step, we draw $G_{k+1}$ from the $d$-dimensional normal
distribution with zero mean and identity covariance matrix, independent of $\{(\Q_j,\P_j)\}_{j=0}^k$.
In the second step, we set the initial conditions $(\Q_k,G_{k+1})$ and compute the position and the momentum after $T$ leapfrog steps. This move is accepted with probability $\alphaacc\defEns{(\Q_k,G_{k+1}),\Phiverlet[h][T](\Q_k,G_{k+1})}$ where
for all $(\q,\p) \in \rset^d \times \rset^d$, $(\tilde{\q},\tilde{\p}) \in \rset^d \times \rset^d$ by
\begin{equation}
\label{eq:def_acc_ratio}
  \alphaacc\defEns{(\q,\p),(\tilde{\q},\tilde{\p})} = \min \parentheseDeux{1,   \exp\parenthese{\Ham(\q,\p)-\Ham(\tilde{\q},\tilde{\p}})} \eqsp.
\end{equation}
It is easily seen that $\tilde{\pi}$ is invariant with respect to $\tilde{\pi}$ (see \cite{fang:sanz-serna:skeel:2016}).  Since $\tilde{\pi}$ \eqref{eq:def_ext_pi} is invariant with respect to the Markov kernel defined by the HMC algorithm on the extended state
space $\rset^d \times \rset^d$, it naturally implies that $\pi$ is a
stationary distribution for the Markov chain $(\Q_k)_{k \geq 0}$,
which is the process which we are interested in. The number of steps $T$ is either a deterministic quantity or a random variable independent of the current state. If the number of
steps $T=1$, then the algorithm reduces to the Metropolis Adjusted
Langevin Algorithm (MALA).



Despite many recent advances, theoretical properties of the HMC algorithm are still not completely
understood. This paper addresses two important issues in the analysis of HMC algorithm: irreducibility and geometric ergodicity.

Irreducibility  plays
an essential role in the theory of Markov chains. In particular, it implies uniqueness of a
invariant distribution. The classical approach to derive
irreducibility of Hastings-Metropolis algorithms on $\rset^d$, outlined for example in \cite{mengersen:tweedie:1996}
\cite{roberts:tweedie:1996:biometrika}, is to use that the proposal
distribution admits a (sufficiently regular) transition density with respect to the Lebesgue
measure. 
For HMC, this condition does not necessarily hold. HMC has been shown
to be irreducible in \cite{cances:legoll:stoltz} in the case where the
state space is compact and the potential is twice continuously
differentiable. In \cite{livingstone:betancourt:byrne:girolami:2016},
under appropriate conditions, irreducibility is shown for a version
of HMC where the number of leap-frog steps $T$ is random, independent
of the proposal, and such that $T=1$ with positive probability.  Under such assumption, irreducibility of HMC boils down to irreducibility of MALA which has been established in
\cite{roberts:tweedie:1996}.  In this paper, we establish the
irreducibility of the HMC algorithm under a general tail condition of
the target density which significantly relaxes the condition of
\cite{cances:legoll:stoltz} and
\cite{livingstone:betancourt:byrne:girolami:2016}.  This result
follows from a general irreducibility result for iterative Markov
models (derived under conditions which are weaker than the ones
reported in the literature) which we believe to be of independent
interest; see \Cref{sec:irred-class-iter}.  Our main tool to establish
irreducibility is the degree theory for continuous maps \cite{outerelo:ruiz:2009}.

In a second part, we establish the geometric ergodicity of the HMC
sampler under the assumptions that the potential $U$ is homogeneous
outside a ball (or is a perturbation of an homogeneous function) and
that the level sets are convex.  Our assumptions
imply that the proposal kernel of HMC satisfies an ‘inwards
acceptance’ property \cite{roberts:tweedie:1996}, which is essential
to show that HMC (and MALA) is geometrically ergodic.

Our results complement the recent paper \cite{livingstone:betancourt:byrne:girolami:2016}. This paper provides a variety of conditions under which the HMC algorithm is not geometrically ergodic. It establishes the geometric ergodicity under the abstract 'inwards acceptance' property for which we provide verifiable sufficient conditions.

In \cite{bou:sanz:2017}, a variant of HMC, referred to as the
Randomized Hamiltonian Monte Carlo (RHMC), is analyzed. This method is
associated with a conti\-nuous-time Markov process for which
$\tilde{\pi}$ given by \eqref{eq:def_ext_pi} is invariant
\cite[Proposition 3.1]{bou:sanz:2017}.  However, sampling such a
process requires the exact Hamiltonian flow (and hence exact
integration of the Hamilton dynamics \eqref{eq:hamil_ode}). The use of
the Hamiltonian flow allows to by-pass the acceptance-rejection step
and makes the analysis easier.
By-passing the discretization step nevertheless reduces the
applicability of the results, since direct integration of the
Hamiltonian flow is most of the time not an option.
We numerically show on a simple example
that the conditions given by \cite{bou:sanz:2017} which imply
geometric ergodicity of RHMC are not sufficient in the case of HMC.

The paper is organized as follows. In \Cref{sec:ergodicity-hmc}, conditions upon which the
HMC kernel, associated with $(Q_k)_{k \in \nset}$, is irreducible, recurrent and Harris-recurrent are given.
In \Cref{sec:geom-ergod-hmc}, conditions under which the HMC kernel is $V$-uniformly geometrically ergodic  are developed and discussed. Some general irreducibility results which are of independent interest, are stated in \Cref{sec:irred-class-iter}. The proofs are gathered in
\Cref{sec:postponed-proofs}.


\subsection*{Notations}
Denote by $\rset_+$ and $\rset_+^*$, the set of non-negative and
  positive real numbers respectively. Denote by $\operatorname{I}_n$ the identity matrix. Denote by $\norm{\cdot}$ the
Euclidean norm on $\rset^d$.  Denote by $\borelSet(\rset^d)$ the
Borel $\sigma$-field of $\rset^d$, $\functionspace[]{\rset^d}$ the set
of all Borel measurable functions on $\rset^d$ and for $f \in
\functionspace[]{\rset^d}$, $\Vnorm[\infty]{f}= \sup_{x \in \rset^d}
\abs{f(x)}$. Denote by $\Leb$ the Lebesgue-measure on $\rset^d$. For
$\mu$ a probability measure on $(\rset^d, \borelSet(\rset^d))$ and
$f \in \functionspace[]{\rset^d}$ a $\mu$-integrable function, denote
by $\mu(f)$ the integral of $f$ \wrt~$\mu$.  For $f \in
\functionspace[]{\rset^d}$, set $\Vnorm[\infty]{f}= \sup_{x \in
  \rset^d} \abs{f(x)}$.  Let $V: \rset^d \to \coint{1,\infty}$ be a
measurable function. For $f \in \functionspace[]{\rset^d}$, the
$V$-norm of $f$ is given by $\Vnorm[V]{f}= \Vnorm[\infty]{f/V}$. For
two probability measures $\mu$ and $\nu$ on $(\rset^d,
\borelSet(\rset^d))$, the $V$-total variation distance of $\mu$ and
$\nu$ is defined as
\begin{equation}
\Vnorm[V]{\mu-\nu} = \sup_{f \in \functionspace[]{\rset^d}, \Vnorm[V]{f} \leq 1}  \abs{\int_{\rset^d } f(x) \rmd \mu (x) - \int_{\rset^d}  f(x) \rmd \nu (x)} \eqsp
\end{equation}
If $V \equiv 1$, then $\Vnorm[V]{\cdot}$ is the total variation  denoted by $\tvnorm{\cdot}$.
For all $x \in \rset^d$ and $M >0$, we denote by $\boule{x}{M}$, the
ball centered at $x$ of radius $M$.  Let $\matrix$ be a $d \times
m$-matrix, then denote by $\matrix^{\Tr}$ and $\det(\matrix)$ (in the case $m=d$) the
transpose and the determinant of $\matrix$ respectively.  Let $k \geq
1$. Denote by $(\rset^{d})^{\otimes k}$ the $k^{\text{th}}$ tensor
power of $\rset^d$, for all $x \in \rset^d, y \in \rset^{\ell}$, $x
\otimes y \in (\rset^d)^{\otimes 2}$ the tensor product of $x$ and
$y$, and $x^{\otimes k} \in (\rset^{d})^{\otimes k}$ the
$k^{\text{th}}$ tensor power of $x$.  For all $x_1,\ldots,x_k \in
\rset^d$, set $\norm{x_1\otimes \cdots \otimes x_k} = \sup_{i \in
  \{1,\ldots,k\}} \norm{x_i}$. We let $\mathcal{L}((\rset^{d})^{\otimes
  k} , \rset^{\ell})$ stand for the set of linear maps from
$(\rset^{n})^{\otimes k} $ to $\rset^{\ell}$ and for $\linearmap \in
\mathcal{L}((\rset^{d})^{\otimes k} , \rset^{\ell})$, we denote by
$\normop{\linearmap}$ the operator norm of $\linearmap$. Let $f :
\rset^d \to \rset^{\ell}$ be a Lipschitz function, namely there exists
$C \geq 0$ such that for all $x,y \in \rset^d$, $ \norm{f(x) - f(y)}
\leq C \norm{x-y}$. Then we denote $\norm{f}_{\Lip} = \inf \{
\norm{f(x) - f(y)}/\norm{x-y} \ | \ x,y \in \rset^d , x \not = y
\}$.  Let $k \geq 0$ and $\open$ be an open subset of
$\rset^d$. Denote by $C^k(\open,\rset^{\ell})$ the set of all $k$
times continuously differentiable funtions from $\open$ to
$\rset^{\ell}$. Let $\Phi \in C^k(\open, \rset^{\ell})$. Write
$\Jac_{\Phi}$ for the Jacobian matrix of $\Phi \in C^{1}(\rset^d,
\rset^{\ell})$, and $D^k \Phi : \open \to
\mathcal{L}((\rset^{d})^{\otimes k} , \rset^{\ell})$ for the
$k^{\text{th}}$ differential of $\Phi \in C^{k}(\rset^d,
\rset^{\ell})$. For smooth enough functions $f
: \rset^d \to \rset$, denote by $\nabla f$ and $\nabla^2 f$ the gradient
and the Hessian of $f$ respectively. Let $\eventA \subset
\rset^d$. We write  $\clos{\eventA},\interior{\eventA}$ and
$\boundary{\eventA}$ for the closure, the interior and the boundary of
$\eventA$, respectively. For any $n_1,n_2 \in \nset$, $n_1 > n_2$, we take the convention that $\sum_{k=n_2}^{n_1} = 0$. 

\section{Ergodicity of the HMC algorithm}
\label{sec:ergodicity-hmc}

For $h >0$ and $T \in \nset^*$, consider the Markov kernel $\Pkerhmc[h][T]$ associated with the Markov chain of the HMC algorithm $(Q_k)_{k \in \nset}$, given for all $\q \in \rset^d$ and $\eventA \in \borelSet(\rset^d)$ by
\begin{align}
  \Pkerhmc[h][T](\q, \eventA) &= \int_{\rset^d} \indi{\eventA}{\Phiverletq[h][T](\q,\tilde{\p})} \ \alphaacc\defEns{(\q,\tilde{\p}),\Phiverlet[h][T](\q,\tilde{\p})}\frac{\rme^{-\norm[2]{\tilde{\p}}/2}}{ (2 \uppi)^{d/2}}  \rmd \tilde{\p}
                                  \nonumber
  \\
&\qquad + \updelta_{\q}(\eventA)  \,   \int_{\rset^d}  \parentheseDeux{1-\alphaacc\defEns{(\q,\tilde{\p}),\Phiverlet[h][T](\q,\tilde{\p})}} \frac{\rme^{-\norm[2]{\tilde{\p}}/2}}{ (2 \uppi)^{d/2}}  \rmd \tilde{\p} \eqsp,
 \label{eq:def_kernel_hmc}
\end{align}
where $\Phiverletq[h][T]$, $\Phiverlet[h][T]$ and $\alphaacc$ are defined by \eqref{eq:def_Phiverlet}-\eqref{eq:def_Phiverletq} and \eqref{eq:def_acc_ratio} respectively.
In this Section, we establish conditions upon which the Markov kernel $ \Pkerhmc[h][T]$ is irreducible or
(Harris) recurrent. 
For
$\expozero \in \ccint{0,1}$, we consider the following assumption on
the potential $\F$.

\begin{assumption}[$\expozero$]
  \label{assum:regOne}
  $\F$ is continuously differentiable and
  \begin{enumerate}[label=(\roman*)]
  \item
  \label{assum:regOne_a}
 there exists $\constzero > 0$  such that for all $\q,x \in \rset^d$,
\begin{equation}
\norm{\nabla \F(\q) - \nabla \F(x)} \leq \constzero\norm{\q-x} \eqsp.
  \end{equation}
\item    \label{assum:regOne_b}
there exists $\constzeroT \geq 0$  such that for all $\q \in \rset^d$,
\begin{equation}
  \norm{\nabla \F(\q)} \leq \constzeroT\defEns{ 1 + \norm{\q}^{\expozero}} \eqsp.
\end{equation}
  \end{enumerate}
\end{assumption}

Before going further, we need to briefly recall some definitions pertaining to Markov chains.
Let $\Pker$ be a Markov kernel on $(\rset^d,\borelSet(\rset^d))$. Let $n$ be an integer and $\mu$
be a nontrivial measure on $\borelSet(\rset^d)$. A
set $\Csf \in \borelSet(\rset^d)$ is called a $(n,\mu)$-small set for $\Pker$ if
for all $x \in \Csf$ and $\Asf \in \borelSet(\rset^d)$, $\Pker^n(x, \msa) \geq \mu(\msa)$.
A set $\Asf \in \borelSet(\rset^d)$ is said to be accessible for $\Pker$
  if for all $x \in \rset^d$, $\sum_{i=1}^\infty \Pker^i(x,\Asf) > 0$.
  A non-trivial $\sigma$-finite measure $\mu$ is an irreducibility
  measure of $\Pker$ \iff\ any set $\Asf \in \borelSet(\rset^d)$
  satisfying $\mu(\Asf) >0$ is accessible.  The Markov kernel $\Pker$ is said to be
  irreducible if it admits an accessible small set or equivalently an
  irreducibility measure (in \cite{meyn:tweedie:2009}, our notion of irreducibility  is referred to as $\phi$-irreducibility, where $\phi$ is an irreducibility measure; here irreducibility therefore means $\phi$-irreducibility). $\Pker$ is said to be a \Tkernel~is there exists a kernel $\Tker$ on $\rset^d \times \mcb(\rset^d)$ and a sequence of non-negative numbers $(a_i)_{i \in \nsets}$ satisfying $\sum_{i=1}^{\plusinfty} a_i =1$, such that
  \begin{enumerate*}[label=(\roman*)]
  \item for any $x \in \rset^d$, $\Tker(x, \rset^d) >0$;
  \item for any $\msa \in \mcb(\rset^d)$, $x \mapsto \Tker(x,\msa)$ is lower semi-continuous;
\item for any $x \in \rset^d$, $\msa \in \mcb(\rset^d)$, $\sum_{i=1}^{\plusinfty} a_i \Pker^i(x,\msa) \geq \Tker(x,\msa)$.
  \end{enumerate*}
  $\Tker$ is referred to as a continuous component of $\Pker$.

  Let $(X_n)_{n \geq 0}$ be the canonical chain associated with $\Pker$
  defined on the canonical space $(\Omega,\mathcal{F},(\mathbb{P}_x, x \in \rset^d))$. A
  set $\Asf \in \borelSet(\rset^d)$ is said to be recurrent if for all $x \in \msa$, $\PE_x[N_\msa]= \plusinfty$ where $N_\msa = \sum_{i=0}^{\plusinfty} \1_{\msa}(X_i)$ is the number of visits to $\msa$. The set $\msa$ is Harris recurrent  if for any $x \in \msa$, $\mathbb{P}_x(N_\msa = \plusinfty) = 1$. The Markov kernel $\Pker$ is said to be Harris
  recurrent if all accessible sets are Harris recurrent. In this case, for all $x \in \rset^d$, and all accessible sets $\msa$, $\PP_x(N_\msa = \plusinfty)=1$.

  Define $\vartheta_1 : \rset_+ \to \rset_+$, for any $s \in \rset_+$ by
  \begin{equation}
    \label{eq:def_vartheta_1}
    \vartheta_1(s) = 1+  s/2 + s^2/4\eqsp.
  \end{equation}
\begin{theorem}
  \label{theo:irred_harris}
Assume \Cref{assum:regOne}($\expozero$) for some $\expozero \in \ccint{0,1}$ and that $\F$ is twice continuously
differentiable. Then, for all $T \in \nsets$, and  $h > 0$ satisfying
\begin{equation}
\label{eq:condition-h,T-harris}
 \left[ \{1 + h\constzero^{1/2} \vartheta_1(h\constzero^{1/2}) \}^T - 1 \right] < 1 \eqsp,
\end{equation}
and $q \in \rset^d$, there exists a $C^1(\rset^d,\rset^d)$-diffeomorphism $\tilde{q} \mapsto \Phiverletqi[h][T](q,\tilde{q})$ such that for any $p \in \rset^d$,
\begin{equation}
\label{theo:irred_harris_a}
\text{if $q_T =   \Phiverletq[h][T](q,p)$, defined by \eqref{eq:def_Phiverletq}, then $p = \Phiverletqi[h][T](q,q_T)$} \eqsp.
\end{equation}
Moreover,
\begin{enumerate}[label=(\roman*), wide, labelwidth=!, labelindent=0pt]
\item   \label{theo:irred_harris_b}
The Markov kernel $\Pkerhmc[h][T]$, is a \Tkernel; more precisely, for any $\eventB \in \mcb(\rset^d)$,
\begin{align}
\label{eq:def_kernel_hmc_false_density}
&\Pkerhmc[h][T](q, \eventB) =  \Tker_{h,T}(q,\eventB) \\
&\qquad + \updelta_{q}(\eventB)(2 \uppi)^{-d/2} \int_{\rset^d}  \parentheseDeux{1-\alphaacc\defEns{(q,\tilde{p}),\Phiverlet[h][T](q,\tilde{p})}} \rme^{-\norm{\tilde{p}}^2/2} \rmd \tilde{p} \eqsp,
\end{align}
where the kernel $ \Tker_{h,T}$ is a continuous component of $\Pkerhmc[h][T]$  and is given by
\begin{equation}
  \label{eq:def_tker}
\Tker_{h,T}(q,\eventB)
  =   (2 \uppi)^{-d/2} \int_{\eventB}    \balphaacc(q,\bar{q})\rme^{-\norm{\Phiverletqi[h][T](q,\bar{q})}^2/2} \detj_{\Phiverletqi[h][T](q,\cdot)}(\bar{q})  \rmd \bar{q} \eqsp,
\end{equation}
setting for $q,\tilde{q} \in \rset^d$, $\balphaacc(q,\bar{q}) =  \alphaacc\defEns{(q,\Phiverletqi[h][T](q,\bar{q})),\Phiverlet[h][T](q,\Phiverletqi[h][T](q,\bar{q}))}$ and  $\detj_{\Phiverletqi[h][T](q,\cdot)}(\tilde{q}) = \absLigne{\det(\Jac_{\Phiverletqi[h][T](q,\cdot)}(\tilde{q}))}$.
\item \label{theo:irred_harris_c} The Markov kernel $\Pkerhmc[h][T]$ is irreducible and the Lebesgue measure is an irreducibility measure. Moreover,  $\Pkerhmc[h][T]$ is aperiodic, Harris recurrent and all the compact sets are $1$-small. Therefore, for all $\q \in \rset^d$,
\begin{equation}
\label{eq:harris-theorem}
\lim_{n \to \plusinfty}    \tvnorm{\delta_\q \Pkerhmc[h][T]^n - \pi} = 0 \eqsp.
\end{equation}
\end{enumerate}
\end{theorem}

\begin{proof}
The proof is postponed to \Cref{sec:proof-crefth-harris_0}.
\end{proof}
For all $h > 0$ and $T \in \nsets$, we have
$\{1 + h \constzero^{1/2} \vartheta_1(h \constzero^{1/2} ) \}^T -1 \leq \rme^{h \constzero^{1/2} T \vartheta_1(h \constzero^{1/2} T)} -1$.
using that  $\vartheta_1$ is nondecreasing.
Then, setting $\bar{S} = c \constzero^{-1/2}$ where $c$ is   the unique positive root  of the equation
$c \vartheta_1(c)  = \log(2)$,  all $T \in \nsets$ and $h  \in \ooint{0,\bar{S}/T}$ satisfy \eqref{eq:condition-h,T-harris}\footnote{Note that conversely, if $h >0$ and $T \in \nsets$ satisfies \eqref{eq:condition-h,T-harris}, necessarily $h \in \oointLigne{0,\constzero^{-1/2}}$ because for any $s > 0$, $\vartheta_1(s) \geq 1$. In addition, since $\rme^{\log(2) s} \leq (1+s)$ for all $s \in \oointLigne{0,1}$, $T$ and $h$ satisfy $hT \leq \tilde{S}= \constzero^{-1/2}$.
}.



In our next result, we relax the second order differentiability
condition on $\F$, and in the case $\beta <1$ we even allow for
arbitrary large values of the step size $h$ and the number of iterations $T$.
The result is less quantitative and the proof
is more involved: we use degree theory for continuous
mapping (the main notions  required in the proof are recalled in  \Cref{sec:defin-usef-results}).
\begin{theorem}\label{theo:irred_D}
Let $h > 0$ and $T \in \nsets$ and assume either
\begin{enumerate}[label=(\alph*)]
\item
\label{theo:irred_D_a}
\Cref{assum:regOne} $(\expozero)$ for some $\expozero \in \coint{0,1}$,
\item
\label{theo:irred_D_b}
\Cref{assum:regOne} $(1)$ and that  $T \in \nsets$ and $h > 0$ satisfy \eqref{eq:condition-h,T-harris}.
\end{enumerate}
Then,
\begin{enumerate}[label=(\roman*)]
\item the HMC kernel $\Pkerhmc[h][T]$ defined by \eqref{eq:def_kernel_hmc} is irreducible, aperiodic, the Lebesgue measure is an
  irreducibility measure and any compact set of $\rset^d$ is  small.
\item $\Pkerhmc[h][T]$ is recurrent and for $\pi$-almost every $q \in \rset^d$,
$\lim_{n \to \plusinfty}    \tvnorm{\delta_q \Pkerhmc[h][T]^n - \pi} = 0$.
\end{enumerate}
\end{theorem}

\begin{proof}
The proof is postponed to \Cref{sec:proof-crefth_irred_D}.
\end{proof}

To the best of the author's knowledge, the first results regarding the
irreducibility of the HMC algorithm are established in
\cite{cances:legoll:stoltz} under the assumption that $\U$ and
$\norm{\nabla \U}$ are bounded above. Note that these assumptions are
in general satisfied only for compact state space. Irreducibility has
also been tackled in
\cite{livingstone:betancourt:byrne:girolami:2016}: in this work
however, the number of leapfrog steps $T$ is assumed to be random and
independent of the current position and momentum. Under this setting
and additional conditions which in particular imply that the number of
leapfrog steps $T$ is equal to $1$ with positive probability,
\cite{livingstone:betancourt:byrne:girolami:2016} shows that the
kernel associated with the HMC algorithm is irreducible. Under this condition,
the proof is  a direct
consequence of the irreducibility of the MALA algorithm - a mixture of
Markov kernels is irreducible as soon as one component of the mixture
is irreducible; the irreducibility of MALA kernel has been established
in \cite{roberts:tweedie:1996}). Finally, \cite[Proposition
3.7]{bou:sanz:2017} shows that RHMC is irreducible under the condition
that $U$ is at least quadratic.  Note that \Cref{theo:irred_D}
establishes irreducibility of HMC of sub-quadratic potential. However,
leap-frog integrator is not numerically stable for lighter than
Gaussian target density, therefore other kind of integrators should be
used instead, see \eg~\cite[Chapter VI]{hairer:wanner:lubish:2002}. 

Note that if $\expozero < 1$, then there is no condition in
\Cref{theo:irred_D} on the step-size for HMC to be ergodic. This
conclusion may at first glance be surprising since if $\pi$ is a
$d$-dimensional Gaussian distribution with covariance matrix $\Sigma$,
then the step-size $h$ has to be chosen smaller than
$2/\sqrt{\lambda_{\mathrm{max}}}$, where $\lambda_{\mathrm{max}}$ is
the largest eigenvalues of $\Sigma$, which is also the Lipschitz
constant of the gradient of the associated potential. If a larger
step-size $h$ is used, the leapfrog integrator is unstable, see
\eg~\cite[Example 3.4, Proposition 3.1]{bou-rabee:sanz-serna:2018},
meaning that the iterates of the algorithm diverge. But the Gaussian
distribution satisfies \Cref{assum:regOne}$(\expozero)$ for
$\expozero=1$ strictly.
We illustrate on a numerical example that under \Cref{assum:regOne}$(\expozero)$, for $\expozero <1$,
  the \textit{unadjusted} HMC proposal is in fact numerically stable
  and the HMC algorithm does converge for a step-size $h > 2 /
  \sqrt{\constzero}$, where $\constzero$ is the Lipschitz constant of $\nabla U$.
  In this example, we consider the potential $U : \rset \to \rset$ given
  for all $x \in \rset$ by $U(x) = 2 \defEnsLigne{1+\abs{x}^2}^{3/4}$. Then
  $U'(x) = 3(\abs{x}^2 +1)^{-1/4}x$ and is Lipschitz with constant
  $\constzero = 3$. We then run the unadjusted/adjusted HMC algorithm for a
  step-size $h =1.5 > 2/\sqrt{\constzero} \approx 1.15$ and a number of
  leapfrog-step $T =2$.  We can observe in
  \Cref{fig:experiments_convergence} the convergence of the HMC
  algorithm for the test function $f : q \mapsto \abs{q}^2$.
  \Cref{fig:experiments_stability} illustrates that the
  adjusted/unadjusted HMC are numerically stable even if $h =1.5 > 2/\sqrt{\constzero} \approx 1.15$, since the gradient
  is sub-linear.

\begin{figure}[h]
	\begin{center}		
		\includegraphics[width=12.2cm]{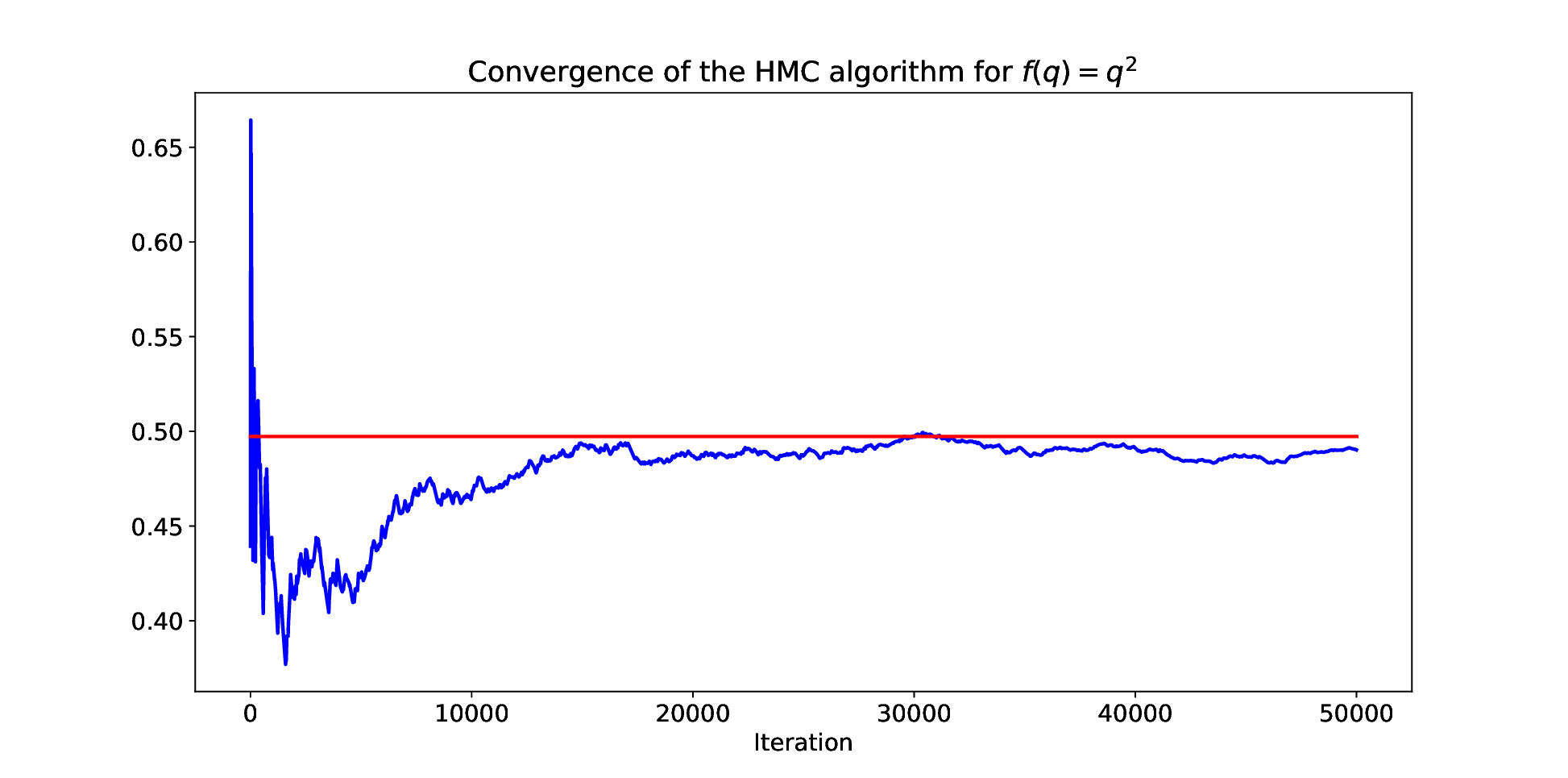}
\end{center}
	\caption{Convergence of the HMC algorithm for $U(x) = 2
  \defEnsLigne{1+\abs{x}^2}^{3/4}$, $h =1.5 > 2/\sqrt{\constzero}$ and $T=2$. The test function is $f : q \mapsto \abs{q}^2$. The red line indicates the real value of $\int_{\rset} f(q) \rmd \pi(q)$ estimated by numerical integration}
	\label{fig:experiments_convergence}
\end{figure}

\begin{figure}[h]
	\begin{center}
	\begin{tabular}{p{0.1cm}cp{0.1cm}c}
&		\includegraphics[width=5.8cm]{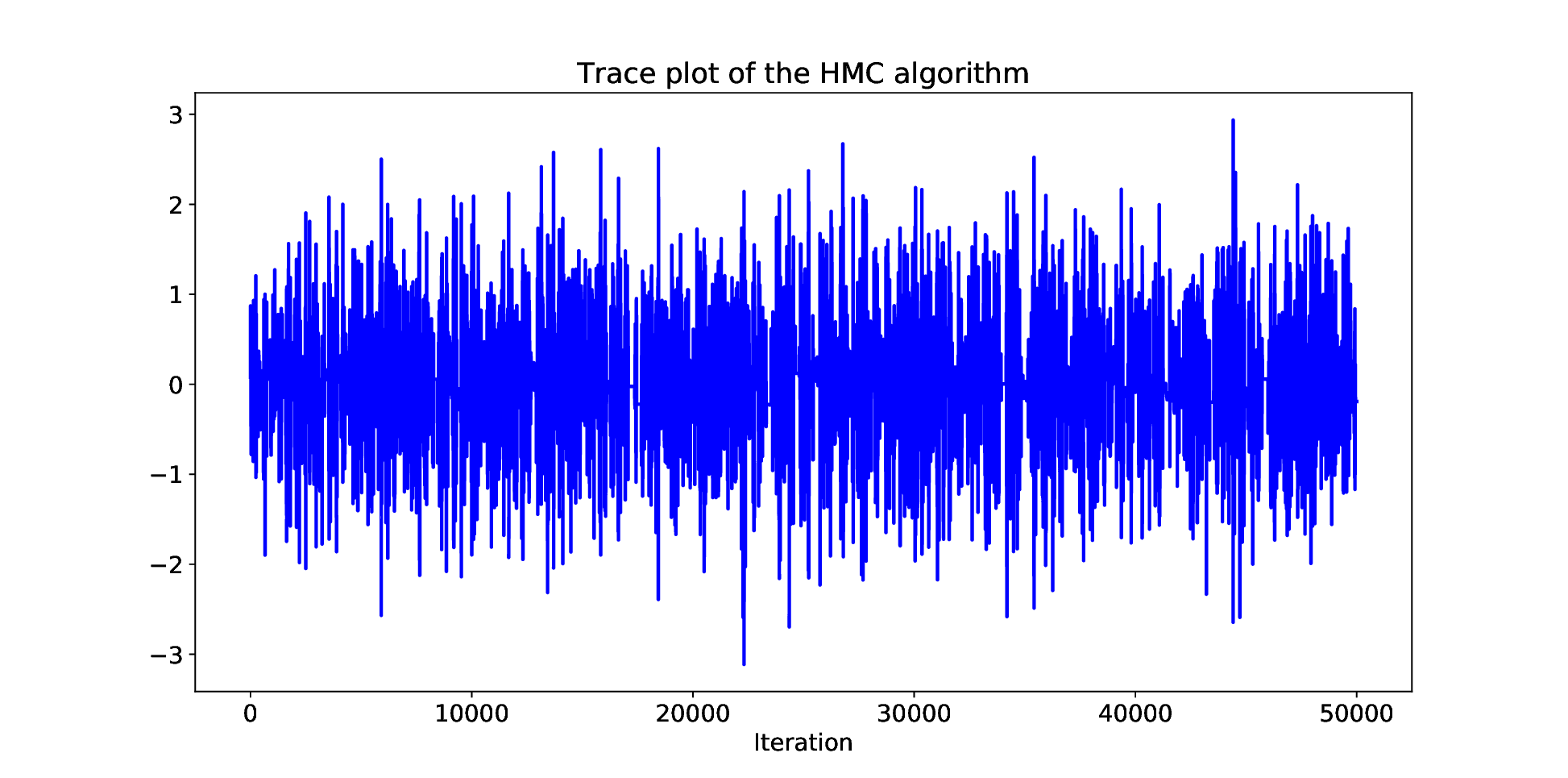}
		& &
		\includegraphics[width=5.8cm]{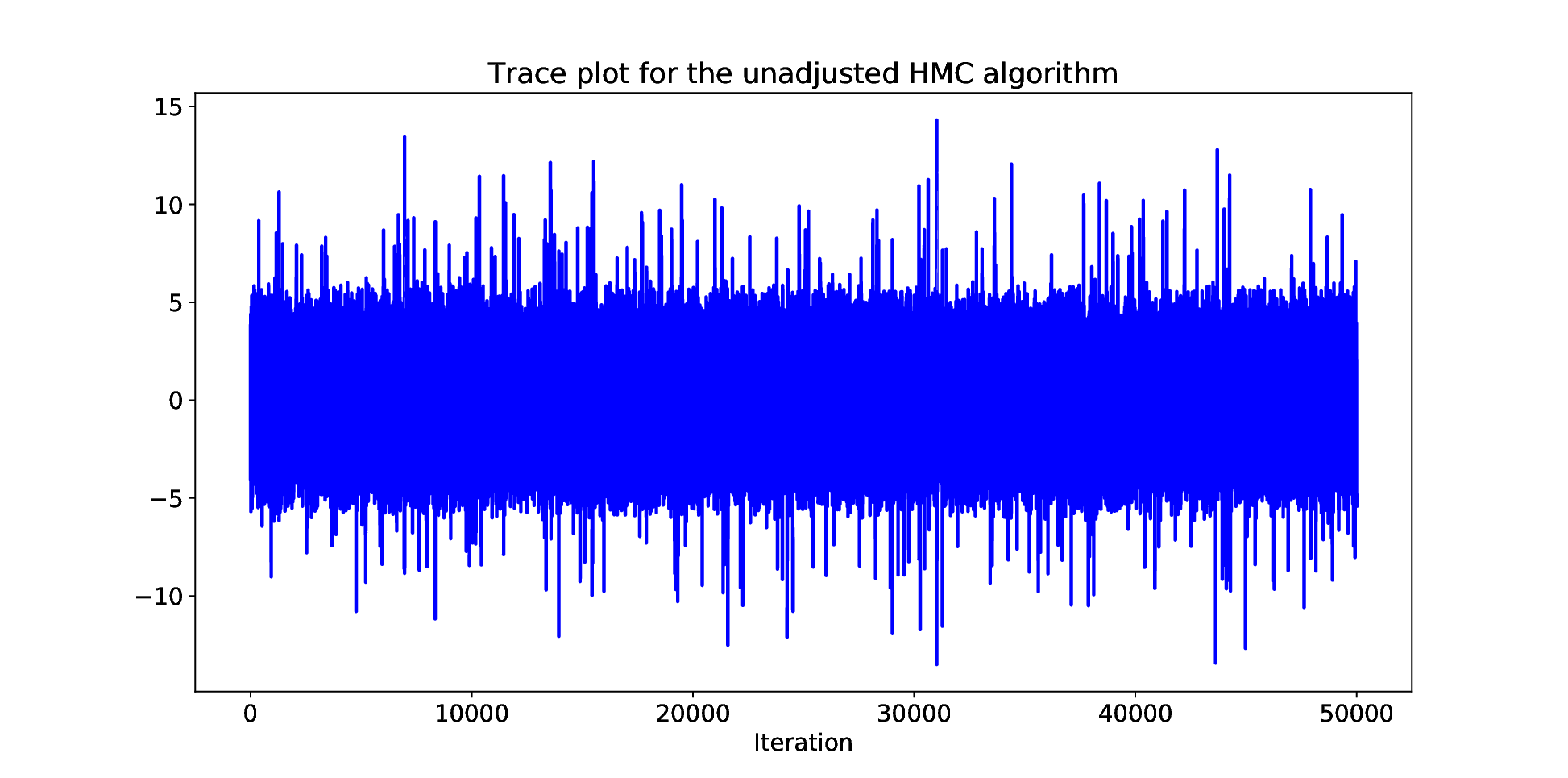}\\
& (a) & & (b)
	\end{tabular}
\end{center}
	\caption{Trace plots for the adjusted (a) / unadjusted (b) HMC algorithm for $U(x) = 2
  \defEnsLigne{1+\abs{x}^2}^{3/4}$, $h =1.5 > 2/\sqrt{\constzero}$ and $T=2$.}
	\label{fig:experiments_stability}
\end{figure}

Finally, note that our results can be easily extended to the case
where the number of steps is random. We briefly describe the main arguments to obtain such
extension.
Let $(\varpi_i)_{i\in \nset^*}$ be a probability distribution on $\nset^*$ and $(h_i)_{i \in \nset^*}$ be a sequence of  positive real
numbers.  Define the randomized Hamiltonian kernel
$\randomkerhmc_{\mathbf{h},\bfvarpi}$ on $(\rset^d, \borelSet(\rset^d))$ associated with $(\varpi_i)_{i\in \nset^*}$ and
$(h_i)_{i \in \nset^*}$ by
\begin{equation}
\label{eq:randomhmc}
\randomkerhmc_{\mathbf{h},\bfvarpi} = \sum_{i\in \nset^*} \varpi_i \Pkerhmc[h_i][i] \eqsp.
\end{equation}
We denote by $\supp(\bfvarpi)= \set{i \in \nset^*}{\omega_i \ne 0}$ the support of the distribution $\bfvarpi$.
\begin{corollary}
  \label{coro:ergod-hmc-algor}
Let $\expozero \in \ccint{0,1}$ and assume \Cref{assum:regOne}($\expozero$).
Let $(\varpi_i)_{i\in \nset^*}$ be a probability distribution on $\nset^*$, $(h_i)_{i \in \nset^*}$ be a sequence of  positive real
numbers, and $\randomkerhmc_{\mathbf{h},\bfvarpi}$ be the randomized Hamiltonian kernel associated with $(\varpi_i)_{i\in \nset^*}$ and
$(h_i)_{i \in \nset^*}$.
\begin{enumerate}[label=(\alph*)]
\item
\label{coro:ergod-hmc-algor_a}
Assume that $\F$ is twice continuously and there exists $i \in \nset^*$ such that   $ [ \{1 + h_i\constzero^{1/2} \vartheta_1( h_i \constzero^{1/2}) \}^i - 1 ] < 1$ and  $\varpi_i > 0$   where $\vartheta_1$ is given by  \eqref{eq:def_vartheta_1}. Then the conclusions  of \Cref{theo:irred_harris}-\ref{theo:irred_harris_c} hold for $\randomkerhmc_{\mathbf{h},\bfvarpi}$.
\item
\label{coro:ergod-hmc-algor_b}
 If $\expozero \in \coint{0,1}$, then the conclusions of \Cref{theo:irred_D}-\ref{theo:irred_D_a} hold for $\randomkerhmc_{\mathbf{h},\bfvarpi}$.
\item
\label{coro:ergod-hmc-algor_c}
 If $\expozero = 1$ and there exists $i \in \supp(\bfvarpi)$ such that  $ [ \{1 + h_i \constzero^{1/2} \vartheta_1(h_i \constzero^{1/2}) \}^i - 1 ] < 1$, then the conclusions of \Cref{theo:irred_D}-\ref{theo:irred_D_b} hold for $\randomkerhmc_{\mathbf{h},\bfvarpi}$.
\end{enumerate}
\end{corollary}

\begin{proof}
  \ref{coro:ergod-hmc-algor_a} follows from \Cref{theo:irred_harris} and \Cref{propo:harris_rec}. \ref{coro:ergod-hmc-algor_b} and \ref{coro:ergod-hmc-algor_c} are straightforward applications of \Cref{theo:irred_D}.
\end{proof}






\section{Geometric ergodicity of HMC}
\label{sec:geom-ergod-hmc}

In this section, we give conditions on the potential $\F$ which imply that
the HMC kernel \eqref{eq:def_Phiverletq} converges geometrically fast to its invariant distribution.
Let $V: \rset^d \to \coint{1,\plusinfty}$ be a measurable function and $\Pker$
be a Markov kernel on $(\rset^d,\borelSet(\rset^d))$. The Markov kernel $\Pker$ is said to
be $V$-uniformly geometrically ergodic if $\Pker$ admits an invariant probability $\pi$
and there exists $\rho \in \coint{0,1}$
and $\varsigma \geq 0$ such that for all $\q \in \rset^d$ and $k \in \nset^*$,
\begin{equation}
  \Vnorm[V]{\Pker^k(\q,\cdot)-\pi} \leq \varsigma \rho^{k} V(\q) \eqsp.
\end{equation}
By \cite[Theorem 16.0.1]{meyn:tweedie:2009}, if $\Pker$ is aperiodic, irreducible and satisfies a Foster-Lyapunov drift condition, \ie~there exists a small set $\Csf$ for $\Pker$, $\lambda \in \coint{0,1}$ and $b < \plusinfty$ such that for all $\q \in \rset^d$,
\begin{equation}
\label{eq:foster-lyapunov}
\Pker V  \leq \lambda V + b \1_{\Csf} \eqsp,
\end{equation}
then $\Pker$ is $V$-uniformly geometrically ergodic. If a function $V : \rset^d \to \coint{1,\infty}$
satisfies \eqref{eq:foster-lyapunov}, then $V$ is said to be a Foster-Lyapunov function for $\Pker$.
We first give an elementary condition to establish the $V$-uniform geometric
ergodicity for a class of generalized Metropolis-Hastings  kernels which includes HMC kernels as a particular example.

Let $\Kker$ be a proposal kernel on $(\rset^d, \borelSet(\rset^{2d}))$ and $\alphagen : \rset^{3 d } \to
\ccint{0,1}$ be an acceptance probability, assumed to be Borel measurable. Consider the Markov kernel $\Pker$ on $(\rset^d,\borelSet(\rset^d))$ defined for all $\q \in \rset^d$ and
$\eventA \in \borelSet(\rset^d)$ by
\begin{equation}
\label{eq:def_kenel_MH}
  \Pker(\q,\eventA) = \int_{\rset^{2d}} \1_{\eventA}(\projq(z)) \alphagen(\q,z) \Kker(\q, \rmd z )
+ \updelta_{\q}(\eventA) \int_{\rset^{2d}} \defEns{1- \alphagen(\q,z) }\Kker(\q, \rmd z)  \eqsp,
\end{equation}
where $\projq : \rset^{d} \times \rset^d \to \rset^d$ is the canonical projection onto the first
$d$ components.
For $h \in \rset^*_+$ and $T \in \nset^*$, $\Pkerhmc[h][T]$ corresponds to $\Pker$ with
$\Kker$ and $\alphagen$  given for all $\q,\p,\x \in \rset^d$ and $\Bsf \in
\borelSet(\rset^{2d})$ respectively by
\begin{align}
\label{eq:def_Pker_proposition_double}
  \PkerhmcD[h][T](\q,\Bsf) &= (2\uppi)^{-d/2}\int_{\rset^{d}} \1_{\Bsf}\parenthese{\Phiverletq[h][T](\q,\tilde{\p}),\tilde{\p}} \rme^{-\norm{\tilde{\p}}^2/2} \rmd \tilde{\p} \eqsp, \\
\label{eq:def_alpha_acc_tilde_hmc}
\tildeAlphaacc(\q,(\tilde{\q},\tilde{\p})) & =
\begin{cases}
\alphaacc\defEns{(\q,\tilde{p}),\Phiverlet[h][T](\q,\tilde{p})} \eqsp, & \text{if}\, \tilde{q}= \Phiverletq[h][T](\q,\tilde{\p}) \eqsp, \\
0 & \text{otherwise} \eqsp,
\end{cases}
\end{align}
where  $\Phiverlet[h][T]$, $\Phiverletq[h][T]$ and $\alphaacc$  are  defined in   \eqref{eq:def_Phiverlet}, \eqref{eq:def_Phiverletq} and \eqref{eq:def_acc_ratio}, respectively. Let $\Vgeo : \rset^d \to \coint{1,\plusinfty}$ be a
\emph{norm-like}  function,  \ie\ a measurable function such that for all $M \in \rset_+$, the level sets $\set{\q \in \rset^d}{\Vgeo(\q) \leq M}$ are compact. Note that if $\Vgeo$ is norm-like, for any $M \in \rset_+$, $\set{\q \in \rset^d}{\Vgeo(\q) \leq M}^{\complementary}$ is non-empty.   The function $\Vgeo$ naturally extends on $\rset^{2d}$ by
setting for all $(\q,\p) \in \rset^{2d}$, $\Vgeo(\q,\p) = \Vgeo(\q)$.
For all $\q \in \rset^d$, define:
\begin{equation}
\label{eq:def_rej_ballV}
  \rejectregion(\q) = \defEns{z \in \rset^{2d} \, , \, \alphagen(\q,z) < 1  } \eqsp, \,
    \ballV(\q) = \defEns{z \in \rset^{2d} \, , \, \Vgeo(\projq(z)) \leq \Vgeo(\q) } \eqsp.
\end{equation}
The set $\rejectregion(\q)$ is the potential rejection region.
Our next result gives a condition on $\Kker$ and $\alphagen$ which
implies that if $V$ is a Foster-Lyapunov function for $\Kker$ then
$\Pker$ satisfies a Foster-Lyapunov drift condition as well. This
result is inspired by \cite[Theorem~4.1]{roberts:tweedie:1996}, which is used to show the $V$-uniform geometric ergodicity of the MALA algorithm.
\begin{proposition}
\label{propo:geo_drift_MH}
  Let $\Vgeo : \rset^d \to \coint{1,\plusinfty}$ be a norm-like  function.
  Assume moreover that there exist  $\lambdageo \in \coint{0,1}$ and $\bgeo \in \rset_+$ such that
  \begin{equation}
  \label{eq:assum:geo_ergo_1}
  \Kker \Vgeo \leq  \lambdageo \Vgeo + \bgeo \eqsp.
  \end{equation}
and
\begin{equation}
\label{eq:assum:geo_ergo_2}
  \lim_{M \to \plusinfty} \sup_{\set{\q \in \rset^d}{\Vgeo(\q) \geq M}} \Kker(\q,\rejectregion(\q) \cap \ballV(\q)) = 0  \eqsp.
\end{equation}
 Then there exist  $\lambdageotilde \in \coint{0,1}$ and $\bgeotilde \in \rset_+$ such that
 $\Pker \Vgeo \leq  \lambdageotilde \Vgeo + \bgeotilde$ where $\Pker$ is given by \eqref{eq:def_kenel_MH}.
\end{proposition}
\begin{proof}
The proof is postponed to \Cref{sec:proof-crefpr}.
\end{proof}
We show below that under appropriate conditions, the proposal kernel $\PkerhmcD[h][T]$ and
the acceptance probability $\tildeAlphaacc$ given by \eqref{eq:def_Pker_proposition_double} and
\eqref{eq:def_alpha_acc_tilde_hmc} satisfy the conditions of
\Cref{propo:geo_drift_MH} which imply that the HMC kernel
$\Pkerhmc[h][T]$ is $V$-uniformly geometrically ergodic. 
For $\m \in \ocint{1,2}$, consider the following assumption:

\begin{assumption}[$m$]
  \label{assum:potential:c}
There exist $\constthree \in \rset^*_+$ and $\constfour \in \rset$ such that for all $\q \in \rset^d$,
  \begin{equation}
    \ps{\nabla \F(\q)}{\q} \geq \constthree \norm{\q}^{m} -\constfour \eqsp.
  \end{equation}
\end{assumption}
For all $\a \in \rset_+^*$ and $\q \in \rset^d$, define
\begin{equation}
\label{eq:def_Va}
\Vdrifta[a] (\q) = \exp(\a \norm{\q}) \eqsp.
\end{equation}
\begin{proposition}
\label{lem:drift_uhmc}
\begin{enumerate}[label=(\alph*)]
\item   \label{lem:drift_uhmc_1}
 Assume   \Cref{assum:regOne}$(m-1)$ and  \Cref{assum:potential:c}$(\m)$ for some $\m \in \ooint{1,2}$. Then, for all $T \in \nsets$,  $h \in \rset^*_+$, and $\a \in \rsetep$, there exist $\lambda \in \coint{0,1}$ and $\b \in \rsetp$ such that
  \begin{equation}
    \label{eq:drift_lem}
      \PkerhmcD[h][T] \Vdrifta[\a] \leq \lambda  \Vdrifta[\a] + \b \eqsp.
  \end{equation}
\item
\label{lem:drift_uhmc_2}
 Assume   \Cref{assum:regOne}$ (1)$ and  \Cref{assum:potential:c}$ (2)$.  Let $\bar{S} > 0$ be such that $\Theta(S) < \constthree$ for any $S \in \ocint{0,\bar{S}}$, where
\begin{align}
\label{eq:definition-function-C}
  \Theta(s)&= 2 \constzero^{1/2} \vartheta_2(s) \{ \rme^{\constzero^{1/2} s \vartheta_1(\constzero^{1/2} s)} - 1\} \\
  & \qquad  \qquad + 6 s^2 \left(   \constzeroT ^2  +  \constzero \vartheta_2^2(s) \{ \rme^{\constzero^{1/2} s \vartheta_1(\constzero^{1/2} s)} - 1\}^2\right) \eqsp.
\end{align}
Then, for all $a \in \rsetep$,  $T \in \nsets$ and  $h \in \ocint{0,\bar{S}/T}$,  there exist  $\lambda \in \coint{0,1}$ and $\b \in \rsetp$ which satisfy \eqref{eq:drift_lem}.
\end{enumerate}
\end{proposition}
\begin{proof}
  The proof is postponed to \Cref{sec:proof-crefl-2}
\end{proof}
We now derive sufficient conditions under which the condition \eqref{eq:assum:geo_ergo_2} of
\Cref{propo:geo_drift_MH} is satisfied.

\begin{assumption}[$\m$]
\label{assum:potential}
\begin{enumerate}[label = (\roman*)]
\item \label{assum:potential:a}
$\F \in C^3(\rset^d)$  and there exists $\constone \in \rset_+^*$ such that for all $\q \in \rset^d$ and $k=2,3$:
\begin{equation}
\norm{D^k \F(\q)}\leq \constone \defEns{1+\norm{\q}}^{\m-k} \eqsp.
\end{equation}
\item \label{assum:potential:b}
There exist $\consttwo \in \rset_+^* $ and $\rhtwo \in \rset^+$ such that for all $\q \in \rset^d$, $\norm{q}\geq \rhtwo$,
\begin{equation}
D^2\F(\q)\defEns{ \nabla \F(\q)\otimes  \nabla \F(\q)}  \geq \consttwo \norm{\q}^{3\m-4} \eqsp.
\end{equation}
  \end{enumerate}
\end{assumption}

It is easily checked that under \Cref{assum:potential}, the results of \Cref{sec:ergodicity-hmc} can be applied, \ie~$\nabla \F$ satisfies \Cref{assum:regOne}($\m-1$); see \Cref{lem:grad_Lip_F}.

Condition \Cref{assum:potential:c}$(m)$ and \Cref{assum:potential}$(m)$ are satisfied by power functions $\q \mapsto c\norm{\q}^\m$. More generally, they are satisfied by $\m$-homogeneously quasiconvex functions with convex level sets  outside a ball and by  perturbations of such functions.

We say that a function $\F_0$ is $m$-homogeneous quasi-convex
outside a ball of radius $\Rexp$ if the following conditions are satisfied:
\begin{enumerate}[(QC-1)]
\item for all $t \geq 1$ and $q \in \rset^d$, $\norm{\q}\geq \Rexp$, $\F_0(t \q)= t^\m \F_0(\q)$.
\item for all $\q \in \rset^d$, $\norm{\q} \geq \Rexp$, the level sets $\{ \x\, :\, \F_0(\x) \leq \F_0(\q)\}$ are convex.
\end{enumerate}
\begin{proposition}
\label{le:convex}
Let $m \in \ccint{1,2}$  and $\Rexp \in \rset_+$.  Assume that the potential $\F$ may be decomposed as
$$
\F(\q)=\F_0(\q)+G(\q) \eqsp, \quad \text{$\q\in \rset^d$, $\norm{\q} \geq \Rexp$} \eqsp,
$$
where the functions $\F_0,G \in C^3(\rset^d)$ satisfy the following two conditions:
  \begin{enumerate}[(A)]
  \item $\F_0$ is $\m$-homogeneously quasiconvex outside a ball of radius $\Rexp$ and $\lim_{\norm{\q} \to \plusinfty} \F_0(\q)=\infty$.
\label{le:convex:a}
\item
\label{le:convex:b}
For $k=2,3$, $\lim_{\norm{\q} \to \plusinfty}\normop{D^k G(\q)}/  \norm{\q}^{\m-k}= 0$.
  \end{enumerate}
Then $\F$ satisfies  \Cref{assum:potential:c}$(m)$ and  \Cref{assum:potential}$(\m)$.
\end{proposition}
\begin{proof}
The proof is postponed to \Cref{sec:proof-crefle:convex}.
\end{proof}

To show that the condition \eqref{eq:assum:geo_ergo_2} of
\Cref{propo:geo_drift_MH} is satisfied under
\Cref{assum:potential}$(m)$, we rely on the following important result which implies that the probability of accepting a move goes to 1 as $\norm{q} \to \infty$.
\begin{proposition}
  \label{propo:accept} Assume
  \Cref{assum:potential}$(m)$ for some $\m \in \ocint{1,2}$. Let $\gamma \in \ooint{0,\m-1}$.
  \begin{enumerate}[label=(\alph*)]
  \item
  \label{propo:accept_1}
  If $\m\in (1,2)$, for all $T \in \nsets$, $h \in \rset_+^*$, there exists $R_{\Ham} \in \rset_+$ such that for
  all $\q_0,\p_0 \in \rset^d$, $\norm{q_0} \geq R_{\Ham}$ and
  $\norm{p_0} \leq \norm{\q_0}^{\gamma}$, $ \Ham(\Phiverlet[h][T](q_0,p_0)) -
  \Ham(q_0,p_0) \leq 0$.
\item
  \label{propo:accept_2}
  If $\m=2$,   there exists $\bar{S} >0$ such that for any $T \in \nsets$ and $h \in \ocint{0, \bar{S}/T^{3/2}}$,  there exists $R_{\Ham} \in \rset_+$ satisfying for all $\q_0,\p_0 \in \rset^d$, $\norm{q_0} \geq R_{\Ham}$ and
  $\norm{p_0} \leq \norm{\q_0}^{\gamma}$, $ \Ham(\Phiverlet[h][T](q_0,p_0)) -
  \Ham(q_0,p_0) \leq 0$.
  \end{enumerate}
\end{proposition}

\begin{proof}
  The proof is postponed to  \Cref{sec:proof-crefth}.
\end{proof}

This result  means that far in the tail the HMC proposal are "inward".
We illustrate the result of \Cref{propo:accept}-\ref{propo:accept_1}
in \Cref{fig:H_behaviour} for $U$ given by $\q \mapsto
(\norm[2]{\q}+\delta)^{\kappa}$ for $\kappa=3/4$, $h = 0.9$ and
$p_0 \in \rset^d$, $\norm{p_0}=1$. Note that this potential satisfies
the condition of the proposition. We can observe that choosing the different
initial conditions $q_0$ with increasing norm imply that $\tilde{T} =
\max\{k \in \nset ; \Ham(\Phiverlet[h][k](q_0,p_0)) - \Ham(q_0,p_0)
<0\}$ increases as well.

 \begin{figure}[h]
   \centering
   \includegraphics[scale=0.3]{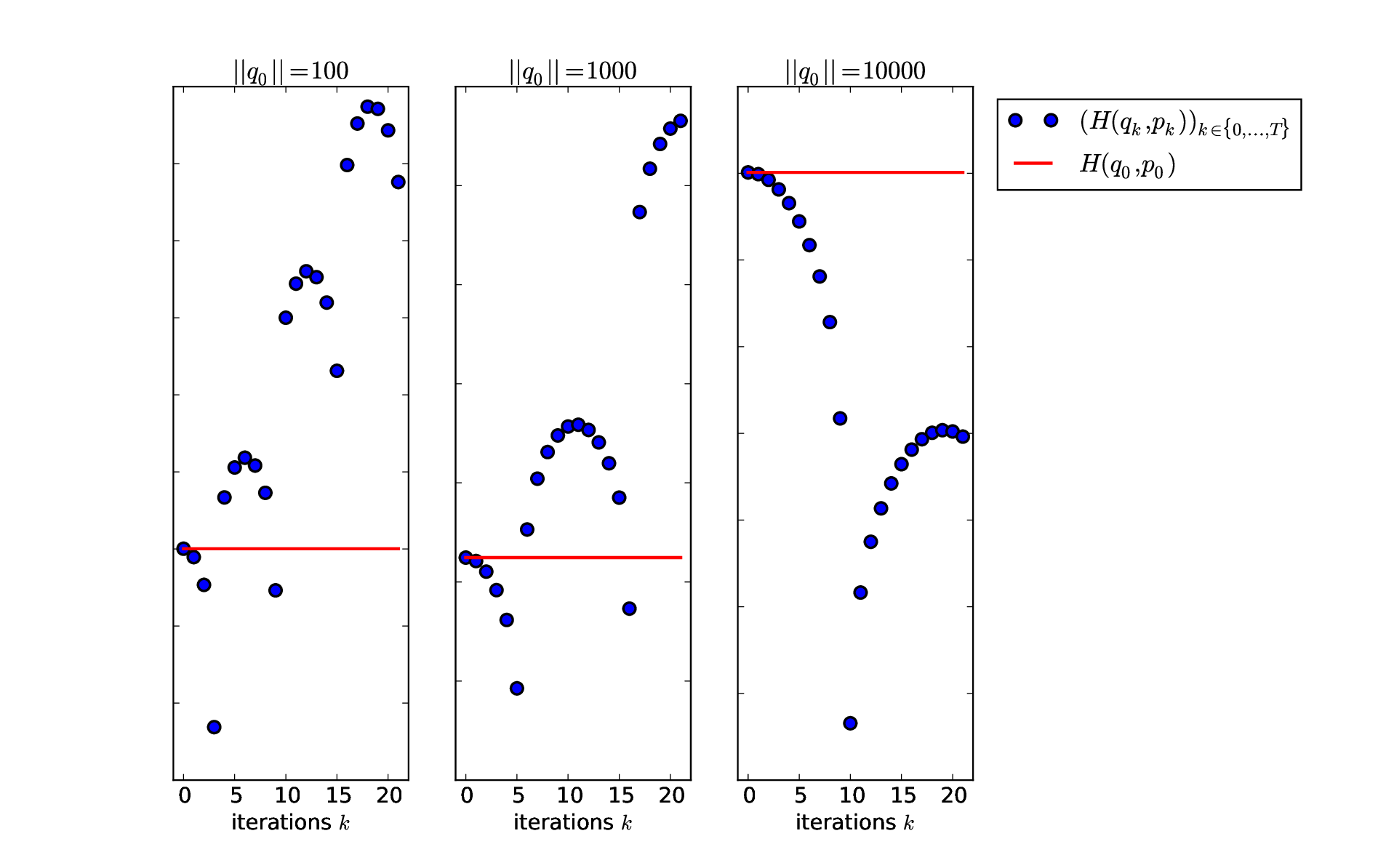}
   \caption{Behaviour of $(\Ham(\Phiverlet[h][k](q_0,p_0)))_{k \in \{0,\ldots,T\}}$ for different initial conditions $q_0$.}
   \label{fig:H_behaviour}
 \end{figure}

 However, in the case $m=2$, \Cref{propo:accept}-\ref{propo:accept_2} only implies that the HMC proposal is inward only if the step size $h$ is sufficiently small with respect to the number of leapfrog step $T$, \ie~is of order $\bigO(T^{-3/2})$. To relax this condition, we strengthen \Cref{assum:potential}($2$) by assuming that $U$ is a smooth perturbation of a quadratic function.
 \begin{assumption}
   \label{ass:pertub}
   There exist $\tilde{U} : \rset^d \to \rset$, continuously differentiable, and  a positive definite matrix $\Sigmabf$ such that
   $U(q) = \ps{\Sigmabf q}{q}/2 + \tilde{U}(q)$ and there exist $\constfive \geq 0$ and  $\varrho \in \coint{1,2}$  such that for any $q,x \in \rset^d$,
   \begin{align}
     \absLigne{\tilde{U}(q)} &\leq \constfive(1+\norm[\varrho]{q}) \eqsp, \quad  \normLigne{\nabla \tilde{U}(q)} \leq \constfive(1+\norm[\varrho-1]{q}) \eqsp,\\  \qquad \qquad \qquad & \normLigne{\nabla \tilde{U}(q) - \nabla \tilde{U}(x)} \leq \constfive \norm{q-x} \eqsp.
   \end{align}
 \end{assumption}
Note that it is straightforward to check that under \Cref{ass:pertub}, the conditions \Cref{assum:regOne}$(1)$ and  \Cref{assum:potential:c}$(2)$ hold.

 \begin{proposition}
  \label{propo:accept_pertub} Assume  \Cref{ass:pertub}  and let $\gamma \in \ooint{0,1}$.
There exists a  constant $\bar{S} >0$ such that for all $T \in \nsets$, $h \in \ocint{0,\bar{S}/T}$, there exists $R_{\Ham} \in \rset_+$ such that for
  all $\q_0,\p_0 \in \rset^d$, $\norm{q_0} \geq R_{\Ham}$ and
  $\norm{p_0} \leq \norm{\q_0}^{\gamma}$, $ \Ham(\Phiverlet[h][T](q_0,p_0)) -
  \Ham(q_0,p_0) \leq 0$.
\end{proposition}
\begin{proof}
  The proof is postponed to  \Cref{sec:proof-crefth_accept_2}.
\end{proof}
We now can establish the geometric ergodicity of the HMC sampler.
\begin{theorem}
  \label{theo:geoErg}
  \begin{enumerate}[label=(\alph*)]
  \item
  \label{item:theo_1}
 If   \Cref{assum:potential:c}$(\m)$ and  \Cref{assum:potential}$(m)$ hold for some $\m\in (1,2)$, then for all $a \in \rset_+^*$,  $T \in \nset^*$ and $h > 0$, the HMC kernel $\Pkerhmc[h][T]$ is $\Vdrifta[\a]$-uniformly geometrically ergodic, where $\Vdrifta[a]$ is defined by \eqref{eq:def_Va}.
\item
  \label{item:theo_2}
  If  \Cref{assum:potential:c}$(2)$ and \Cref{assum:potential}$(2)$ hold, then there exists $\bar{S}>0$ such that for all  $a \in \rset^*_+$, $T \in \nset^*$ and $h \in \ooint{0,\bar{S}/T^{3/2}}$, $\Pkerhmc[h][T]$ is $\Vdrifta[\a]$-uniformly geometrically ergodic.
\item \label{item:theo_3}
  If \Cref{ass:pertub} holds, then there exists $\bar{S}>0$ (depending only on $\Sigmabf$ and $\constfive$) such that for all $a \in \rset^*_+$, $T \in \nset^*$ and $h \in \ooint{0,\bar{S}/T}$, $\Pkerhmc[h][T]$ is $\Vdrifta[\a]$-uniformly geometrically ergodic.
  \end{enumerate}
\end{theorem}
\begin{proof}[Proof of \Cref{theo:geoErg}]
  It is enough to consider \ref{item:theo_1} as the proof
  of  \ref{item:theo_2} and \ref{item:theo_3} follows exactly the same lines
  taking $\bar{S}$ small enough.  \Cref{lem:drift_uhmc} shows that for all $T \in \nsets$,  $h \in \rset^*_+$, and $\a \in \rsetep$, there exist $\lambda \in \coint{0,1}$ and $\b \in \rsetp$ such that the Foster-Lyapunov drift condition
  $\PkerhmcD[h][T] \Vdrifta[\a] \leq \lambda  \Vdrifta[\a] + \b$ is satisfied.
  By \Cref{propo:accept}, there exists $R_{\Ham} \geq 0$ such that for all $\q \in \rset^d$, $\norm{\q} \geq R_{\Ham}$,
\begin{equation}
\label{eq:proofpgeo_erg_0}
  \int_{\rejectregion(q)}   \PkerhmcD[h][T](\q, \rmd z ) \leq (2\uppi)^{-d/2} \int_{\{ \norm{\p} \geq \norm{q}^{\gamma}\} } \rme^{-\norm{p}^2/2} \rmd p \eqsp,
\end{equation}
for $\gamma \in \ooint{0,m-1}$ where $\rejectregion(\q) = \set{z \in \rset^{2d}}{\tildeAlphaacc(\q,z) < 1 }$ (see~\eqref{eq:def_alpha_acc_tilde_hmc}), which implies that
\begin{equation}
\label{eq:proof:geo_erg:1}
  \lim_{M \to \plusinfty} \sup_{\norm{\q} \geq M} \int_{\rejectregion(q)} \PkerhmcD[h][T](\q, \rmd z ) = 0 \eqsp,
\end{equation}

Since $\Vdrifta[a]$ is norm-like,  \Cref{propo:geo_drift_MH} implies that  for all $T > 0$ and $h > 0$, there exists $\lambdageotilde$ and $\bgeotilde$ (depending upon $a$, $h$ and $T$) such that
$\Pkerhmc[h][T] \Vdrifta[a] \leq \lambdageotilde \Vdrifta[a] + \bgeotilde$.
For all $M \geq 0$ the level sets $\{ \Vdrifta[a] \leq M\}$ are compact and hence small by \Cref{theo:irred_D}.
\cite[Corollary~14.1.6]{douc:moulines:priouret:2018} then shows that there exists a small set $\msc$, $\check{\lambda} \in \coint{0,1}$ and $\check{b} \in \coint{0,1}$ such that $\Pkerhmc[h][T] \Vdrifta[a] \leq \check{\lambda} \Vdrifta[a] + \check{b} \1_{\msc}$.  Since $\Pkerhmc[h][T]$ is aperiodic, the result follows from \cite[Theorem~15.2.4]{douc:moulines:priouret:2018}.
\end{proof}

We  finally consider the case where the number of leapfrog steps is a random variable
independent of the current state.
\begin{theorem}
\label{coro:geo_ergod-hmc-algor}
\begin{enumerate}[label=(\alph*)]
  \item
  \label{item:geo_ergod-hmc-algor:theo_1}
 If   \Cref{assum:potential:c}$(\m)$ and  \Cref{assum:potential}$(m)$ hold for  $\m\in (1,2)$,  then for all probability distributions $\bfvarpi=(\omega_i)_{i \in \nset^*}$ on $\nset^*$, all sequences $\mathbf{h}= (h_i)_{i \in \nset*}$ of positive numbers,  and $a \in \rset^*_+$, the randomized kernel  $\randomkerhmc_{\mathbf{h},\bfvarpi}$ \eqref{eq:randomhmc} is $\Vdrifta[\a]$-uniformly geometrically ergodic, where $\Vdrifta[a]$ is defined by \eqref{eq:def_Va}.
\item
  \label{item:geo_ergod-hmc-algor:theo_2}
  If  \Cref{assum:potential:c}$(2)$ and \Cref{assum:potential}($2$) hold, then there exists $\bar{S}>0$ such that for all probability distributions $\bfvarpi=(\omega_i)_{i \in \nset^*}$ on $\nset^*$, all sequences $\mathbf{h}= (h_i)_{i \in \nset^*}$ satisfying $\max_{i \in \supp(\bfvarpi)} i^{3/2} h_i \leq \bar{S}$,  and $a \in \rset^*_+$, $\randomkerhmc_{\mathbf{h},\bfvarpi}$ is $\Vdrifta[\a]$-uniformly geometrically ergodic.
\item   \label{item:geo_ergod-hmc-algor:theo_3}
  If \Cref{ass:pertub} holds, then there exists $\bar{S}>0$ (depending only on $\Sigmabf$ and $\constfive$) such that for all  probability distributions $\bfvarpi=(\omega_i)_{i \in \nset^*}$ on $\nset^*$, all sequences $\mathbf{h}= (h_i)_{i \in \nset^*}$ satisfying $\max_{i\in \supp(\bfvarpi)} i h_i \leq \bar{S}$,  and $a \in \rset^*_+$, $\randomkerhmc_{\mathbf{h},\bfvarpi}$ is $\Vdrifta[\a]$-uniformly geometrically ergodic.
\end{enumerate}
\end{theorem}
\begin{proof}
It is enough to consider \ref{item:geo_ergod-hmc-algor:theo_1} as the proofs
of \ref{item:geo_ergod-hmc-algor:theo_2} and \ref{item:geo_ergod-hmc-algor:theo_3} are along the same lines.
Set $a \in \rset_+^*$. It is established in the proof of \Cref{theo:geoErg}  that  for all $i \in \nset^*$
$\Pkerhmc[i][h_i]$ satisfies a Foster-Lyapunov drift condition:
there exists $\check{\lambda}_i \in \coint{0,1}$ and $\check{b}_i < \infty$ such that
$\Pkerhmc[i][h_i] \Vdrifta[a] \leq \lambda_i \Vdrifta[a] + b_i$,
By \Cref{coro:ergod-hmc-algor},   $\randomkerhmc_{\mathbf{h},\bfvarpi}$ is irreducible and aperiodic and all the compact sets are small. We conclude by applying \cite[Theorem~15.2.4]{douc:moulines:priouret:2018}.
\end{proof}

Compared to \cite{livingstone:betancourt:byrne:girolami:2016}, which
establishes geometric ergodicity of the HMC kernel under an implicit
assumption on the behaviour of the acceptance rate, our conditions are
directly verifiable on the potential $U$.

On the other hand, our conditions are different than the one given by
\cite{bou:sanz:2017} to establish the geometric ergodicity of the
idealized randomized HMC, which assumed to exactly solve the Hamiltonian
ODE \eqref{eq:hamil_ode}. These conditions are the following
1)$\int_{\rset^d} \norm[2]{q} \rmd \pi(q) < \plusinfty$, 2) there
exist $C_1 \in \ooint{0,1}$ and $C_2 >0$ such that for all $q \in \rset^d$
\begin{equation}
\label{eq:hyp_boo_rabee_sanz_serna}
  (1/2) \ps{\nabla U(q)}{q} \geq C_1 U(q) + \frac{(\tau^{-1}C_1/4)^2+\tau^{-2}C_1(1-C_1)/4}{2(1-C_1)}\norm[2]{q} -C_2 \eqsp,
\end{equation}
where $\tau>0$ is the duration parameter of the RHMC algorithm.  Note
that these conditions assumed that the target density is lighter than
Gaussian. In comparison, our results can be applied to sub-quadratic
potentials. In addition, it can be shown that HMC is not geometrically
ergodic under \eqref{eq:hyp_boo_rabee_sanz_serna} on the following example associated with the potential defined by  \eqref{eq:def_U_mixture_gaussian} below.

The main difference with the
setting of \cite{bou:sanz:2017} is that HMC has a acceptance/rejection
step and the integrated acceptance ratio
\[ q \mapsto \int_{\rset^d} \alphaacc\defEnsLigne{(\q,\p),\Phiverlet[h][T](\q,\p)}
\rme^{-\norm[2]{p}/2} (2 \uppi)^{-d/2} \rmd p
\]
must not go to $0$ as
$\norm{q}$ goes to $\plusinfty$. This is essentially the reason why
\Cref{assum:potential} differs from \eqref{eq:hyp_boo_rabee_sanz_serna}. Indeed, to show that an
irreducible Markov kernel $\mathrm{P}$ on $(\rset^d, \mcb(\rset^d))$ is not geometrically
ergodic with respect to an invariant measure $\mu$, \cite[Theorem
5.1]{roberts:tweedie:1996:biometrika} states the following sufficient condition
\begin{equation}
  \label{eq:condition_non_geo_ergodicity}
 \mathrm{ess \, sup}_{q \in \rset^d} \mathrm{P}(q, \{q\}) = 1 \eqsp,
\end{equation}
 where
$\mathrm{ess\, sup}$ is taken with respect to $\mu$. Consider then
the target density $\pi$ with potential $U$ given for all $q =(q_1,q_2) \in \rset^2$ by
\begin{equation}
\label{eq:def_U_mixture_gaussian}
  U(q) = -\log(\rme^{-q_1^2 - 5 q_2^2} + \rme^{-5q_1^2 - q_2^2}) \eqsp.
\end{equation}
Note that $U$ satisfies the condition
\eqref{eq:hyp_boo_rabee_sanz_serna}. On the contrary, we may show that
\eqref{eq:condition_non_geo_ergodicity} holds, and therefore HMC is
not geometrically ergodic for such a potential $U$.  However, the
detailed calculations are very technical and not particularly
informative and we prefer to present a numerical evidence that
\eqref{eq:condition_non_geo_ergodicity} holds. Indeed,
\Cref{fig:accept_mix_gaussian} displays numerical computations of the mean acceptance ratio,
$\int_{\rset^2} \alphaacc\defEnsLigne{(\q,\p),\Phiverlet[h][T](\q,\p)}
\rme^{-\norm[2]{p}/2} (2 \uppi)^{-1} \rmd p= 1 -\Pkerhmc[h][T](q,\{q\})$ for $q_1 \in \{200,250,$
$300,350,400,450,500\}$,
$q_2 \in \ccint{q_1+10^{-4},q_1+2\cdot 10^{-4}}$ and $T=1$  which
corresponds to MALA. We can observe that the larger $q_1$, the smaller $1-\Pkerhmc[h][T](q,\{q\})$, which illustrates that  \eqref{eq:condition_non_geo_ergodicity}  holds for the HMC
kernel.

 \begin{figure}[h]
   \centering
   \includegraphics[scale=0.4]{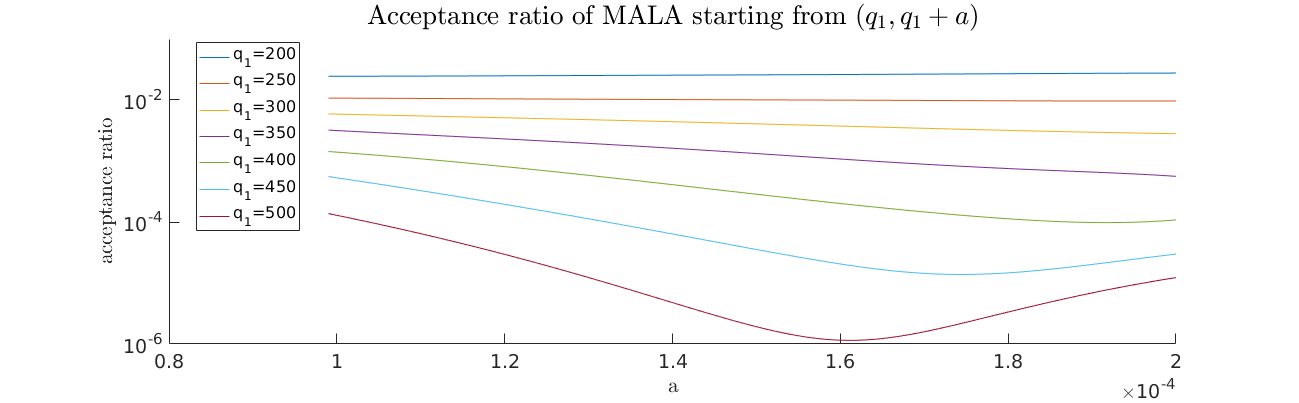}
   \caption{}
   \label{fig:accept_mix_gaussian}
 \end{figure}



\section{Irreducibility for a class of iterative models}
\label{sec:irred-class-iter}

In this Section we establish the irreducibility of a Markov kernel associated to a random iterative model.
These results are of independent interest.
Let $\hfunb : \rset^d \times \rset^d \to \rset^d$ and $\alphagen: \rset^d \times \rset^d \to \ccint{0,1}$ be Borel measurable
functions and $\phib : \rset^d \to \ccint{0,\plusinfty}$ be a
probability density with respect to the Lebesgue measure.  Consider the Markov kernel $\Pkerb$ defined for all $x \in \rset^d$ and $\eventA \in \borelSet(\rset^d)$ by
\begin{equation}
\label{eq:def_pkerb}
  \Pkerb(x,\eventA) = \int_{\rset^d} \indi{\eventA}{\hfunb(x,z)} \alpha(x,z) \phib(z) \rmd z + \bar{\alpha}(x) \delta_x(\eventA) \eqsp,
\end{equation}
where $\bar{\alpha}(x)= \int_{\rset^d} \alpha(x,z) \phib(z) \rmd z$. 
Define for all $x \in \rset^d$, $\hfunb_x : \rset^d \to \rset^d$ by $\hfunb_x = \hfunb(x,\cdot)$.

First, we give  a result  from geometric measure
theory together with a proof for the reader's convenience, which will be essential for the proof of the statements of this section.  Let
$\ouvert\subset\rset^d$ be an open set and $\Theta: \ouvert\to \rset^d$ be
a measurable function such that there exist $y_0 , \tilde{y}_0 \in \rset^d$ and
$M, \tilde{M} > 0$ satisfying $\ball{\tilde{y}_0}{\tilde{M}} \subset \ouvert$ and
  \begin{equation}
    \label{eq:condition_sigma_finite}
\ball{y_0}{M} \subset \Theta(\ball{\tilde{y}_0}{\tilde{M}}) \eqsp.
  \end{equation}
 Define the measure $\lambda_{\Theta}$ on $(\rset^d,\borelSet(\rset^d))$
  by setting for any $\eventA \in \borelSet(\rset^d)$
\begin{equation}
  \label{eq:push_forward_mes}
\lambda_\Theta (\eventA) \eqdef \Leb\defEns{\Theta^{-1}(\eventA) \cap \ball{\tilde{y}_0}{\tilde{M}}}  \eqsp.
\end{equation}
Note 
that $\lambda_\Theta$ is  a finite measure. Therefore by the Lebesgue decomposition theorem (see
\cite[Section 6.10]{rudin:1987}) there exist two  measures
$\lambda_\Theta^{(\text{a})}, \lambda_\Theta^{(\text{s})}$ on
$(\rset^d,\borelSet(\rset^d))$, which are absolutely continuous and
singular with respect to the Lebesgue measure on $\rset^d$
respectively, such that $\lambda_\Theta = \lambda_\Theta^{(\text{a})} +
\lambda_\Theta^{(\text{s})}$.
\begin{proposition}\label{le:simple}
  Let $\ouvert\subset\rset^d$ be open and $\Theta: \ouvert\to \rset^d$ be a Lipschitz function
  satisfying \eqref{eq:condition_sigma_finite}.
 For any version  $\phi_\Theta$ of the density of $\lambda_\Theta^{(\text{a})}$ with respect to the Lebesgue
  measure on $\rset^d$, it holds
$$
\phi_\Theta(y)\geq \1_{ \ball{y_0}{M}}(y) \norm{\Theta}_{\Lip}^{-d}\eqsp, \quad \text{$\Leb$-a.e.}
$$
\end{proposition}
\begin{proof}
  Denote by $L = \norm{\Theta}_{\Lip}$. Let $y\in \ball{y_0}{M}$. By  \eqref{eq:condition_sigma_finite}, we may pick
  $z \in \ball{\tilde{y}_0}{\tilde{M}}$ such that $\Theta(z) = y$. Let
  $\delta_0 >0$ be such that $\ball{z}{\delta_0/L} \subset
  \ball{\tilde{y}_0}{\tilde{M}}$.  Since $\Theta$ is
  Lipschitz continuous,  for all
  $\delta \in \rset_+^*$,
  $\Theta(\ball{z}{\delta/L} \cap \open)\subset \ball{y}{\delta}$. Hence, for all
  $\delta \in \ocint{0,\delta_0}$, we have
  $$
\lambda_\Theta(\ball{y}{\delta})\geq
  L^{-d}\Leb(\ball{z}{\delta}) =   L^{-d}\Leb(\ball{y}{\delta}) \eqsp.
$$
 The claim follows from the differentiation theorem
  for measures, see \cite[Theorem 7.14]{rudin:1987}.
\end{proof}

We can now state our main results. Let $\rassG,\MassG \in \rset^*_+$ and $y_0 \in \rset^d$.
Consider the following assumptions.
\begin{assumptionG}
\label{assumG:phi}
$\phib$ and $\alphagen$ are lower semicontinuous and positive on $\rset^d$ and $\rset^{2d}$ respectively.
\end{assumptionG}


\begin{assumptionG}[$\rassG,y_0,\MassG$]
  \label{assumG:irred_b}
  \begin{enumerate}[label=(\roman*), wide, labelwidth=!, labelindent=0pt]
  \item \label{assumG:irred_b_item_i} There exists $\constLiphx \in \rset_+$ such that for all $x \in \ball{0}{\rassG}$, $\hfunb_x$ is
    $\constLiphx$-Lipschitz, \ie~for all $z_1,z_2 \in \rset^d$,
    $\norm{\hfunb_x(z_1)-\hfunb_x(z_2)} \leq \constLiphx
    \norm{z_1-z_2}$.
\item \label{assumG:irred_b_item_ii} There exist $\tilde{y}_0 \in \rset^d$ and $\tMassG \in \rset_+^*$,
  such that for all $x \in \ball{0}{\rassG}$, $\ball{y_0}{\MassG} \subset \hfunb_x(\ball{\tilde{y}_0}{\tMassG})$.
  \end{enumerate}
\end{assumptionG}

\begin{theorem}
\label{theo:irred}
Assume \Cref{assumG:phi} and that there exist $y_0 \in \rset^d$, $R > 0$ and $M > 0$ such that \Cref{assumG:irred_b}($\rassG,y_0,\MassG$)  is satisfied. Then $\ball{0}{\rassG}$ is $1$-small for $\Pkerb$: for all $x \in \ball{0}{\rassG}$ and $\eventA \in \borelSet(\rset^d)$,
  \begin{equation}
    \Pkerb(x,\eventA) \geq \constLiphx^{-d} \min_{(x,z) \in \ball{0}{R} \times \ball{\tilde{y}_0}{\tMassG}} \defEns{\alphagen(x,z) \phib(z)} \Leb\defEns{\eventA \cap \ball{y_0}{\MassG}} \eqsp,
  \end{equation}
where $(\tilde{y}_0,\tilde{M}) \in \rset^d \times \rset_+^*$ is defined in \Cref{assumG:irred_b}($\rassG,y_0,\MassG$).
\end{theorem}

\begin{proof}
For all $x  \in \ball{0}{\rassG}$ and $\eventA \in \borelSet(\rset^d)$ we get
\begin{align}
\nonumber
\Pkerb(x,\eventA)  &= \int_{\rset^d}\indi{\eventA}{\hfunb(x,z)} \alphagen(x,z) \phib(z) \rmd z  = \int_{\rset^d}\indi{\hfunb_x^{-1}(\eventA)}{z} \alphagen(x,z) \phib(z) \rmd z \\
&\geq  \min_{(x,z) \in \ball{0}{R} \times \ball{\tilde{y}_0}{\tilde{M}}} \defEns{\alphagen(x,z)\phib(z)} \Leb\defEns{\hfunb_x^{-1}(\eventA) \cap\ball{\tilde{y}_0}{\tilde{M}}} \eqsp.
\label{eq:coro_leb_irred_1}
\end{align}
The proof follows from \Cref{le:simple} and \Cref{assumG:irred_b}$(R,y_0,M)$-\ref{assumG:irred_b_item_i} which imply
$ \Leb\defEns{\hfunb_x^{-1}(\eventA )\cap \ball{\tilde{y}_0}{\tilde{M}}} \geq  \constLiphx^{-d}  \Leb\defEns{\eventA  \cap \ball{y_0}{M}}$.
\end{proof}
The following Corollary is a straightforward consequence of \Cref{theo:irred}.
\begin{corollary}
\label{coro:irred}
Assume \Cref{assumG:phi} and  that there exists $(y_0,M) \in \rset^d \times \rset_+^*$ such that  for all $\rassG \in \rset_+^*$ \Cref{assumG:irred_b}($\rassG,y_0,\MassG$). Then $\Pkerb$ is irreducible with irreducibility measure $\Leb\defEns{\cdot \cap \ball{y_0}{\MassG}} $. In addition, all the compact sets are $1$-small.
\end{corollary}
In the next proposition, we give examples of functions $f$ which satisfy \Cref{assumG:irred_b}.
\begin{proposition}\label{le:degree_application}
Let  $\ga$ a function from $ \rset^d\times \rset^d$ to $\rset^d$ and $\ra \in \rset^{*}_+$. Assume that
\begin{enumerate}[label=(\roman*)]
\item
\label{propo:irred_b_item_i}
 there exists $\lipgr \in \rset_+$  such that for all $z_1,z_2,x \in \rset^d$, $\norm{x} \leq \ra$,
  \begin{equation}
    \label{eq:5}
    \norm{g(x,z_1) - g(x,z_2)} \leq \lipgr\norm{z_1-z_2} \eqsp.
  \end{equation}
\item
\label{propo:irred_b_item_ii}
there exist $ \Cga_{\ra,0} , \Cga_{\ra,1}
\in \rset_+$ such that for all $x,z \in \rset^d$, $\norm{x} \leq \ra$
\begin{equation}
  \label{eq:4}
\norm{g(x,z)} \leq  \Cga_{\ra,0} +   \Cga_{\ra,1} \norm{z}
\end{equation}
\end{enumerate}

Let $\bg \in \rset$ and define $\hga : \rset^d \times \rset^d$ for all $x,z \in \rset^d$ by
\begin{equation}
  \hga(x,z) =  \bg z + \ga(x,z) \eqsp.
\end{equation}
If $\norm{\bg} > \Cga_{\ra,1} $, then $\hga$ satisfies \Cref{assumG:irred_b}($\ra,0,\MassG$) for all $\MassG \in \rset_+^*$ with $\tilde{y}_0=0$ and 
\begin{equation}
  \label{eq:deftildeM}
\tilde{M} = \{M  + \Cga_{\ra,0} \}/(\norm{\bg}-\Cga_{\ra,1} ) \eqsp.
\end{equation}
\end{proposition}
We preface the proof by recalling some basic notions of degree theory.
\label{sec:defin-usef-results}
Let $\Dset$ be a bounded open set of $\rset^d$. Let $f:
\Dsetc \to \rset^d$ be a continuous function on
$\Dsetc$ continuously differentiable on $\Dset$. An element $x \in
\Dset$ is said to be a \emph{regular point} of $f$ if the Jacobian matrix of $f$ at $x$, $\Jac_f(x)$, is invertible.
An element $y \in f(\Dset)$ is said to be a \emph{regular value} of $f$ if any $x \in
f^{-1}(\{ y\})$ is a regular point.  

Let $f : \Dsetc \to \rset^d$ be a continuous function,  $C^{\infty}$-smooth on $\Dset$. Let $y \in \rset^d
  \setminus f(\partial \Dset)$ be a regular value of $f$. It is shown in \cite[Proposition and Definition 1.1]{outerelo:ruiz:2009} that the set $f^{-1}(\{y\})$ is finite. The degree of $f$ at $y$ is defined by
\begin{equation}
  \deg(f,\Dset,y) = \sum_{x \in f^{-1}(\{y \})} \sign\defEns{\det \parenthese{\Jac_f(x)}} \eqsp.
\end{equation}

\begin{proposition}[\protect{\cite[Proposition and Definition 2.1]{outerelo:ruiz:2009}}]
\label{defProp:degree_cont}
  Let $f : \Dsetc \to \rset^d$ be a continuous function and $y \in
  \rset^d \setminus f(\partial \Dset)$.
  \begin{enumerate}[label=(\alph*)]
  \item
\label{defProp:degree_cont_i}
 Then there exists  $g \in C(\Dsetc, \rset^d) \cap C^{\infty}(\Dset, \rset^d)$ such that $y$ is a regular value of $g$
  and $\sup_{x \in \Dsetc} \abs{f(x)-g(x)} < \dist(y,f(\partial
  \Dset))$.
\item For all functions $g_1,g_2:\Dsetc \to \rset^d$ satisfying \ref{defProp:degree_cont_i},
  \begin{equation}
    \deg(g_1,\Dset,y) = \deg(g_2,\Dset,y) \eqsp.
  \end{equation}
  \end{enumerate}
\end{proposition}
Under the assumptions of \Cref{defProp:degree_cont}, the degree of $f$ at $y$ is then defined for any $g:\Dsetc \to \rset^d$ satisfying \ref{defProp:degree_cont_i} by
\begin{equation}
  \deg(f,\Dset,y) =  \deg(g,\Dset,y) \eqsp.
\end{equation}

\begin{proposition}[\protect{\cite[Proposition
  2.4]{outerelo:ruiz:2009}}]
  \label{theo:deg_modif}
  Let $f,g : \Dsetc \to \rset^d$ be  continuous functions. Define
  $\hpy:\ccint{0,1} \times \rset^d \to \rset^d$ for all $t \in
  \ccint{0,1}$ and $x \in \rset^d$ by $\hpy(t,x) = t f(x) +
  (1-t)g(x)$. Let $y \in \rset^d \setminus \hpy(\ccint{0,1} \times \partial \Dset)$. Then
\begin{equation}
  \deg(f,\Dset,y) =  \deg(g,\Dset,y) \eqsp.
\end{equation}
\end{proposition}
We have now all the necessary results to prove \Cref{le:degree_application}.
\begin{proof}[Proof of \Cref{le:degree_application}]
Since $\hga(x,z) =  \bg z + \ga(x,z)$ and $\ga(x,\cdot)$ is Lipschitz with a Lipschitz constant which is uniformly bounded over the ball $\ball{0}{R}$,  $\hga_x$ is Lipschitz with bounded Lipschitz constant over this ball. Hence \Cref{assumG:irred_b}($\ra,0,\MassG$)-\ref{assumG:irred_b_item_i} holds.

  For all $x \in \rset^d$, denote by $\hga_x : z \mapsto \hga(x,z)$ where $\hga(x,z)=bz + g(x,z)$.
  Let $\MassG \in \rset_+^*$. We show that for all $x \in
  \ball{0}{\ra}$, $\ball{0}{\MassG} \subset
  \hga_x(\ball{0}{\tMassG})$, where $\tMassG$ is given by
  \eqref{eq:deftildeM}, which is precisely
  \Cref{assumG:irred_b}($\ra,0,\MassG$)-\ref{assumG:irred_b_item_ii}.

  Let $x \in \ball{0}{\ra}$ and consider the continuous homotopy $\hog : \ccint{0,1}
\times \rset^d$ between the functions $z \mapsto \bg z$ and $\hga_x$ defined for all
$t \in \ccint{0,1}$ and $z \in \rset^d$ by
\begin{equation}
  \hog(t,z) = t \bg z + (1-t)\hga_x(z) = \bg z + (1-t)  \ga(x,z)  \eqsp.
\end{equation}
Then by \ref{propo:irred_b_item_ii}, since $\abs{\bg} \geq \Cga_{\ra,1} $, for all $t\in \ccint{0,1}$ and $z \not \in
\ball{0}{\tMassG}$, where $\tMassG$ is given by \eqref{eq:deftildeM},
\begin{equation}
   \abs{\hog(t,z)} \geq  \abs{\bg z} -(1-t)\defEns{\Cga_{\ra,0} +\Cga_{\ra,1} \abs{z} } \geq \MassG \eqsp.
\end{equation}
In particular, we have $\hog(\ccint{0,1} \times \partial
\ball{0}{\tMassG}) \subset \rset^d \setminus \ball{0}{\MassG}$. Let
$z \in \ball{0}{\MassG}$, then by
\Cref{theo:deg_modif} we have
\begin{equation}
  \deg(\hga_x,\ball{0}{\tMassG},z) = \deg(\bg \Id, \ball{0}{\tMassG},z) = 1 \eqsp.
\end{equation}
Besides, by \cite[Corollary 2.5, Chapter IV]{outerelo:ruiz:2009},
$\deg(\hga_x,\ball{0}{\tMassG},z) \not = 0$ implies that there exists
$y \in \ball{0}{\tMassG}$ such that $\hga_x(y) = z$. Finally \Cref{assumG:irred_b}($\ra,0,\MassG$)-\ref{assumG:irred_b_item_ii} follows since this
result holds for all $z \in \ball{0}{\MassG}$.
\end{proof}

\section{Proofs}
\label{sec:postponed-proofs}
 In the sequel, $C \geq 0$
  is a constant which can change from line to line but does not depend
  on $h$. Let $h >0$ and $T \in \nset^*$. Note that a simple induction (see \cite[Proposition 4.2]{livingstone:betancourt:byrne:girolami:2016})  implies that for all $(\q_0,\p_0) \in \rset^d
\times \rset^d$ and $k \in \{1,\ldots T\}$, the $k^{\text{th}}$
iteration of the leap-frog integration, $(q_k,p_k) = \Phiverlet[h][k](\q,\p)$, where $\Phiverlet[h][k]$ is defined by \eqref{eq:def_Phiverlet}, takes the form
\begin{align}
\label{eq:qk}
\q_k&=\q_0+kh\p_0-\frac{kh^2}{2} \nabla \F(\q_0)-h^2 \gperthmc[k](\q_0,\p_0)\\
\label{eq:pk}
\p_{k}&= \p_0-\frac{h}{2} \defEns{\nabla \F(\q_0)+\nabla \F \circ \Phiverletq[h][k] (\q_0,\p_0)}-h \sum_{i=1}^{k-1}\nabla \F \circ \Phiverletq[h][i](\q_0,\p_0)  \eqsp,
\end{align}
where  $\gperthmc[k] :\rset^d \times \rset^d \to \rset^d$ is given  for all $(\q,\p) \in \rset^d \times
\rset^d$ by
\begin{equation}
  \label{eq:def_gperthmc}
\gperthmc[k](\q,\p) =    \sum_{i=1}^{k-1}(k-i)\nabla \F \circ \Phiverletq[h][i](\q,\p) \eqsp.
\end{equation}

We prefaces the proofs of our main results by useful bounds  on
  the position and the momentum in the intermediate steps of the leap-frog integration.

  \begin{lemma}
    \label{lem:bound_first_iterate_leapfrog_a}
  Assume \Cref{assum:regOne}$(\expozero)$-\ref{assum:regOne_a}. Then, for any $k \in \nsets$, $h \geq 0$, $(q_0,p_0) \in \rset^{2d}$ and $(x_0,v_0) \in \rset^{2d}$,
  \begin{align}
   & \norm{q_k-x_k} + \constzero^{-1/2} \norm{p_k-v_k} \\
    & \qquad \qquad \leq \defEns{1+h \constzero^{1/2} \vartheta_1(h \constzero^{1/2})}^{k} \defEns{\norm{q_0-x_0} + \constzero^{-1/2} \norm{p_0-v_0}} \eqsp,
  \end{align}
  where $(q_k,p_k)= \Phiverlet[h][k](q_0,p_0)$, $(x_k,v_k)= \Phiverlet[h][k](x_0,v_0)$ and  $\Phiverlet[h][k]$ and $\vartheta_1$ are defined by \eqref{eq:def_Phiverlet} and \eqref{eq:def_vartheta_1}, respectively.
\end{lemma}
\begin{proof}
  Note that it is sufficient to show the result for $k=1$ and to apply a straightforward induction.
  Let $h >0$, $(q_0,p_0) \in \rset^{2d}$ and $(x_0,v_0) \in \rset^{2d}$. Using \eqref{eq:qk}, the triangle inequality and   \Cref{assum:regOne}$(\expozero)$-\ref{assum:regOne_a}, we first obtain
  \begin{align}
    \nonumber
    \norm{q_1-x_1} & = \norm{q_0 -h^2 \nabla U(q_0) /2 + h p_0 - \defEns{x_0 - h^2/2 \nabla U(x_0) + h v_0}} \\
    \label{eq:bound_q_1_x_1}
         & \leq (1+ h^2 \constzero /2) \norm{q_0 - x_0} + h \norm{p_0 - v_0} \eqsp.
  \end{align}
  Second, similarly using \eqref{eq:pk}, we have that
\begin{align}
\label{eq:bound_p_1_v_1}
&    \norm{p_1-v_1}  \\
&= \norm{p_0-v_0 -(h/2) \defEns{\nabla U(q_1) + \nabla U(q_0)} + (h/2) \defEns{\nabla U(x_1) + \nabla U(x_0)}} \\
\nonumber
&\leq \norm{p_0 - v_0} + (h \constzero/2) \defEns{ \norm{x_1- q_1} + \norm{x_0 - q_0}} \\
\nonumber
& \leq \left(1+h^2 \constzero/2\right) \norm{p_0-v_0} +h\constzero(1+h^2\constzero /4) \norm{q_0-x_0} \eqsp,
\end{align}
  where we have used \eqref{eq:bound_q_1_x_1} for the last inequality. Summing up \eqref{eq:bound_q_1_x_1} and \eqref{eq:bound_p_1_v_1}, we get the desired result for $k=1$.
  \end{proof}
  \begin{lemma}
  \label{lem:bound_first_iterate_leapfrog_b}
Let $\beta \in \ccint{0,1}$ and assume \Cref{assum:regOne}$(\expozero)$-\ref{assum:regOne_b}.
  \begin{enumerate}[label=(\roman*)]
\item
\label{lem:bound_first_iterate_leapfrog_1}
For any $h_0 >0$, $T \in \nsets$, there exists $C < \infty$ (which depends only on $T,h_0$
 and $\constzeroT$) such that for all $h \in \ocint{0,h_0}$,
  $(\q_0,\p_0) \in \rset^d \times \rset^d$ and $k \in \{1,\ldots, T\}$
  \begin{align}
\label{lem:bound_first_iterate_leapfrog_1_q}
    \norm{\q_k-\q_0} &\leq C h\defEns{  \norm{\p_0} +h(1+ \norm{\q_0}^{\expozero})}\\
\label{lem:bound_first_iterate_leapfrog_1_p}
    \norm{\p_k-\p_0}& \leq C h\defEns{  1+ \norm{\p_0}^{\expozero}+ \norm{\q_0}^{\expozero}} \eqsp,
  \end{align}
where $(\q_k,\p_k) = \Phiverlet^{\circ k}_{h}(\q_0,\p_0)$ and  $\Phiverlet^{\circ k}_{h}$ is defined by \eqref{eq:def_Phiverlet}.
\item \label{lem:bound_first_iterate_leapfrog_b_2}
  If in addition \Cref{assum:regOne}$(\expozero)$-\ref{assum:regOne_a} holds, for any $k \in \nsets$, $h >0$, $(q_0,p_0) \in \rset^{2d}$,
  \begin{align}
    &    \norm{q_k-q_0} + \constzero^{-1/2}\norm{p_k - p_0} \leq (\constzero^{1/2}\vartheta_1(h \constzero^{1/2}))^{-1} \\
    & \quad \times \defEns{(1+h \constzero^{1/2}\vartheta_1(h \constzero^{1/2}))^{k+1} - 1} \defEns{\vartheta_2(h) (\norm[\beta]{q_0} +1) + \vartheta_3(h) \norm{p_0}} \eqsp,
  \end{align}
  where $\vartheta_1$ is defined by \eqref{eq:def_vartheta_1} and
  \begin{align}
  \label{eq:definition-vartheta-2}
    \vartheta_2(h) &= \constzeroT/\constzero^{1/2} + \constzeroT h/2 + \constzero^{1/2}\constzeroT h^2/4 \eqsp, \\
  \label{eq:definition-vartheta-3}
    \vartheta_3(h) &= 1+\constzero^{1/2} h /2 \eqsp.
  \end{align}
  \end{enumerate}
\end{lemma}
\begin{proof}
  \begin{enumerate}[label={(\roman*)},wide=0pt, labelindent=\parindent]
  \item Let $T \in \nsets$ and $h_0 >0$.  We prove by induction that for all $k \in \{1,\ldots,T\}$ there exists $C_k \geq 0$ (which depends only on $T,h_0$
and $\constzeroT$) such that for all $h \in \ocint{0,h_0}$ and
  $(q_0,p_0) \in \rset^d \times \rset^d$
\begin{equation}
\label{lem:bound_first_iterate_leapfrog_1_q}
\begin{aligned}
\norm{q_k-q_0} \leq C_k h\defEns{  \norm{p_0} +h(1+ \norm{q_0}^{\expozero})} \\
\norm{p_k-p_0} \leq C_k h \defEns{ 1 + \norm{p_0}^{\expozero} + \norm{q_0}^{\expozero}} \eqsp.
\end{aligned}
\end{equation}
  where $(q_k,p_k) = \Phiverlet^{\circ k}_{h}(q_0,p_0)$.
 Let $ h \in \ocint{0, h_0}$ and $(q_0,p_0) \in \rset^d \times
    \rset^d$.
 The case $k=1$ is immediate by
\Cref{assum:regOne}($\beta$)-\ref{assum:regOne_b} and \eqref{eq:qk}. Let $k \in \{1,\cdots,
T-1\}$ and assume that the inequalities hold for all $i \in
\{1,\dots, k\}$. Then by \eqref{eq:qk} and
\Cref{assum:regOne}($\beta$)-\ref{assum:regOne_b}, we get
\begin{align}
\label{eq:lem:bound_first_iterate_leapfrog_1}
\norm{q_{k+1}-q_0}
&\leq (k+1)h \norm{p_0}+\frac{k+1}{2}h^2 \constzeroT  \defEns{1+\norm{ q_0}^{\expozero} }\\
&\qquad \qquad +h^2 \constzeroT \sum_{i=1}^{k}(k+1-i)\defEns{1+ \norm{q_i}^{\expozero} }\eqsp.
\end{align}
By the induction hypothesis and using that $t \mapsto t^{\expozero}$ is sub-additive on $\rset^+$ and $t^\beta \leq 1 + t$ for $t \in \rset^+$, we get for all $i \in \{1,\cdots, k\}$,
\begin{equation}
\norm{q_i}^{\expozero} \leq \norm{q_0}^{\expozero} + \norm{q_i-q_0}^\beta \leq 1+ \norm{q_0}^{\expozero} +C_i h\defEns{\norm{p_0}+h(1+ \norm{q_0}^{\expozero})} \eqsp,
\end{equation}
Plugging this inequality in \eqref{eq:lem:bound_first_iterate_leapfrog_1}
conclude the proof of \eqref{lem:bound_first_iterate_leapfrog_1_q}.
Consider now \eqref{lem:bound_first_iterate_leapfrog_1_p}. Since by
definition $p_{k+1} = p_k-(h/2)\defEnsLigne{\nabla \F(q_{k}) + \nabla
  \F(q_{k+1})}$, using the triangle inequality,
\Cref{assum:regOne}($\beta$)-\ref{assum:regOne_b},
\eqref{lem:bound_first_iterate_leapfrog_1_q} to bound $\norm{q_{k}}$ and
$\norm{q_{k+1}}$, and the induction hypothesis, we get that there exist some
constants $C_{k+1,1},C_{k+1,2}$ which only depend on $T,h_0$ and
$\constzeroT$ such that
\begin{align}
&  \norm{p_{k+1} - p_0} \leq \norm{p_k-p_0} + (h/2)\defEnsLigne{\norm{\nabla \F(q_{k})} + \norm{\nabla \F(q_{k+1})}} \\
&\quad \leq C_{k+1,1}h\defEns{1+\norm{p_0}^{\expozero}+\norm{q_0}^{\expozero}} + (\constzeroT h/2)\defEns{2+\norm{q_{k}}^{\expozero}+ \norm{q_{k+1}}^{\expozero}}\\
&\quad \leq C_{k+1,1}h\defEns{1+\norm{p_0}^{\expozero}+\norm{q_0}^{\expozero}} + (C_{k+1,2}h/2)\defEns{1+ \norm{q_0}^{\expozero} + \norm{p_0}^{\expozero}} \eqsp.
\end{align}
Therefore, \eqref{lem:bound_first_iterate_leapfrog_1_q} is satisfied which concludes the induction and the proof.
\item Let $k \in \nset$, $h >0$ and $(q_0, p_0) \in \rset^{2d}$. Using \eqref{eq:qk}, the triangle inequality and \Cref{assum:regOne}($\beta$), we have
  \begin{align}
\label{eq:1_lem_bound_first_iterate_leapfrog_b_2}
&    \norm{q_{k+1} - q_0}  = \norm{q_k - q_0 - (h^2/2) \nabla U(q_k) + h p_k} \\
\nonumber
&                          \leq (1+h^2 \constzero/2) \norm{q_k - q_0} + (h^2\constzeroT/2) (\norm[\beta]{q_0} +1) + h \norm{p_k-p_0} + h \norm{p_0}  \eqsp.
  \end{align}
Second, similarly using \eqref{eq:pk}, we get that
\begin{align}
\nonumber
&    \norm{p_{k+1} - p_0}   \leq  \norm{p_k-p_0 + (h/2) \defEns{\nabla U(q_{k+1}) + \nabla U(q_k)} } \\
    \nonumber
                         & \leq \norm{p_k - p_0} + (h \constzero/2) \defEns{ \norm{q_{k+1}- q_0} + \norm{q_k - q_0}} + h \constzeroT \defEns{\norm[\beta]{q_0} +1} \\
    \nonumber
                         & \leq \norm{p_k - p_0}  +  h \constzeroT \defEns{\norm[\beta]{q_0} +1} +(h \constzero/2) \defEns{  + h \norm{p_k-p_0} + h \norm{p_0}} \\
    \label{eq:2_lem_bound_first_iterate_leapfrog_b_2}
&     +(h \constzero/2) \defEns{ (2+\constzero h^2/2) \norm{q_k-q_0} + (h^2\constzeroT/2) (\norm[\beta]{q_0}+1)}  \eqsp.
  \end{align}
  where we have used \eqref{eq:1_lem_bound_first_iterate_leapfrog_b_2} for the last inequality.
  Summing up \eqref{eq:1_lem_bound_first_iterate_leapfrog_b_2} and \eqref{eq:2_lem_bound_first_iterate_leapfrog_b_2} and using the definition \eqref{eq:def_vartheta_1} of $\vartheta_1(h)$, we get that, setting $A_k= \norm{q_k-q_0}  + \constzero^{-1/2} \norm{p_k-p_0}$,
  \[
  A_{k+1} \leq (1+h\constzero^{1/2}\vartheta_1(h\constzero^{1/2})) A_k    + h \defEns{ \vartheta_2(h) (\norm[\beta]{q_0} +1) +  \vartheta_3(h) \norm{p_0}} \eqsp.
\]
  By a straightforward induction, we obtain that
  \begin{equation}
    A_{k+1} \leq \sum_{i=1}^{k+1} \parentheseDeux{(1+h \constzero^{1/2} \vartheta_1(h \constzero^{1/2} ))^{k+1-i} h \defEns{ \vartheta_2(h) (\norm[\beta]{q_0} +1) +  \vartheta_3(h) \norm{p_0} }} \eqsp,
  \end{equation}
  which completes the proof of \ref{lem:bound_first_iterate_leapfrog_b_2}.
\end{enumerate}
\end{proof}

\begin{lemma}
\label{lem:inverse_1}
Let $\beta \in \ccint{0,1}$ and assume \Cref{assum:regOne}($\beta)$. Then for any $T \in \nsets$, $h >0$,
\begin{align}
\label{eq:lem:inverse_1}
  &  \sup_{(\q,\p,v) \in \rset^{3d}} \{ \norm{ \gperthmc[T](q,\p) - \gperthmc[T](q,v)} / \norm{\p-v} \} \\
  & \qquad \qquad \qquad \qquad  \leq   (T/h)  \left\{ (1+\ h \constzero^{1/2} \vartheta_1(h \constzero^{1/2}))^T - 1 \right\} \eqsp,
\end{align}
where  $\gperthmc[T]$ and $\vartheta_1$ are defined by  \eqref{eq:def_gperthmc} and \eqref{eq:def_vartheta_1} respectively.
  In addition, for any $q \in \rset^d$,
  \begin{equation}
    \label{eq:lem:inverse_2}
    \norm{\gperthmc[T](q,0)} \leq  (T/h)  \{ (1+h \constzero^{1/2} \vartheta_1(h \constzero^{1/2} ))^{T}-1\}  \vartheta_2(h) (\norm[\beta]{q} +1)  + T^2 \norm{\nabla U(q)} \eqsp,
  \end{equation}
  where $\vartheta_2$ is defined in \eqref{eq:definition-vartheta-2}.
\end{lemma}
\begin{proof}
By
\Cref{lem:bound_first_iterate_leapfrog_a}, for any $i \in \nsets$, we get
\[
\sup_{(q,\p,v) \in \rset^{3d}} \defEns{ \normLigne{ \Phiverletq[h][i](q,\p) - \Phiverletq[h][i](q,v)} / \norm{\p-v} }
\leq \constzero^{-1/2} A^i
\]
where $A=(1+h \constzero^{1/2} \vartheta_1(h \constzero^{1/2} ))$.
Therefore by definition of $\gperthmc[T]$ \eqref{eq:def_gperthmc} and using  \Cref{assum:regOne}($\beta$), for any $h >0$, $T \in \nsets$, we get that
\begin{align}
&  \sup_{(\q,\p,v) \in \rset^{3d}} \{ \norm{ \gperthmc[T](q,\p) - \gperthmc[T](q,v)} / \norm{\p-v} \} \\
&\qquad  \leq \constzero   \sum_{k=1}^{T-1} (T-i) \sup_{(q,\p,v) \in \rset^{3d}} \defEns{ \normLigne{ \Phiverletq[h][i](q,\p) - \Phiverletq[h][i](q,v)} / \norm{\p-v} }\\
& \qquad \leq    T \parentheseDeux{ A^{T} - 1 } /(h \vartheta_1(\constzero^{1/2} h)) \,
\end{align}
showing \eqref{eq:lem:inverse_1} since $\vartheta_1(h\constzero^{1/2}) \geq 1$.

We now consider \eqref{eq:lem:inverse_2}. By \eqref{eq:def_gperthmc}, \Cref{assum:regOne}($\beta$)-\ref{assum:regOne_a} and \Cref{lem:bound_first_iterate_leapfrog_b}-\ref{lem:bound_first_iterate_leapfrog_b_2}, we have that for any $q \in \rset^d$,
\begin{align}
&  \norm{\gperthmc[T](q,0)}  \leq \sum_{i=1}^{T-1} (T-i)  \norm{ \nabla U \circ \Phiverletq[h][i](\q,0) - \nabla U(q)} + T^2 \norm{\nabla U(q)} \\
& \leq  T \constzero \sum_{i=1}^{T-1}  \norm{ \Phiverletq[h][i](\q,0) -q} + T^2 \norm{\nabla U(q)} \\
& \quad \leq T \constzero^{1/2} \vartheta_1^{-1} (h \constzero^{1/2}) \sum_{i=1}^{T-1}\{A^{i+1}-1\} \defEns{\vartheta_2(h) (\norm[\beta]{q} +1)}  + T^2 \norm{\nabla U(q)} \eqsp,
\end{align}
which completes the proof of \eqref{eq:lem:inverse_2} using that $\vartheta_1(h \constzero^{1/2}) \geq 1$.
\end{proof}

\begin{lemma}
  \label{lem:bounded_cum_error}
  Assume \Cref{assum:regOne}$(1)$. Then for any $T \in \nsets$, $h >0$, and $q,p \in \rset^d$,
  \begin{align}
    &\sum_{i=1}^{T}  \norm{ \Phiverletq[h][i](\q,p) -q} \\
    &\leq  L^{-1/2}_1 T [\{1 + h  \constzero^{1/2} \vartheta_1(h\constzero^{1/2})\}^T  -1] \defEns{\vartheta_2(h) (1+\norm{q})  + \vartheta_3(h) \norm{p}} \eqsp,
  \end{align}
where  $\vartheta_1$ is defined by  \eqref{eq:def_vartheta_1}.
\end{lemma}
\begin{proof}
For  $T \geq 2$ and $h >0$,  by  \Cref{lem:bound_first_iterate_leapfrog_b}-\ref{lem:bound_first_iterate_leapfrog_b_2},   we have
\begin{align}
\sum_{i=1}^{T-1}  \norm{ \Phiverletq[h][i](\q,p) -q}
&\leq 
 \constzero^{-1/2} \vartheta_1^{-1}(h \constzero^{1/2}) T \defEns{\{1+h \constzero^{1/2}\vartheta_1(h\constzero^{1/2})\}^{T}-1} \\
 & \qquad \qquad \times \, \defEns{\vartheta_2(h) (\norm{q} +1) + \vartheta_3(h) \norm{p}} \eqsp.
\end{align}
The proof is completed upon using that $\vartheta_1(h \constzero^{1/2}) \geq 1$.
\end{proof}


\subsection{Proofs of \Cref{sec:ergodicity-hmc}}


\subsubsection{Proof of \Cref{theo:irred_harris} }
\label{sec:proof-crefth-harris_0}
We first prove  \eqref{theo:irred_harris_a}.  Under the assumption that $\F$ is twice continuously
  differentiable, it follows by a straightforward induction, that for
  all $h >0$ and $q \in \rset^d$, $p \mapsto
  \Phiverletq[h][k](q,p)$, defined by  \eqref{eq:def_Phiverletq}, and $p \mapsto \gperthmc[k](q,p)$, defined by \eqref{eq:def_gperthmc}, are
  continuously differentiable and for all $(q,p) \in \rset^d \times
  \rset^d$,
\begin{equation}
  \Jac_{p,\gperthmc[T]}(q,p) =  \sum_{i=1}^{T-1}(T-i)\defEns{\nabla^2 \F \circ \Phiverletq[h][i](q,p)} \Jac_{p,\Phiverletq[h][i]}(q,p) \eqsp,
\end{equation}
where for all $q \in \rset^d$, $\Jac_{p,\gperthmc[k]}(q,p)$ ($\Jac_{p,\Phiverletq[i][h]}(q,p)$ respectively) is the Jacobian of the function $\tilde{p} \mapsto
\gperthmc[k](q,\tilde{p})$ ($\tilde{p} \mapsto
\Phiverletq[i][h](q,\tilde{p})$ respectively) at $p \in \rset^{d}$.

Under \Cref{assum:regOne}, $\sup_{x \in \rset^d} \normLigne{\nabla^2 \F(x)}
 \leq \constzero$, therefore by \Cref{lem:inverse_1}, we have that for any $T \in \nsets$ and $h >0$,
\begin{equation}
  \label{eq:inverse_1}
 \sup_{(\q,\p) \in \rset^d \times \rset^d} \norm{\Jac_{p,\gperthmc[T]}(q,p)}
 \leq  T (\{1 + h \constzero^{1/2} \vartheta_1(h \constzero^{1/2})\}^T  -1) /h  \eqsp.
\end{equation}
For any $q \in \rset^d$, $T\in \nsets$ and $h >0$, define $\phia_{q,T,h}(p)$  for all  $p \in \rset^d$ by
\begin{equation}
  \phia_{q,T,h} (p) = p-(h/T) \gperthmc[T](q,p) \eqsp.
\end{equation}
It is a well known fact (see for example
\cite[Exercise 3.26]{duistermaat:kolk:2004}) that if
\begin{equation}
  \label{eq:inverse_1_2}
  \sup_{(q,p) \in \rset^d \times \rset^d} (h/T)\norm{ \Jac_{p,\gperthmc[T]}(q,p)} < 1 \eqsp,
\end{equation}
then for any $q \in \rset^d$, $\phia_{q,T,h}$ is a
diffeomorphism and  therefore by \eqref{eq:qk}, the same conclusion holds
for $p \mapsto \Phiverletq[h][T](q,p)$. Using \eqref{eq:inverse_1}, if $T \in \nsets$ and $h > 0$  satisfies \eqref{eq:condition-h,T-harris},
then the condition \eqref{eq:inverse_1_2} is verified and as a result \eqref{theo:irred_harris_a}.

Denoting for any $q \in \rset^d$ by $\Phiverletqi[h][T](q,\cdot) : \rset^d \to \rset$ the
continuously differentiable inverse of $p \mapsto
\Phiverletq[h][T](q,p)$ and using a change of variable with $\Phiverletqi[h][T](q,\cdot)$ in \eqref{eq:def_kernel_hmc} concludes the proof of \eqref{eq:def_kernel_hmc_false_density}.

We now show that $\Tker_{h,T}$ satisfies the condition which implies that $\Pkerhmc[h][T]$ is a \Tkernel. We first establish some estimates on the function $(q,p) \mapsto \Phiverletqi[h][T](q,p)$. By
\eqref{eq:inverse_1_2} and \eqref{eq:qk}, for any $q,p,v \in \rset^d$, there exists $\varepsilon \in \ooint{0,1}$ such that $  \normLigne{\Phiverletq[h][T](q,p)-\Phiverletq[h][T](q,v)} \geq (hT) \normLigne{\phi_{q,T,h}(p)-\phi_{q,T,h}(v)} \geq (hT) (1-\varepsilon)\norm{p-v}$ which implies that that there exists $C \geq 0$ satisfying
\begin{equation}
  \label{eq:regularity_phinverse1}
  \begin{aligned}
    \norm{\Phiverletqi[h][T](q,p)-\Phiverletqi[h][T](q,v)} &\leq (1-\varepsilon)^{-1} \norm{v-p}\eqsp, \\
    \norm{  \Phiverletqi[h][T](q,p)} &\leq C\defEns{\norm{\p} + \norm{\Phiverletq[h][T](q,0)}} \eqsp.
  \end{aligned}
\end{equation}
In addition, for $q,x,p \in \rset^d$, we have setting $\tilde{q} = \Phiverletqi[h][T](q,p)$ that
\begin{align}
  \nonumber
  \normLigne{\Phiverletqi[h][T](q,p) - \Phiverletqi[h][T](x,p)} &= \normLigne{\tilde{q} - \Phiverletqi[h][T](x, \Phiverletq[h][T](q,\tilde{q}))} \\
  \nonumber
                                                                &= \normLigne{\Phiverletqi[h][T](x, \Phiverletq[h][T](x,\tilde{q})) - \Phiverletqi[h][T](x, \Phiverletq[h][T](q,\tilde{q}))} \eqsp,
\end{align}
which implies by \eqref{eq:regularity_phinverse1} and \Cref{lem:bound_first_iterate_leapfrog_a}
that there exists $C \geq 0$ satisfying
\begin{equation}
  \label{eq:regularity_phinverse2}
  \norm{\Phiverletqi[h][T](q,p) - \Phiverletqi[h][T](x,p)} \leq C \norm{q-x} \eqsp.
\end{equation}

We now can prove that $\Tker_{h,T}$ is the continuous component of $\Pkerhmc[h][T]$. First by \eqref{eq:def_tker}, for all $\eventB \in \borelSet(\rset^d)$,
\begin{equation}
\label{eq:minoration_pseudo_density_P}
    \Tker_{h,T}(q, \eventB) \geq (2 \uppi)^{-d/2} \Leb(\eventB)
 \times \inf_{\bar{q} \in \eventB} \defEns{ \balphaacc(q,\bar{q}) \rme^{-\norm{\Phiverletqi_q(\bar{q})}^2/2}\detj_{\Phiverletqi[h][T](q,\cdot)}(\bar{q})} \eqsp,
\end{equation}
with the convention $0 \times \plusinfty = 0$ and
\begin{equation}
  \balphaacc(q,\bar{q}) =  \alphaacc\defEns{(q,\Phiverletqi[h][T](q,\bar{q})),\Phiverlet[h][T](q,\Phiverletqi[h][T](q,\bar{q}))}\eqsp. 
\end{equation}
Since the function $  (q,p) \mapsto (\Phiverletq[h][T](q,p),\Phiverletqi[h][T](q,p), \detj_{\Phiverletqi[h][T](q,\cdot)}(p)) $
is continuous on $\rset^d\times \rset^d$ by \Cref{lem:bound_first_iterate_leapfrog_a}, \eqref{eq:regularity_phinverse1} and \eqref{eq:regularity_phinverse2}, and for any $q,p \in \rset^d$, $\Jac_{\Phiverletq[h][T](\q,\cdot)}(\Phiverletqi[h][T](q,p))
\Jac_{\Phiverletqi[h][t](q,\cdot)}(\p) = \operatorname{I}_n$, we get that  $\Tker_{h,T}(q,\eventB) >0$ for all $q \in \rset^d$ and all compact set $\eventB$ satisfying $\Leb(\eventB) > 0$. Therefore, using that the Lebesgue measure is regular which implies that for any $\msa \in \mcb(\rset^d)$ with $\Leb(\msa) >0$, there exists a compact set $\msb \subset\msa$, $\Leb(\msb)>0$, we can conclude that $\Pkerhmc[h][T]$ is irreducible with respect to the Lebesgue measure. In addition, we get  $\Tker_{h,T}(q,\rset^d) >0$, and therefore we obtain that $\Pkerhmc[h][T]$ is aperiodic.  Similarly we get that any compact set is $(1,\Leb)$-small.

It remains to show that for any $\eventB \in\mcb(\rset^d)$, $q \mapsto \Tker_{h,T}(q,\eventB)$ is lower semi-continuous which is a straightforward consequence of Fatou's Lemma and that for any $p \in \rset^d$, $q \mapsto (\Phiverlet[h][T](q,p), \Phiverletqi[h][T](q,p),\detj_{\Phiverletqi[h][T](q,\cdot)}(p))$ is continuous.

Finally, the last statements of \ref{theo:irred_harris_c} follows from \Cref{propo:harris_rec} in \Cref{sec:harr-recurr-metr} which implies that  $ \Pkerhmc[h][T]$ is Harris recurrent and  \cite[Theorem 13.0.1]{meyn:tweedie:2009} which implies  \eqref{eq:harris-theorem}.

\subsubsection{Proof of \Cref{theo:irred_D}}
\label{sec:proof-crefth_irred_D}
We use \Cref{coro:irred}. Indeed $\Pkerhmc[h][T]$ is
of form \eqref{eq:def_pkerb} and it is straightforward to check that it
satisfies \Cref{assumG:phi} (note that \Cref{lem:bound_first_iterate_leapfrog_a}
shows that $\Phiverlet[h][T]$ is a Lipshitz function on $\rset^{2d}$).

We now check that $\Pkerhmc[h][T]$ satisfies \Cref{assumG:irred_b}($\rassG,0,\MassG$) for all $\rassG,\MassG \in
\rset_+^*$ using \Cref{le:degree_application}.  By \eqref{eq:qk}, for all $T \in \nsets$, $h >0$, $q,p \in \rset^d$,
\begin{equation}
  \label{eq:phiverlet_gqth}
  \Phiverletq[h][T](q,p) = T
h p + g_{q,T,h}(p)
\end{equation}
where $g_{q,T,h}(p) = q - (Th^2/2) \nabla \F(q) -
h^2 \gperthmc[T](q,p)$ where $\gperthmc[T]$ is defined by \eqref{eq:def_gperthmc}. \Cref{lem:inverse_1} shows that for any $T \in \nsets$ and $h >0$, it holds that
\begin{equation}
    \label{eq:2:theo:irred_D}
\sup_{p,v,q \in  \rset^d} \frac{\norm{g_{q,T,h}(p)-g_{q,T,h}(v)}}{\norm{p - v}} \leq T h [\{1 + h \constzero^{1/2} \vartheta_1(h \constzero^{1/2} )\}^T-1] \eqsp,
\end{equation}
which implies that the condition
\Cref{le:degree_application}-\ref{propo:irred_b_item_i} is satisfied. To check that
condition  \Cref{le:degree_application}-\ref{propo:irred_b_item_ii} holds, we consider separately the two cases: $\beta <1$ and $\beta =1$.

\begin{enumerate}[label=$\bullet$, wide, labelwidth=!, labelindent=0pt]
\item Consider first the case $\beta <1$. By \Cref{assum:regOne}-\ref{assum:regOne_b},
for any $T \in \nsets$ and $h >0$, we get
\begin{equation}
\norm{\gperthmc[T](\q,\p)} \leq  T \sum_{i=1}^{T-1} \norm{\nabla \F \circ \Phiverletq[h][i](\q,\p)} \leq
\constzeroT T \sum_{i=1}^{T-1} \defEns{ 1 + \norm{\Phiverletq[h][i](\q,\p)}^{\expozero}}
 \eqsp.
\end{equation}
Hence, by \Cref{lem:bound_first_iterate_leapfrog_b}-\ref{lem:bound_first_iterate_leapfrog_1}
there exists $C \geq 0$ such that for all $R\in \rset_+^*$ and
$q,p \in \rset^d$, $\norm{q} \leq R$,
\begin{equation}
\label{eq:3:theo:irred_D}
\norm{g_{q,T,h}(p)} \leq C \defEns{1+R^{\beta} +\norm{p}^{\expozero}} \eqsp,
\end{equation}
which implies that condition \ref{propo:irred_b_item_ii} of \Cref{le:degree_application} holds for any $T \in \nsets$ and $h >0$.

\item Consider now the case $\beta =1$.  For any $T \in \nsets$, $h >0$,  $q,p \in \rset^d$ we get using \Cref{assum:regOne}-\ref{assum:regOne_a}
\begin{align}
  \norm{g_{q,T,h}(p)} &\leq \norm{q} + Th^2 \constzero  \norm{q} /2 + Th^2 \norm{\nabla U(0)} /2\\
  & \qquad \qquad +h^2 \norm{\gperthmc[T](q,p) - \gperthmc[T](q,0)} + h^2 \norm{ \gperthmc[T](q,0)} \eqsp.
\end{align}
Therefore using \Cref{lem:inverse_1}, for any $q,p \in \rset^d$, $\norm{q} \leq R$ for $R \geq 0$, for any $T \in \nsets$ and $h >0$ satisfying \eqref{eq:condition-h,T-harris}, there exists $C \geq 0$ such that
\begin{equation}
  \norm{g_{q,T,h}(p)} \leq C + h T  [ \{1+ h \constzero^{1/2} \vartheta_1(h\constzero^{1/2})\}^T-1]  \norm{p} \eqsp,
\end{equation}
showing that condition \ref{propo:irred_b_item_ii} of \Cref{le:degree_application} is satisfied.
\end{enumerate}

Therefore,  \Cref{le:degree_application} can be applied and for any $T \in \nsets$ and $h >0$ if $\beta <1$ and for any $h > 0$ and $T \in \nsets$ satisfying \eqref{eq:condition-h,T-harris} if $\beta =1$, $\Pkerhmc[h][T]$ satisfies \Cref{assumG:irred_b}($\rassG,0,\MassG$) for all $\rassG,\MassG \in
\rset_+^*$.  \Cref{coro:irred} concludes the proof of \ref{theo:irred_D_a} and \ref{theo:irred_D_b}.
The last statement then follows from   \cite[Theorem 14.0.1]{meyn:tweedie:2009}.




\subsection{Proofs of \Cref{sec:geom-ergod-hmc}}

\subsubsection{Proof of \Cref{propo:geo_drift_MH}}
\label{sec:proof-crefpr}
By construction   \eqref{eq:def_kenel_MH}, for all $\q \in \rset^d$, we have
\begin{align}
&\Pker \Vgeo (\q) - \Vgeo(\q) = \int_{\rset^{2d}} \defEns{\Vgeo(\projq(z)) - \Vgeo(\q)} \alphagen(\q,z) \Kker(\q, \rmd z )  \\
& \qquad = \Kker\Vgeo(\q)-\Vgeo(\q) +\int_{\rset^{2d}} \defEns{\Vgeo(\projq(z)) - \Vgeo(\q)} \defEns{\alphagen(\q,z)-1} \Kker(\q, \rmd z ) \eqsp.
\end{align}
Using \eqref{eq:assum:geo_ergo_1}, this implies for all $\q \in \rset^d$,
\begin{equation}
\label{eq:proof_geo_drift_MH_1}
\Pker \Vgeo (\q)  \leq \lambdageo \Vgeo(\q) + b
+ \int_{\rset^{2d}} \defEns{\Vgeo(\projq(z)) - \Vgeo(\q)} \defEns{\alphagen(\q,z)-1} \Kker(\q, \rmd z ) \eqsp.
\end{equation}

Note that by definition \eqref{eq:def_rej_ballV} of $\rejectregion(\q)$ and $\ballV(\q)$
\begin{align}
&\int_{\rset^{2d}} \defEns{\Vgeo(\projq(z)) - \Vgeo(\q)} \defEns{\alphagen(\q,z)-1} \Kker(\q, \rmd z )
\\ &  \qquad  \qquad  \qquad  \leq    \int_{\rejectregion(\q) \cap \ballV(\q) } \defEns{ \Vgeo(\q)-\Vgeo(\projq(z))}  \Kker(\q, \rmd z ) \eqsp.
\end{align}
Therefore by \eqref{eq:assum:geo_ergo_2}, we get
\begin{equation}
 \lim_{M\to \plusinfty} \sup_{\set{\q \in \rset^d}{\Vgeo(\q) \geq M}}  \int_{\rset^{2d}} \left\{\Vgeo(\projq(z))/\Vgeo(\q) -1 \right\} \left\{ \alphagen(\q,z)-1 \right\} \Kker(\q, \rmd z ) \leq 0 \eqsp.
\end{equation}
The proof then follows from combining this result and \eqref{eq:proof_geo_drift_MH_1} since they imply
\begin{equation}
   \lim_{M\to \plusinfty} \sup_{\set{\q \in \rset^d}{\Vgeo(\q) \geq M}}  \Pker \Vgeo (\q) / \Vgeo(\q) \leq \lambda  \eqsp.
\end{equation}

\subsubsection{Proof of \Cref{lem:drift_uhmc}}
\label{sec:proof-crefl-2}

Let $a \in \rset_+^*$. Under \Cref{assum:regOne}$(m-1)$ with $m \in \ocint{1,2}$,   \Cref{lem:bound_first_iterate_leapfrog_a} shows that, for all $\q_0 \in \rset^d$,
$\p \mapsto \Phiverletq[h][T](\q_0,\p)$ is Lipschitz, with a Lipschitz constant $L_{h,T} \in \rset_+$
\begin{equation}
\label{eq:definition-lipshitz}
L_{h,T} \eqdef \defEns{1+h \constzero^{1/2} \vartheta_1(h \constzero^{1/2})}^{T} \eqsp.
\end{equation}
Therefore by the log-Sobolev inequality \cite[Proposition 5.5.1, (5.4.1)]{bakry:gentil:ledoux:2014} and \eqref{eq:def_Pker_proposition_double}, we get for all $\q_0 \in \rset^d$
\begin{equation}
\PkerhmcD[h][T] \Vdrifta[\a](\q_0) \leq \exp\parenthese{(aL_{h,T})^2/2 + a \Eproof[h][T](\q_0)} \eqsp,
\end{equation}
with
\begin{equation}
\Eproof(\q_0) = (2\uppi)^{-d/2} \int_{\rset^d} \norm{\Phiverletq[h][T](\q_0,\p)} \rme^{-\norm{\p}^2/2} \rmd \p \eqsp.
\end{equation}
Set $p_0 \in \rset^d$.
Denote for all $k \in \{0,\ldots,T\}$, $q_k =
\Phiverletq[h][k](\q_0,\p_0)$ and consider the following decomposition given by  \eqref{eq:qk}:
\begin{equation}
\label{eq:drift_uhmc_3}
\norm{\q_T}^2  =  \norm{\q_0}^2 + \operatorname{A}^{(1)}_{h,T}(\q_0,\p_0) -2h^2 \operatorname{A}^{(2)}_{h,T}(\q_0,\p_0) \eqsp,
\end{equation}
where
\begin{align}
\operatorname{A}^{(1)}_{h,T}(\q_0,\p_0) & = 2Th \ps{\q_0}{ \p_0} + \norm{ Th\p_0-(Th^2/2) \nabla \F(\q_0)-h^2 \sum_{i=1}^{T-1}(T-i)\nabla \F (\q_i)}^2 \\
\operatorname{A}^{(2)}_{h,T}(\q_0, \p_0) & = \ps{\q_0}{ (T/2) \nabla \F(\q_0)+ \sum_{i=1}^{T-1}(T-i)\nabla \F (\q_i)} \eqsp.
\end{align}
Jensen's inequality shows that, for all $\q_0 \in \rset^d$,
\[
\Eproof[h][T](\q_0) \leq \left( \norm{\q_0}^2 + \bar{\operatorname{A}}^{(1)}_{h,T}(\q_0) - 2 h^2 \bar{\operatorname{A}}^{(2)}_{h,T}(\q_0) \right)^{1/2} \eqsp,
\]
where we have set $\bar{\operatorname{A}}_{h,T}^{(i)}(\q_0)= (2\uppi)^{-d/2} \int_{\rset^d} \operatorname{A}_{h,T}^{(i)}(\q_0,\p) \rme^{-\norm{\p}^2/2} \rmd \p$, $i=1,2$.
Therefore to conclude the proof, it is sufficient to show that
\begin{equation}
\label{eq:drift_uhmc_minus1}
\limsup_{\norm{\q_0} \to \plusinfty} \defEnsLigne{\Eproof(\q_0) - \norm{\q_0}} = - \infty.
\end{equation}
\begin{enumerate}[label=(\alph*),leftmargin=0cm,itemindent=0.5cm,labelwidth=1.2\itemindent,labelsep=0cm,align=left]
\item Consider the case $m \in \ooint{1,2}$. Using \Cref{assum:regOne}$(m-1)$ and  \Cref{lem:bound_first_iterate_leapfrog_b}-\ref{lem:bound_first_iterate_leapfrog_1}, we get that  there exists a constant $C_0 \geq 0$ such that for all $\p_0,\q_0 \in \rset^d$ and $i \in \{1,\dots,T-1\}$,
  \begin{equation}
    \label{eq:drift_nabla_U_q_i}
    \norm{\nabla \F(\q_i)} \leq C_0 \{1 + \norm{\p_0} + \norm{\q_0}^{m-1}\}
  \end{equation}
  which implies that
\begin{equation}
\label{eq:bound-A-1}
|\bar{A}^{(1)}_{h,T}(\q_0)| \leq C_1 \{1 + \norm{\q_0}^{2(m-1)} \} \eqsp,
\end{equation}
for some constant $C_1 \geq 0$.  On the other hand, note that for any $q_0,p_0 \in \rset^d$,  $\operatorname{A}^{(2)}_{h,T}(\q_0,\p_0) =  \operatorname{A}^{(2,1)}_{h,T}(\q_0,\p_0) +  \operatorname{A}^{(2,2)}_{h,T}(\q_0,\p_0)$ with
\begin{align}
\label{eq:definition-A-2-1}
\operatorname{A}^{(2,1)}_{h,T}(\q_0,\p_0) &= \frac{T}{2} \ps{\q_0}{ \nabla \F(\q_0)}+\sum_{i=1}^{T-1}(T-i)  \ps{\q_i}{\nabla \F (\q_i)}, \\ 
\label{eq:definition-A-2-2}
\operatorname{A}^{(2,2)}_{h,T} &=- \sum_{i=1}^{T-1}(T-i)  \ps{\q_0-\q_i}{\nabla \F (\q_i)} \eqsp.
\end{align}
Under \Cref{assum:potential:c}$(m)$, for any $q_0, p_0 \in \rset^d$, we have that
\begin{equation}
\label{eq:lower-bound-A-2-1}
\operatorname{A}_{h,T}^{(2,1)}(\q_0,\p_0) \geq \constthree \frac{T}{2} \norm{\q_0}^m - \frac{T (T-1)}{2} \constfour \eqsp.
\end{equation}
Further, by \eqref{eq:drift_nabla_U_q_i} and  \Cref{lem:bound_first_iterate_leapfrog_b}-\ref{lem:bound_first_iterate_leapfrog_1},  there exists  $C_2 \geq 0$, such that for all $\p_0,\q_0 \in \rset^d$,
\begin{equation}
\label{eq:bound-A-2-2}
|\operatorname{A}^{(2,2)}_{h,T}(\q_0,\p_0)| \leq C_2 \{1 + \norm{p_0}^2 + \norm{\q_0}^{2(m-1)} \} \eqsp,
\end{equation}
Combining \eqref{eq:lower-bound-A-2-1} and \eqref{eq:bound-A-2-2}, there exists  $C_3 \geq 0$ such that for any $q_0 \in \rset^d$,
\begin{equation}
\label{eq:bound-A-2}
\bar{\operatorname{A}}^{(2)}(\q_0) \geq \frac{T \constthree}{2} \norm{\q_0}^m - C_3 \{1 + \norm{\q_0}^{2(m-1)} \} \eqsp.
\end{equation}
Combining \eqref{eq:bound-A-1} and \eqref{eq:bound-A-2}, and using that $m < 2$, we finally obtain that \eqref{eq:drift_uhmc_minus1} holds.
\item  By Cauchy-Schwarz and Hölder inequality and since $\nabla U$ satisfies  \Cref{assum:regOne}$(1)$, we have for any $q_0,p_0 \in \rset^d$,
\begin{align}
&\operatorname{A}^{(1)}_{h,T}(\q_0,\p_0)
\leq 2hT \norm{q_0} \norm{p_0} \\
& +3 \parentheseDeux{ h^2 T^2  \norm[2]{\p_0} +  2 h^4 T^4 \constzeroT^2 (1+ \norm[2]{ \q_0}) +   2 h^4 T^2  \constzero^2  \defEns{ \sum_{i=1}^{T-1} \norm{\q_i - q_0}}^2} \eqsp,
\end{align}
which implies using \Cref{lem:bounded_cum_error}, $\vartheta_1(s) \geq 1$ for any $s \geq 0$, and the dominated convergence theorem that
\begin{align}
\label{eq:bound-A-1-m=2}
&\limsup_{\norm{q_0} \to \plusinfty} |\bar{\operatorname{A}}^{(1)}_{h,T}(\q_0)|/ \norm[2]{q_0}  \\
&\qquad \qquad \leq
6 h^4 T^4 \left(  \constzeroT^2 +  \constzero \vartheta_2^2(h) \left[ \{1 + h \constzero^{1/2} \vartheta_1(\constzero^{1/2} h)\}^T -1 \right]^2 \right)  \eqsp.
\end{align}
Similarly using in addition  \Cref{assum:potential:c}($2$), we get that for any $q_0,p_0 \in \rset^d$,
\begin{align}
\operatorname{A}^{(2)}_{h,T}(\q_0,\p_0)
&= \ps{\q_0}{ (T^2/2)  \nabla \F(\q_0)+ \sum_{i=1}^{T-1}(T-i)\{\nabla \F (\q_i) - \nabla U(q_0)\}}  \\
&\geq  (T^2/2) \{\constthree  \norm[2]{q_0} - \constfour\} - T \constzero \norm{q_0} \sum_{i=1}^{T-1}\norm{ q_i - q_0 }  \eqsp.
\end{align}
Then, \Cref{lem:bounded_cum_error} and the Fatou Lemma imply that
\begin{align}
& \liminf_{\norm{q_0} \to \plusinfty} h^2 \bar{\operatorname{A}}^{(2)}_{h,T}(\q_0)/\norm{q_0}^2  \\
 & \qquad \qquad  \qquad  \geq
h^2 T^2  \left( \constthree/2 - \constzero^{1/2} \vartheta_2(h) [(1+ h \constzero^{1/2}\vartheta_1(h \constzero^{1/2}))^T-1]\right)  \eqsp.
\end{align}
Therefore, for all  $h > 0$, and $T \in \nset^*$, one obtains
\[
\limsup_{\norm{q_0} \to \plusinfty} \{ \Eproof[h][T](\q_0) \}^2  /\norm[2]{q_0}   \leq 1 -  T^2 h^2 (\constthree -  \Theta(hT))  \eqsp,
\]
where $\Theta$ is defined in \eqref{eq:definition-function-C}. The proof follows.
\end{enumerate}

\subsubsection{Proof of \Cref{le:convex}}
\label{sec:proof-crefle:convex}

Note that condition~\ref{le:convex:a}  implies that
\begin{equation}
  \label{eq:le:convex:b:conseq}
  \inf_{\set{\q}{\norm{\q}=\Rexp}} \F_0(\q) >0 \eqsp.
\end{equation}
Condition \Cref{assum:potential}-\ref{assum:potential:a} follows from \ref{le:convex:b} using that, by \ref{le:convex:a},
 for all $\q \in \rset^d$, $\norm{q} \geq \Rexp$
\[
\F_0(q) =  (\norm{q}/\Rexp)^{\m}\F_0 (\Rexp  q / \norm{q} ) 
\]

  In addition, \Cref{assum:potential:c} is also
  easy to check using the Euler's homogeneous function theorem that
$\ps{\nabla \F_0(\q)}{
 \q}= \m \F_0(\q)$ for all $\q \in \rset^d$, $\norm{\q} \geq \Rexp$.

  We show below that \Cref{assum:potential}-\ref{assum:potential:b}
  holds. First, since $\lim_{\norm{\q} \to \plusinfty} \F_0(\q) = \plusinfty$ and $\F_0$ is continuous  for all $K \geq 0$, $\lset_K=\{ \q \in \rset^d \ ; \ \F_0(\q) \leq K\}$
 is compact. Besides, using \eqref{eq:le:convex:b:conseq} and that $\F_0$ is continuous,  we can define
  \begin{equation}
    \label{eq:def_M_le:convex}
M =  \sup_{\q \in \ball{0}{\Rexp}} \F_0(\q) +1 \in \ooint{1, \plusinfty}\eqsp,
  \end{equation}
and for all $\q \not \in \lset_M$,
\begin{equation}
  \label{eq:deftq_le:convex_0}
 t_q = \sup \defEns{ t \in \ccint{0,1} \ ; \ \F_0(t \q ) =M } \eqsp,
\end{equation}
 which satisfies
  \begin{equation}
    \label{eq:deftq_le:convex}
    \F_0(t_{\q} \q) = M > \sup_{\x \in \ball{0}{\Rexp}} \F_0(\x) \eqsp, \eqsp t_{\q} q \in \partial \lset_M \eqsp \text{ and } \eqsp t_{\q}
  \norm{\q} > \Rexp \eqsp.
  \end{equation}
 Finally using \ref{le:convex:a}, we get that the set
  $\lset_M$  is
  convex.

  To show \Cref{assum:potential}-\ref{assum:potential:b}, we check first
  that it is sufficient to prove that
  \begin{equation}
D^2\F_0(\x)\defEnsLigne{\nabla
    \F_0(\x)\otimes \nabla \F_0(\x)}>0 \text{ for any } \x \in \partial
  \lset_M \eqsp.
  \end{equation}
 Indeed note that if this statement holds, since $\F \in C^2(\rset^d)$
  and $\partial \lset_M$ is compact, we have
  \begin{equation}
    \label{eq:le:convex:1}
\varepsilon =  \inf_{x \in \partial \lset_M} D^2\F_0(\x)\defEnsLigne{\nabla
    \F_0(\x)\otimes \nabla \F_0(\x)}>0 \eqsp.
  \end{equation}
Let now $\q \not \in \lset_M$ and $t_{\q}$ defined by \eqref{eq:deftq_le:convex_0}.
Since by \ref{le:convex:a}, for all $u
  \geq 1$ and $z \in \rset^d$, $\norm{z} \geq \Rexp$, $\F_0(uz) = u^\m \F_0(z)$,
  differentiating with respect to $z$, we get $\nabla \F_0(uz) = u^{\m-1}
  \nabla \F_0(z)$ and $D^2 \F_0(uz) = u^{\m-2} D^2\F_0(z)$. Therefore by \eqref{eq:deftq_le:convex}, we get
 \begin{equation}
    \label{eq:le:convex:2}
   D^2\F_0(\q)\defEns{\nabla
  \F_0(\q)\otimes \nabla \F_0(\q)} = t_{\q}^{4-3m} D^2\F_0(t_{\q} \q)\defEns{\nabla
  \F_0(t_{\q} \q )\otimes \nabla  \F_0(t_{\q} \q)} \eqsp.
 \end{equation}
Using \eqref{eq:deftq_le:convex} again and since $\partial \lset_M$ is compact, we get that
there exists $R_2 \geq 0$ such that $t_{\q} \norm{q} \in \ccint{\Rexp,R_2}$. Hence by \eqref{eq:le:convex:2}, we have
\begin{equation}
   D^2\F_0(\q)\defEns{\nabla
  \F_0(\q)\otimes \nabla \F_0(\q)} \geq \varepsilon \norm{q}^{3m-4} \min\parentheseDeux{\Rexp^{4-3m},R_2^{4-3m}} \eqsp.
\end{equation}
Thus
 \Cref{assum:potential}-\ref{assum:potential:b} holds for $\F_0$.
Finally \ref{le:convex:b}
 implies that  the function $\F=\F_0+G$ satisfies
 \Cref{assum:potential}-\ref{assum:potential:b} as well.

 Let $ \x \in \partial \lset_M$, we now show that $D^2\F_0(\x)\defEns{\nabla
 \F_0(\x)\otimes \nabla \F_0(\x)}>0$. By Euler's homogeneous function theorem and since $M \geq 1$, we have that
 $\abs{\ps{\nabla \F_0(\x)}{\x}}\geq \m >0$.  Denote by $\Pi$ the tangent hyperplane
 of $\partial\lset_M$ at $\x$, defined by $\Pi = \{ \q \in \rset^d \, :
 \, \ps{\nabla \F_0(x) }{ \x-\q} = 0\}$.  Since $\lset_M$ is convex, for all $\q \in
 \lset_M$ and $t \in \ccint{0,1}$, $ t^{-1}( \F_0(t\q +(1-t)\x)
 -\F_0(\x))\leq 0$. So taking the limit as $t$ goes to $0$, we get that
 $\ps{\nabla \F_0(\x) }{ \q-\x} \leq 0$. Therefore, $\lset_M$ is
 contained in the half-space $\Pi^- = \{ q \in \rset^d \, ; \, \ps{\nabla
 \F_0(\x) }{ \q-\x} \leq 0 \}$.

Define the  $\m$-homogeneous
 function $\tilde{\F} : \rset^d \to \rset_+$  for all $\q \in \rset^d$  by
 \begin{equation}
\label{eq:def:tilde_F}
   \tilde{\F}(\q) = M \abs{\frac{\ps{\q}{ \nabla \F_0(\x)}}{\ps{\x}{ \nabla \F_0(\x)}}}^\m \eqsp.
 \end{equation}
 Since $\F_0(x) = M$, by \eqref{eq:def_M_le:convex}, $\norm{x} > \Rexp$
 and therefore there exists $\epsilon_0 \in \rset_+^*$ such that
 \begin{equation}
\label{eq:inclusion_ball}
   \ball{x}{\epsilon_0} \subset \rset^d \setminus \ball{0}{\Rexp}\eqsp.
 \end{equation}
 We
 now show that $\tilde{\F}(\q) \leq \F_0(\q)$ for all $\q \in
 \ball{\x}{\epsilon}$ with
 \begin{equation}
\epsilon = 2^{-1}\min\parentheseDeux{\epsilon_0, \{\ps{\x}{ \nabla \F_0(\x)}\}/\norm[2]{\nabla \F_0(\x)}} \eqsp.
 \end{equation}
  First consider $\q \in \Pi$. We next argue by contradiction that
 \begin{equation}
\label{eq:bound_Pi}
   \F_0(\q) \geq M = \tilde{\F}(\q)    \eqsp.
 \end{equation}
Indeed assume that  $\F_0(\q) < M$. Then by continuity of $\F_0$, we get that $\q \in \interior{\lset_{M}}$. But since $\lset_{M} \subset \Pi^-$, we get $\q \in \interior{(\Pi^-)}$ which is impossible since $\q \in \Pi = \boundary{\Pi^-} = \clos{\Pi^-} \setminus \interior{(\Pi^-)}$.

Let $\q \in
 \ball{\x}{\epsilon}$. Note that  $  \q= \x + \norm[-2]{\nabla \F_0(\x)}\ps{\q-\x}{\nabla \F_0 (\x)}\nabla \F_0(\x) + z$,
where  $z \in \rset^d$ is orthogonal to $\nabla \F_0(\x)$. Define
\begin{equation}
  u = \frac{\ps{\x }{ \nabla \F_0(\x)}}{\ps{\q}{ \nabla \F_0(\x)}}\eqsp.
\end{equation}
Then $u\q \in \Pi$ and by \eqref{eq:bound_Pi}, $\F_0(u \q) \geq M$. If $u \geq 1$, using \ref{le:convex:a} and \eqref{eq:def:tilde_F}, we get
\begin{equation}
\label{eq:6}
  \F_0(\q) \geq  u^{-\m}M = \tilde{\F}(\q) \eqsp.
\end{equation}
In turn, if $u < 1$, since $\norm{\q- \x} \leq \epsilon_0$, by \eqref{eq:inclusion_ball} and \ref{le:convex:a}, $\F_0(\q) = u^{-1} \F_0(u \q)$ and \eqref{eq:6} still holds.

Consider the three times differentiable functions $\phi$ and $\tilde{\phi}$
defined for all $v \in \rset$ by
$$
\phi(v)=\F_0(\x+v \nabla \F_0(\x))\quad\textrm{and}\quad \tilde{\phi}(v)= \tilde{\F}(\x+v \nabla \F_0(\x)) \eqsp.
$$
First, since for all $\q \in \ball{\x}{\epsilon}$, $\F_0(\q) \geq \tilde{\F}(\q)$, we have
\begin{equation}
  \label{eq:bound_f_tildef}
  \phi(v) \geq \tilde{\phi}(v) \eqsp, \text{for all $v \in \ccint{-\epsilon/\norm{\nabla \F_0(\x)},\epsilon/\norm{\nabla \F_0(\x)}}$}\eqsp.
\end{equation}
Moreover, by definition $\F_0(\x) = \tilde{\F}(\x)$ and $\nabla \tilde{\F}(\x)$
is colinear to $ \nabla \F_0(\x)$. Using  Euler's homogeneous function theorem for $\F_0$
and $\tilde{\F}$, we get that $\nabla \tilde{\F}(\x) = \nabla
\F_0(\x)$. Therefore $\phi(0) = \tilde{\phi}(0)$, $\phi'(0) = \tilde{\phi}'(0)$. Combining these equalities, \eqref{eq:bound_f_tildef} and using a Taylor expansion around $0$ of order $2$ with exact remainder for $\phi$ and $\tilde{\phi}$ shows that necessary
\begin{equation}
  D^2\F_0(\x)\defEns{\nabla
  \F_0(\x)\otimes \nabla \F_0(\x)} =   \phi''(0) \geq \tilde{\phi}''(0) >0 \eqsp,
\end{equation}
which concludes the proof.

\subsubsection{Proof of \Cref  {propo:accept}}
\label{sec:proof-crefth}
We preface the proof by several technical preliminary Lemmas.

\begin{lemma}
\label{lem:grad_Lip_F}
Assume \Cref{assum:potential}($\m$)-\ref{assum:potential:a} for some $m\in \ocint{1,2}$.
Then, for all $q,x \in \rset^d$,
$\norm{\nabla \F(\q) - \nabla \F(\x)} \leq \constone  \norm{\q-\x}$
and $\norm{\nabla \F(\q) - \nabla \F(\x)} \leq \constone (m-1)^{-1} \norm{\q-\x}^{m-1}$.
In particular, \Cref{assum:regOne}($m-1$) holds with $\constzero= \constone$ and $\constzeroT= \constone (m-1)^{-1} \vee \norm{\nabla \F(0)}$.
\end{lemma}

\begin{proof}
First by \Cref{assum:potential}($m$)-\ref{assum:potential:a}, we get for all $\q,\x \in \rset^d$,
\begin{align}
\nonumber
  \norm{\nabla \F(\q) - \nabla \F(\x)}
&= \norm{\int_{0}^1 \nabla^2 \F(\x +t(\q-\x)) \defEns{\q-\x} \rmd t } \\
\label{eq:lem_grad_Lip_F_eq_0}
&\leq \constone  \norm{\q-\x} \int_{0}^1 \defEns{1+\norm{\x +t(\q-\x)}}^{\m-2} \rmd t   \eqsp.
\end{align}
Therefore, for all $\q,\x \in \rset^d$, we get $ \normLigne{\nabla \F(\q) - \nabla \F(\x)} \leq \constone \normLigne{\q-\x}$. For all $\q, \x \in \rset^d$, since $m \in \ocint{1,2}$, we have
\begin{align}
&\int_{0}^1 \defEns{1+\norm{\x +t(\q-\x)}}^{\m-2} \rmd t  \leq  \int_{0}^1\defEns{ 1+\abs{\norm{\x} - t \norm{\q-\x}}}^{m-2}  \rmd t  \\
& \leq \int_{0}^{1 \wedge \frac{\norm{\x}}{\norm{\q-\x}}}\defEns{1+  \norm{\x} - t \norm{\q-\x}}^{m-2}  \rmd t
+ \int_{1 \wedge \frac{\norm{\x}}{\norm{\q-\x}}}^1 \defEns{1+ t \norm{\q-\x}-\norm{\x}}^{m-2}  \rmd t \\
&\leq (m-1)^{-1} \norm{\q-\x}^{m-2} \eqsp.
\end{align}
Plugging this result in \eqref{eq:lem_grad_Lip_F_eq_0} concludes the proof.

\end{proof}

\begin{lemma}
  \label{lem:prepa_bound_diff_ham}
Assume  \Cref{assum:regOne}$(\beta)$ for $\beta \in \ocint{0,1}$. Let $\gamma \in \ooint{0,\beta}$.
\begin{enumerate}[label=(\roman*)]
\item   \label{lem:prepa_bound_diff_ham_1}
If $ \beta \in \ooint{0,1}$, for any $T \in \nsets$ and  $h_0 \in \rset^*_+$, there exist $\kappa \in \rset^*_+$ and $R \in \rset_+$ such that for all $h \in \ocint{0,h_0}$,  $\q_0,p_0 \in \rset^d$ satisfying $ \norm{p_0} \leq
\norm{\q_0}^{\gamma}$ and $\norm{\q_0} \geq R$, and $i,j,k \in
\{0,\ldots,T\}$,
\begin{align}
\label{eq:bound_iterate_q_2_prood_diff_ham_eq}
&\norm{q_0} \leq \kappa  \norm{\Phiverletq[h][k](q_0,p_0)}\eqsp, \\ &\norm{\Phiverletq[h][i](q_0,p_0)-\Phiverletq[h][j](q_0,p_0)} \leq \kappa h \norm{\Phiverletq[h][k](q_0,p_0)}^{\beta}  \eqsp,
\end{align}
  where $\Phiverletq[h][\ell]$ are defined by \eqref{eq:def_Phiverletq} for $\ell \in \nset^*$.
\item \label{lem:prepa_bound_diff_ham_2}
If $\beta =1$, then there exist $\kappa, \bar{S} \in \rset_+^*$ (depending only on $\constzero$ and $\constzeroT$) such that for any $T \in \nsets$, $h \in \ooint{0,\bar{S}/T}$,  $\q_0,p_0 \in \rset^d$ satisfying $ \norm{p_0} \leq
\norm{\q_0}^{\gamma}$ and $\norm{\q_0} \geq 1$, and $i,j,k \in
\{0,\ldots,T\}$,
\begin{align}
\label{eq:bound_iterate_q_2_prood_diff_ham_eq_2}
&\norm{q_0} \leq 2  \norm{\Phiverletq[h][k](q_0,p_0)} \leq 3 \norm{q_0} \eqsp, \\
&\norm{\Phiverletq[h][i](q_0,p_0)-\Phiverletq[h][j](q_0,p_0)}
\leq \kappa Th
\rme^{(1+\vartheta_1(Th))Th}  \norm{\Phiverletq[h][k](q_0,p_0)} \eqsp,
\end{align}
where $\vartheta_1$ is defined in \eqref{eq:def_vartheta_1}.
\end{enumerate}
\end{lemma}

\begin{proof}
\begin{enumerate}[label=(\roman*),leftmargin=0cm,itemindent=0.5cm,labelwidth=1.2\itemindent,labelsep=0cm,align=left]
\item
Let $T \in \nsets$, $h_0 \in \rset_+^*$ and  $h \in \ocint{0,h_0}$.
Denote for all $k \in \{0,\ldots,T\}$ by $(q_k,p_k) =
  \Phiverlet[h][k](q_0,p_0)$, $q_0, p_0 \in \rset^d$.
  By  \Cref{assum:regOne}$(\beta)$ and  \Cref{lem:bound_first_iterate_leapfrog_b}-\ref{lem:bound_first_iterate_leapfrog_1}, there exist $C \geq 0$ and  $R_1 \geq 0$ such that for all $\q_0,\p_0 \in \rset^d$ satisfying $ \norm{\p_0} \leq \norm{\q_0}^{\gamma}$ and $\norm{\q_0} \geq R_1$, for all $k \in \{0,\ldots,T\}$, we have
\begin{equation}\label{eq:bound_iterate_q_1_prood_diff_ham_0}
 \norm{\q_k-\q_0} \leq C h \norm{\q_0}^{\m-1} \eqsp.
\end{equation}
Then since $m<2$, there exists $R_2 \geq R_1$ and $\omega >0$ such that such that for all $\q_0,p_0 \in \rset^d$ satisfying $ \norm{p_0} \leq \norm{\q_0}^{\gamma}$  and $\norm{\q_0} \geq R_2$, for all $k \in \{0,\ldots,T\}$,
\begin{equation}
  \norm{\q_0} \leq \omega \norm{\q_k} \eqsp.
\end{equation}
In addition, using this inequality and
\eqref{eq:bound_iterate_q_1_prood_diff_ham_0} again, we get that for
all $\q_0,p_0 \in \rset^d$ satisfying $ \norm{p_0} \leq
\norm{\q_0}^{\gamma}$ and $\norm{\q_0} \geq R_2$, for all $i,j,k \in
\{0,\ldots,T\}$,
\begin{equation}
\label{eq:bound_iterate_q_2_prood_diff_ham_2}
\norm{\q_{i}-\q_{j}} \leq 2C h \norm{\q_0}^{\m-1} \leq 2 Ch \omega^{\m-1}\norm{\q_k}^{\m-1}  \eqsp.
\end{equation}
\item Let $T \in \nsets$, $h \in \rset_+^*$.
Denote for all $k \in \{0,\ldots,T\}$ by $(q_k,p_k) =
  \Phiverlet[h][k](q_0,p_0)$, $q_0, p_0 \in \rset^d$.
 By  \Cref{assum:regOne}$(1)$ and \Cref{lem:bound_first_iterate_leapfrog_b}-\ref{lem:bound_first_iterate_leapfrog_b_2}, $\vartheta(s) \geq 1$ for any $s \geq 0$, we get that  for all $\q_0, \p_0 \in \rset^{2d}$ satisfying $\norm{\p_0} \leq \norm{\q_0}^\gamma$ and $k \in \{0,\ldots,T\}$,
\begin{equation}\label{eq:bound_iterate_q_1_prood_diff_ham_0_m_2}
 \norm{\q_k-\q_0} \leq \constzero^{-1/2}\left\{(1 + h \constzero^{1/2} \vartheta_1(h \constzero^{1/2}))^{k+1}-1\right\} \left\{ \vartheta_2(h) \vee \vartheta_3(h) \right\} (1 + \norm{\q_0}) \eqsp,
\end{equation}
where $\vartheta_1$, $\vartheta_2$ and $\vartheta_3$ are defined in \eqref{eq:def_vartheta_1} and  \eqref{eq:definition-vartheta-2} respectively.

Therefore,  there exists $\bar{S} >0$ (depending only on $\constzero$ and $\constzeroT$) such that for any $T \in \nsets$ and $h \in \ooint{0,\bar{S}/T}$, for any $q_0,p_0 \in \rset^d$ satisfying $ \norm{\p_0} \leq \norm{\q_0}^{\gamma}$ and $\norm{\q_0} \geq  1 $, $ \normLigne{\Phiverletq[h][k](q_0,p_0)-\q_0} \leq \norm{q_0}/2$ for any $k \in \{0,\ldots,T\}$. As a result, for any $T \in \nsets$ and $h \in \ooint{0,\bar{S}/T}$, for any $q_0,p_0 \in \rset^d$ satisfying $ \norm{\p_0} \leq \norm{\q_0}^{\gamma}$ and $\norm{\q_0} \geq  1 $,  for any $k \in \{0,\ldots,T\}$,
\begin{equation}
  \norm{\q_0} \leq 2 \norm{\Phiverletq[h][k](q_0,p_0)} \leq 3 \norm{q_0} \eqsp.
\end{equation}
In addition, using this inequality and
\eqref{eq:bound_iterate_q_1_prood_diff_ham_0_m_2} again, we get that there exists $C \geq 1$ (depending only on $\constzero$ and $\constzeroT$) such that  for any $T \in \nsets$ and $h \in \ooint{0,\bar{S}/T}$, setting $S= hT$, and  for
all $\q_0,p_0 \in \rset^d$ satisfying $ \norm{p_0} \leq
\norm{\q_0}^{\gamma}$ and $\norm{\q_0} \geq 1$, for all $i,j,k \in
\{0,\ldots,T\}$,
\begin{align}
\norm{\Phiverletq[h][i](q_0,p_0) - \Phiverletq[h][j](q_0,p_0)} 
&\leq 2C S \rme^{(1+\vartheta_1(S))S} \norm{\q_0} \\
&\leq 4 C S \rme^{(1+\vartheta_1(S))S}\norm{\Phiverletq[h][k](q_0,p_0)}  \eqsp.
\end{align}
\end{enumerate}
\end{proof}

\begin{lemma}
\label{lem:variation_assum_hessian}
Assume \Cref{assum:potential}($\m$) for some $m \in \ocint{1,2}$.
Then there exist $\delta \in \ooint{0,1}$, $R_0 \in \rset_+$ and  $\BB_0 \in \rset_+^*$ such that for all $\q,\x,z \in \rset^d$, with
\begin{equation}
\label{eq:hyp_variation_assum_hessian}
\norm{\q} \geq R_0 \eqsp, \qquad \text{ and } \quad \max\parenthese{\norm{\q-\x},\norm{\q-z}} \leq \delta \norm{\q}  \eqsp,
\end{equation}
we have
\begin{equation}
D^2 \F(\q) \defEns{\nabla \F(\x) \otimes \nabla \F(z)} \geq \BB_0 \norm{\q}^{3\m-4}  \eqsp.
\end{equation}
\end{lemma}
\begin{proof}
Under \Cref{assum:potential}($m$), using \Cref{lem:grad_Lip_F}, it can be easily checked that there
exists $C_U \geq 0$ (depending only on $\constone$ and $m$) such that for all  $\q,\x,z \in \rset^d$ satisfying  \eqref{eq:hyp_variation_assum_hessian},
for $\delta \in \ooint{0,1}$ and $R_0 \geq \rhtwo$,
\begin{equation}
D^2 \F(\q) \defEns{\nabla \F(\x) \otimes \nabla \F(z)}  \geq \consttwo  \norm{\q}^{3\m-4} - C_U \{1+ \delta^{m-1} \norm{\q}^{3m-4}\}\eqsp.
\end{equation}
The proof is concluded by taking $\delta$ sufficiently small and $R_0$ sufficiently large.
\end{proof}

 \begin{lemma}
 \label{lem:diff_hamiltonian_taylor_exp}
 Assume that $\F$ is twice continuously differentiable. Then for all $q_0,p_0 \in \rset^d$ and $h \in \rset^*_+$, the following identity holds
 \begin{align}
& \Ham \circ \Phiverlet[h][1](q_0,p_0) - \Ham(q_0,p_0) =  h^2\int_0^1D^2\F(q_t)\defEns{p_{0}}^{\otimes 2} (1/2-t) \  \rmd t \\
& + h^3 \int_0^1D^2\F(q_t)\defEns{p_{0}\otimes \nabla \F(q_{0})}(t-1/4) \ \rmd t  \\
&  -\frac{h^4}{4}\int_0^1 D^2\F(q_t)\defEns{\nabla \F(q_{0})}^{\otimes 2} \ t \  \rmd t
+\frac{h^4}{8}\norm{ \int_0^1 \nabla ^2 \F( q_{t}) p_{0} \ \rmd t  }^2\\
&   - \frac{h^5}{8}\ps{\int_0^1\nabla^2\F(q_t)\nabla \F(q_0) \ \rmd t }{\int_0^1 \nabla^2 \F(q_t)p_0 \ \rmd t }
   \\
   &+\frac{h^6}{32}\norm{\int_0^1 \nabla^2 \F( q_{t})\nabla \F (q_0) \ \rmd t }^2
\eqsp,
\end{align}
where  $\Phiverlet[h][1]$ is defined in \eqref{eq:def_Phiverlet}, $(q_1,p_1) = \Phiverlet[h][1](q_0,p_0)$, and $q_t = q_0 +t (q_1-q_0)$ for $t \in \ccint{0,1}$.
\end{lemma}
\begin{proof}
Using the definition of $\Ham(q,p)= \frac{1}{2}\norm{p}^2+ \F(q)$, we get
\begin{equation}
  \Ham(q_1,p_1) - \Ham(q_0,p_0) = (1/2)(\norm{p_1}^2 - \norm{p_0}^2) + \F(q_1) - \F(q_0) \eqsp.
\end{equation}
First, Taylor's formula  with exact remainder  enables us to write
\begin{equation}\label{eq:1}
\F(q_{1})-\F(q_0)= \ps{\nabla \F(q_0)}{ (q_{1}-q_0)}+\int_0^1D^2\F(q_{t})\defEns{q_{1}-q_0 }^{\otimes 2}(1-t) \ \rmd t \eqsp.
\end{equation}
Since  $\nabla \F(q_{1}) = \nabla \F(q_0) + \int_{0}^1 \nabla^2 \F(q_t) \defEns{q_1-q_0} \rmd t $,  we get
 \begin{equation}\label{eq:2}
p_{1}=p_0-\frac{h}{2} \parenthese{\nabla \F(q_{0})+\nabla \F(q_{1})}= p_{0}- h\nabla \F(q_0)-\frac{h}{2}\int_0^1\nabla ^2\F( q_{t}) \defEns{q_{1}-q_{0}} \rmd t \eqsp.
\end{equation}
Using that $q_1 = \Phiverletq[h][1](p_0,q_0)$, with $\Phiverletq[h][1]$ defined by \eqref{eq:def_Phiverletq}, in \eqref{eq:1} and \eqref{eq:2}, we get
\begin{align}
  &\F(q_{1})-\F(q_{0})\\
&=\ps{\nabla \F(q_{0})}{ hp_{0}-(h^2/2)\nabla \F(q_{0})} + \int_0^1D^2\F(q_t)\defEns{q_{1}-q_{0}}^{\otimes 2}(1-t) \rmd t   \eqsp,
\end{align}
and
\begin{align}
&\frac{1}{2}(\norm{p_{1}}^2-\norm{p_{0}}^2)= \frac{h^2}{2}\norm{\nabla \F(q_{0})}^2+ \frac{h^2}{8} \norm{\int_0^1 \nabla^2\F(q_t)\defEns{q_{1}-q_{0} }\rmd t }^2\\
&- h\langle p_{0},\nabla \F(q_{0})\rangle -(h/2)\int_0^1 D^2\F(q_t)\defEns{p_{0}\otimes (q_{1}-q_{0})} \rmd t \\
&+ (h^2/2)\int_0^1D^2\F(q_t) \defEns{\nabla \F(q_{0})\otimes (q_{1}-q_{0})} \rmd t \eqsp.
\end{align}
Summing these equalities up  and observing appropriate cancellations yields
\begin{flalign}
\nonumber
&H(q_{1},p_{1})-H(q_{0},p_{0})=\int_0^1D^2\F(q_t)\defEns{q_{1}-q_{0}}^{\otimes 2}(1-t) \rmd t\\
\nonumber
& - (h/2)\int_0^1D^2\F(q_t)\defEns{p_{0}\otimes (q_{1}-q_{0})} \rmd t +   (h^2/8) \norm{ \int_0^1 \nabla ^2\F(q_t)\defEns{q_{1}-q_{0}} \rmd t}^2\\
\nonumber
&+ (h^2/2)\int_0^1D^2\F(q_t)\defEns{\nabla \F(q_{0})\otimes (q_{1}-q_{0})} \rmd t
\end{flalign}
\vspace{-0.8cm}
\begin{equation}
\label{eq:decomp_lem_hamil}
= I_1+I_2+I_3+I_4 \eqsp.
\end{equation}
By using $q_1 = \Phiverletq[h][1](p_0,q_0)$ again in the definition of each $I_j$ we obtain successively
\begin{align}
I_1&=h^2\int_0^1D^2\F( q_t)\defEns{p_0}^{\otimes 2} (1-t) \rmd t- h^3 \int_0^1D^2\F(q_t)\defEns{p_{0}\otimes \nabla \F(q_{0})}(1-t)\rmd t\\
&\qquad \qquad \qquad +(h^4/4)\int_0^1D^2\F(q_t) \defEns{\nabla \F(q_{0})}^{\otimes 2}(1-t)\rmd t \eqsp, \\
I_2&= - (h^2/2)\int_0^1D^2\F(q_t) \defEns{p_{0}}^{\otimes 2} \rmd t + (h^3/4) \int_0^1D^2\F(q_t)\defEns{p_{0}\otimes \nabla   \F(q_{0})} \rmd t \eqsp,\\
I_3&= (h^4/8) \norm{ \int_0^1\nabla^2 \F(q_t)p_{0} \ \rmd t }^2+(h^6/32) \norm{\int_0^1\nabla^2\F(q_t)\nabla \F(q_{0})\  \rmd t }^2
\\
& \qquad \qquad \qquad  - (h^5/8)\ps{\int_0^1\nabla^2\F(q_t)\nabla \F(q_0) \ \rmd t }{\int_0^1 \nabla^2 \F(q_t)p_0 \ \rmd t } \eqsp.
\end{align}
and
\begin{align}
  I_4& = (h^3/2)\int_0^1D^2\F(q_t)\defEns{\nabla \F(q_{0})\otimes p_{0}}\rmd t \\
  & \qquad \qquad \qquad \qquad  \qquad \qquad - (h^4/4)\int_0^1D^2\F(q_t)\defEns{\nabla \F(q_{0})}^{\otimes 2}\rmd t \eqsp,
\end{align}
Gathering all these equalities  in \eqref{eq:decomp_lem_hamil} concludes the proof.
 \end{proof}

\begin{proof}[Proof of \Cref  {propo:accept}]
Let $\gamma \in \ooint{0,m-1}$, $T \in \nsets$, $h_0 \in \rset_+^*$ and  $h \in \ocint{0,h_0}$.
Denote for all $k \in \{0,\ldots,T\}$ by $(q_k,p_k) =
  \Phiverlet[h][k](q_0,p_0)$, $q_0, p_0 \in \rset^d$.
  For all $q_0,p_0 \in \rset^d$, consider the following decomposition
\begin{equation}
  \label{eq:diff_ham_decompo}
H(p_T,q_T)-H(p_0,q_0)=\sum_{k=0}^{T-1}\defEns{H(p_{k+1},q_{k+1})-H(p_{k},q_{k})} \eqsp.
\end{equation}
We show that each term in the sum in the right hand side of this equation is nonpositive if $\norm{\q_0}$ is large enough and $\norm{p_0} \leq \norm{q_0}^{\gamma}$.
By \Cref{lem:diff_hamiltonian_taylor_exp}, we have
\begin{equation}
\label{eq:diff_ham_k}
H(q_{k+1},p_{k+1})-H(q_{k},p_{k})
= -(h^4/4)A_k+h^2 B_k+ h^3C_k+(h^4/8) D_k \eqsp,
\end{equation}
where, setting $q_{t,k}= q_k + t(q_{k+1}-q_k)$ for $t \in \ccint{0,1}$,
\begin{align}
  A_k &=\int_0^1 D^2\F(q_{t,k})\defEns{\nabla \F(q_{k})}^{\otimes 2} \ t \  \rmd t \\
B_k& = \int_0^1D^2\F(q_{t,k})\defEns{p_{k}}^{\otimes 2} (1/2-t) \  \rmd t \\
C_k & = \int_0^1D^2\F(q_{t,k})\defEns{p_{k}\otimes \nabla \F(q_{k})}(t-1/4) \ \rmd t \\
D_k & = \norm{ \int_0^1 \nabla ^2 \F( q_{t,k}) p_{k} \ \rmd t  }^2 +(h^2/4)\norm{\int_0^1 \nabla^2 \F( q_{t,k})\nabla \F (q_{k}) \ \rmd t }^2 \\
& \qquad \qquad -  h \ps{\int_0^1\nabla^2\F(q_{t,k})\nabla \F(q_{k}) \ \rmd t }{\int_0^1 \nabla^2 \F(q_{t,k})p_{k} \ \rmd t }
\end{align}
Since $q_{t,k}-q_k= -(t h^2/2) \nabla \F(q_k) + th p_k$ and $\int_{0}^1(1/2-t) \, \rmd t = 0$, we have
for all $q_0,p_0 \in \rset^d$,
\begin{align}
B_k &= \int_{0}^1 \int_0^1 D^3 \F(q_{k} +s(q_{t,k} - q_{k})) \defEns{p_{k}^{\otimes 2} \otimes (q_{t,k}-q_{k})} (1/2-t)\rmd s \ \rmd t 
  \\
  \label{eq:definition-B_k}
    &= h B_{k,1} - h^2 B_{k,2}
\end{align}
where
\begin{align}
\label{eq:definition-B_k-1}
B_{k,1} &= \int_{0}^1 \int_0^1 D^3 \F(q_{k} +s(q_{t,k} - q_{k})) \defEns{p_{k}}^{\otimes 3}t (1/2-t)\rmd s \ \rmd t \\
\label{eq:definition-B_k-2}
B_{k,2} &=  -\frac{1}{2} \int_{0}^1 \int_0^1 D^3 \F(q_{k} +s(q_{t,k} - q_{k})) \defEns{p_{k}^{\otimes 2} \otimes \nabla \F(q_k) } t (1/2-t)\rmd s \ \rmd t \eqsp.
\end{align}
Consider now the term $C_k$ in \eqref{eq:diff_ham_k}. Similarly, using again $\int_0^1 (t- 1/2) \rmd t= 0$ and then \eqref{eq:pk}, we get $C_k = C_{k,1} +  C_{k,2} + C_{k,3}$,  where
\begin{align}
\label{eq:definition-C_k-1}
&C_{k,1} = h \int_{0}^1 \int_{0}^t D^3 \F(q_{k} +s(q_{k,t} - q_{k})) \defEns{p_{k}^{\otimes 2} \otimes \nabla \F(q_k) } t (t-1/2)\rmd s \ \rmd t \\
\nonumber
& -(h^2/2)\int_{0}^1 \int_{0}^t  D^3 \F(q_{k} +s(q_{k,t} - q_{k})) \defEns{p_{k} \otimes \parenthese{\nabla \F(q_k)}^{\otimes 2} } t (t-1/2)\rmd s \ \rmd t \\
\label{eq:definition-C_k-2}
&C_{k,2} =   \int_0^1D^2\F(q_{t,k})\defEns{p_{0}\otimes \nabla \F(q_{k})} \ \rmd t \eqsp, \\
\label{eq:definition-C_k-3}
&C_{k,3} = -h \sum_{i=1}^{k-1}\int_0^1D^2\F(q_{t,k})\defEns{\nabla \F(q_i)\otimes \nabla \F(q_{k})} \ \rmd t
\\
&\qquad \qquad  \qquad \qquad - (h/2)\int_0^1D^2\F(\q_{t,k})\defEns{\parenthese{\nabla \F(\q_0)+\nabla \F (\q_k)} \otimes \nabla \F(\q_{k})} \ \rmd t
\end{align}
We will next estimate each of these terms separately.
Let $\delta \in \ooint{0,1}$ and $\BB_0 \in \rset_+^*$ be the constants defined in \Cref{lem:variation_assum_hessian}.

\begin{enumerate}[label=(\alph*),leftmargin=0cm,itemindent=0.5cm,labelwidth=1.2\itemindent,labelsep=0cm,align=left]
\item
We first consider the case $m\in \ooint{1,2}$.
  By \Cref{lem:grad_Lip_F} and \Cref{lem:bound_first_iterate_leapfrog_b}-\ref{lem:bound_first_iterate_leapfrog_1}, there exist $C \geq 0$ and  $R_1 \geq \rhtwo$ such that for all $\q_0,p_0 \in \rset^d$ satisfying $ \norm{p_0} \leq
\norm{\q_0}^{\gamma}$ and $\norm{\q_0} \geq R_1$, for all $i \in
\{0,\ldots,T\}$,
\begin{equation}
\label{eq:bound_iterate_q_3_prood_diff_ham_3_0}
\begin{aligned}
&\norm{\q_{i}-\q_{0}}\leq (\delta/2) \norm{\q_0} \\ 
&\norm{\p_i - \p_0} \leq C(\norm{\p_0} + h \norm{q_0}^{m-1}) \leq C(\norm{\q_0}^\gamma + h \norm{q_0}^{m-1}) \eqsp.
\end{aligned}
\end{equation}
By  \Cref{lem:variation_assum_hessian}, \Cref{lem:prepa_bound_diff_ham}-\ref{lem:prepa_bound_diff_ham_1} and \eqref{eq:bound_iterate_q_3_prood_diff_ham_3_0}, there exists $R_2 \geq R_1$ such that for all $\q_0,\p_0 \in \rset^d$, $\norm{\q_0} \geq R_2$ and $\norm{\p_0} \leq \norm{\q_0}^{\gamma}$,
we get that
\begin{equation}
\label{eq:bound_A_k}
  \inf_{\norm{p_0} \leq \norm{q_0}^{\gamma}} A_k \geq \BB_0 \norm{q_k}^{3\m-4} \geq
  \BB_0 \{ (1-\delta/2)^{3m-4} \wedge (1+\delta/2)^{3m-4} \} \norm{q_0}^{3\m-4} \eqsp.
\end{equation}
Hence, $\limsup_{\norm{q_0} \to \plusinfty} \sup_{\norm{p_0} \leq \norm{q_0}^{\gamma}}\defEns{A_k/\norm{q_0}^{3\m-4}} > 0$.
We now bound $B_k$. Using \Cref{assum:potential}-\ref{assum:potential:a}, \Cref{lem:grad_Lip_F} and \eqref{eq:bound_iterate_q_3_prood_diff_ham_3_0}, we get by \eqref{eq:definition-B_k} that
\begin{equation}
\label{eq:bound_B_k}
\limsup_{\norm{q_0} \to \plusinfty} \sup_{\norm{p_0} \leq \norm{q_0}^{\gamma}}\defEns{\abs{B_k}/\norm{q_0}^{4\m-6}} < \infty \eqsp.
\end{equation}
Combining \Cref{assum:potential}-\ref{assum:potential:a}, \Cref{lem:grad_Lip_F} and \eqref{eq:bound_iterate_q_3_prood_diff_ham_3_0} again, we get by crude estimate that there exists $C \geq 0$ such that
\begin{equation}
  \label{eq:bound_D_k}
  \limsup_{\norm{q_0} \to \plusinfty} \sup_{\norm{p_0} \leq \norm{q_0}^{\gamma}} \defEns{\abs{D_k}/\norm{q_0}^{4\m-6}} \leq C h^2 \eqsp.
\end{equation}
We finally bound the two terms $C_{k,1}$ and $C_{k,2}$. First,
using the same reasoning as for $B_k$, we get that
\begin{equation}
\begin{aligned}
  &\limsup_{\norm{q_0} \to \plusinfty} \sup_{\norm{p_0} \leq \norm{q_0}^{\gamma}} \defEns{\abs{C_{k,1}}/\norm{q_0}^{4\m-6}} < \infty \eqsp, \\
  &\limsup_{\norm{q_0} \to \plusinfty} \sup_{\norm{p_0} \leq \norm{q_0}^{\gamma}} \defEns{\abs{C_{k,2}}/\norm{q_0}^{2\m-3+\gamma}} < \infty \eqsp.
\end{aligned}
\end{equation}
Arguing like in \eqref{eq:bound_A_k}, we get that
$\limsup_{\norm{q_0} \to \plusinfty}  \sup_{\norm{p_0} \leq \norm{q_0}^{\gamma}} \defEns{C_{k,3}/\norm{q_0}^{3\m-4}} < 0$.
Gathering all these results and  using that  $3m-4 \geq \max(4m-6,2m-3+\gamma)$ for $m \in \ooint{1,2}$ and $\gamma \in \ooint{0,m-1}$, we get that for all $k\in \{0, \ldots,T-1\}$,
\begin{equation}
  \limsup_{\norm{q_0} \to \plusinfty}  \sup_{\norm{p_0} \leq \norm{q_0}^{\gamma}} \defEns{\Ham(q_{k+1},p_{k+1})-\Ham(q_{k},p_{k})}/\norm{q_0}^{3m-4} < 0 \eqsp,
\end{equation}
which concludes the proof.
\item
  Consider now the case $\m=2$.
First   by \Cref{lem:grad_Lip_F} and \Cref{lem:bound_first_iterate_leapfrog_b}-\ref{lem:bound_first_iterate_leapfrog_b_2}, there exist $\bar{S}_1 \geq 0$ and  $R_1 \geq \rhtwo$ such that for all $T \in \nsets$ and $h \in \ocint{0,\bar{S}_1/T}$, $\q_0,p_0 \in \rset^d$ such that $ \norm{p_0} \leq
\norm{\q_0}^{\gamma}$ and $\norm{\q_0} \geq R_1$, and $i \in
\{0,\ldots,T\}$,
\begin{equation}\label{eq:bound_iterate_q_3_prood_diff_ham_3_0_2}
\norm{\q_{i}-\q_{0}}\leq (\delta/2) \norm{\q_0}  \eqsp.
\end{equation}
and
\begin{equation}
\label{eq:bound_iterate_q_3_prood_diff_ham_3_0_3}
\begin{aligned}
\norm{\q_{i}-\q_{0}}&\leq \norm{p_0} Th + (1/2) (T+1)^2 h^2 (\constzeroT + \constzero \delta/2) \norm{q_0}\eqsp, \\
 \norm{\p_i - \p_0} &\leq h T \{\constzeroT + (\constzeroT + \constzero \delta/2) \norm{q_0}  \} \eqsp,
\end{aligned}
\end{equation}
where $\constzero$ and $\constzeroT$ are defined in \Cref{lem:grad_Lip_F}. By  \Cref{lem:variation_assum_hessian}, \Cref{lem:prepa_bound_diff_ham}-\ref{lem:prepa_bound_diff_ham_2} and \eqref{eq:bound_iterate_q_3_prood_diff_ham_3_0_2}, there exists $R_2 \geq R_1$ such that for all $\q_0,\p_0 \in \rset^d$, $\norm{\q_0} \geq R_2$ and $\norm{\p_0} \leq \norm{\q_0}^{\gamma}$
\begin{equation}
\label{eq:bound_A_k_2}
  \inf_{\norm{p_0} \leq \norm{q_0}^{\gamma}} A_k \geq \BB_0 \norm{q_k}^{2} \geq
  \BB_0 (1-\delta/2)^{2} \norm{q_0}^{2} \eqsp.
\end{equation}
Hence,
\begin{equation}
\label{eq:bound_A_k_2_2}
\limsup_{\norm{q_0} \to \plusinfty} \sup_{\norm{p_0} \leq \norm{q_0}^{\gamma}}\defEns{A_k/\norm{q_0}^{2}} \geq   \BB_0 (1-\delta/2)^{2}  \eqsp.
\end{equation}
We now bound $B_k$. Using \Cref{assum:potential}-\ref{assum:potential:a}, \Cref{lem:grad_Lip_F} and \eqref{eq:bound_iterate_q_3_prood_diff_ham_3_0_3}, we get by \eqref{eq:definition-B_k} that there exists $\rmD_1 \geq 0$ which does not depend on $T$ and $h$ such that
\begin{equation}
\label{eq:bound_B_k_2}
\limsup_{\norm{q_0} \to \plusinfty} \sup_{\norm{p_0} \leq \norm{q_0}^{\gamma}}\defEns{\abs{B_k}/\norm{q_0}^{2}} \leq \rmD_1 h  \{(hT)^3 + (hT)^4\}  \eqsp.
\end{equation}
Combining \Cref{assum:potential}-\ref{assum:potential:a}, \Cref{lem:grad_Lip_F} and \eqref{eq:bound_iterate_q_3_prood_diff_ham_3_0_3} again, we get by crude estimate that there exists $\rmD_2 \geq 0$ which does not depend on $T$ and $h$ such that
\begin{equation}
  \label{eq:bound_D_k_2}
  \limsup_{\norm{q_0} \to \plusinfty} \sup_{\norm{p_0} \leq \norm{q_0}^{\gamma}} \defEns{\abs{D_k}/\norm{q_0}^{2}} \leq \rmD_2 (hT)^2 \eqsp.
\end{equation}
We finally bound the two terms $C_{k,1}$ and $C_{k,2}$. First,
using the same reasoning as for $B_k$, we get that there exists $\rmD_3 \geq 0$ which does not depend on $T$ and $h$ such that
\begin{equation}
\label{eq:bound_C_k_2}
\begin{aligned}
&\limsup_{\norm{q_0} \to \plusinfty} \sup_{\norm{p_0} \leq \norm{q_0}^{\gamma}} \defEns{\abs{C_{k,1}}/\norm{q_0}^{2}} < \rmD_3 h \{(hT)^4 + (hT)^5\} \eqsp,\\
&\limsup_{\norm{q_0} \to \plusinfty} \sup_{\norm{p_0} \leq \norm{q_0}^{\gamma}} \defEns{\abs{C_{k,2}}/\norm{q_0}^{1+\gamma}} < \infty \eqsp.
\end{aligned}
\end{equation}
Finally, arguing like in \eqref{eq:bound_A_k_2_2}, we get that
\begin{equation}
\label{eq:bound_C_k_2_2}
  \limsup_{\norm{q_0} \to \plusinfty}  \sup_{\norm{p_0} \leq \norm{q_0}^{\gamma}} \defEns{C_{k,3}/\norm{q_0}^{2}} < 0 \eqsp.
\end{equation}
Combining \eqref{eq:bound_A_k_2_2}-\eqref{eq:bound_B_k_2}-\eqref{eq:bound_D_k_2}-\eqref{eq:bound_C_k_2} and \eqref{eq:bound_C_k_2_2} in  \eqref{eq:diff_ham_k}, and using that  $2 \geq 1+ \gamma $ for  $\gamma \in \ooint{0,1}$, we get that for all $k\in \{0, \ldots,T-1\}$,
\begin{align}
  &  \limsup_{\norm{q_0} \to \plusinfty}  \sup_{\norm{p_0} \leq \norm{q_0}^{\gamma}} \defEns{\Ham(q_{k+1},p_{k+1})-\Ham(q_{k},p_{k})}/\norm{q_0}^{2} \\
  & \qquad \qquad\qquad \qquad \leq - \BB_0 (1-\delta/2)^{2} h^4 + \rmD_1\{(hT)^3 + (hT)^4\} h^3 \\
  & \qquad \qquad  \qquad \qquad \qquad + \rmD_2(hT)^2h^4 + \rmD_3\{(hT)^4 + (hT)^5\} h^4\eqsp.
\end{align}
Therefore, there exists $\bar{S}_4\leq \bar{S}_3$ such for any $T\in \nsets$, $h \in \ocint{0,\bar{S}_4/T^{3/2}}$,
\begin{equation}
    \limsup_{\norm{q_0} \to \plusinfty}  \sup_{\norm{p_0} \leq \norm{q_0}^{\gamma}} \defEns{\Ham(q_{k+1},p_{k+1})-\Ham(q_{k},p_{k})}/\norm{q_0}^{2}< 0 \eqsp,
  \end{equation}
  which completes the proof.
\end{enumerate}
\end{proof}

\subsubsection{Proof of \Cref{propo:accept_pertub}}
\label{sec:proof-crefth_accept_2}

\begin{lemma}
\label{lem:variation_assum_hessian_pertub}
Assume \Cref{ass:pertub}.
Then there exist $\delta \in \ooint{0,1}$, $R_1 \in \rset_+$ $\BB_1 \in \rset_+^*$ such that for all $\q,\x,z \in \rset^d$, with
\begin{equation}
\label{eq:hyp_variation_assum_hessian_pertub}
\norm{q} \geq R_1 \eqsp, \qquad  \max\parenthese{\norm{\q-\x},\norm{\q-z}} \leq \delta \norm{\q}  \eqsp,
\end{equation}
we have
\begin{equation}
  \ps{\Sigmabf \nabla \F(\x)}{\Sigmabf z} \geq \BB_1 \norm{\q}^{2}  \eqsp.
\end{equation}
\end{lemma}
\begin{proof}
Under \Cref{ass:pertub},  it can be easily checked that there
exists $\tilde{C}_U \geq 0$ (depending only on $\constfive$ and $\Sigmabf$) such that for all  $\q,\x,z \in \rset^d$ satisfying  \eqref{eq:hyp_variation_assum_hessian} for $\delta \in \ooint{0,1}$ and $R_1 \in \rset_+$,
\begin{equation}
 \ps{\Sigmabf \nabla \F(\x)}{\Sigmabf  z}  \geq \ps{\Sigmabf^2 q}{\Sigmabf q} - \tilde{C}_U  (\delta \norm{\q}^{2} + \norm[\rho]{\q})\eqsp, \quad \norm{q} \geq R_1 \eqsp.
\end{equation}
The proof is concluded by using that $\Sigmabf$ is definite positive and  taking $\delta$ sufficiently small and $R_1$ sufficiently large.
\end{proof}

\begin{proof}[Proof of \Cref{propo:accept_pertub}]
  Note that by \Cref{ass:pertub}, \Cref{lem:bound_first_iterate_leapfrog_b}-\ref{lem:bound_first_iterate_leapfrog_b_2}, \Cref{lem:prepa_bound_diff_ham}-\ref{lem:prepa_bound_diff_ham_2} and \Cref{lem:variation_assum_hessian_pertub}, there exists $\BB_1,\bar{S}_1 >0$, $R_1 \geq 0$, such that for any $T \in \nsets$, $h \in \ocint{0,\bar{S}_1/T}$, $q_0,p_0 \in \rset^d$, $\norm{p_0} \leq \norm{q_0}^{\gamma}$, $\norm{q_0} \geq \max(1,R_1)$ and $k,i \in \{0, \ldots,T\}$, $\norm{q_0} \leq 2 \norm{q_k} \leq 3 \norm{q_0}$, $\abs{\tilde{U}(q_k)} \leq C_1 \norm{q_0}^{\rho}$, 
  \begin{equation}
\label{eq:proof_pertub_accept_1}
 \ps{\Sigmabf \nabla U(q_i)}{\Sigma q_k} \geq \BB_1 \norm[2]{q_k}  \eqsp, \quad  \norm{\nabla U(q_i)} \leq C_1 \norm{q_k} \eqsp,
  \end{equation}
  where $q_k = \Phiverletq[T][k](q_0,p_0)$ and  $C_1= \max(4\constfive, 3(\norm{\Sigmabf} + 2 \constfive))$.
  Let now $T \in \nsets$, $h \in \ocint{0,\bar{S}_1/T}$ and denote for any $k \in \{0,\ldots,T\}$, $(q_k,p_k) = \Phiverlet[T][k](q_0,p_0)$ for $q_0,p_0 \in \rset^d$. We consider the following decomposition:
  \begin{equation}
\label{eq:proof_pertub_accept_2}
H(p_T,q_T)-H(p_0,q_0)=\sum_{k=0}^{T-1}\defEns{H(p_{k+1},q_{k+1})-H(p_{k},q_{k})} \eqsp.
\end{equation}
We show below that there exists $\bar{S} < \bar{S}_1$ such that, for all $h \geq 0$ and $T \geq 0$ satisfying $hT \leq \bar{S}$,
\begin{equation}
\label{eq:proof_pertub_accept_3}
\limsup_{\norm{q_0} \to \plusinfty} \sup_{\norm{p_0} \leq \norm{q_0}^{\gamma}} [ \defEns{H(p_{k+1},q_{k+1})-H(p_{k},q_{k})}/\norm[2]{q_0}] <0 \eqsp,
\end{equation}
from which  the proof follows.
First for any $q_0,p_0 \in \rset^d$, $k \in \{0,\ldots,T-1\}$, we have
\begin{equation}
\label{eq:proof_pertub_accept_4}
  H(p_{k+1},q_{k+1})-H(p_{k},q_{k}) = A_k + B_k + C_k \eqsp,
\end{equation}
where $2 A_k =  \ps{\Sigmabf q_{k+1}}{q_{k+1}} - \ps{\Sigmabf q_{k}}{q_{k}}$,
$B_k = \tilde{U}(q_{k+1}) - \tilde{U}(q_k)$, and $2 C_k = \norm[2]{p_{k+1}} - \norm[2]{p_k}$.
By \eqref{eq:proof_pertub_accept_1} and \Cref{ass:pertub}, we have
\begin{equation}
\label{eq:proof_pertub_accept_6}
  \lim_{\norm{q_0} \to \plusinfty} \sup_{\norm{p_0} \leq \norm{q_0}^\gamma} \abs{B_k}/\norm[2]{q_0} =0 \eqsp,
\end{equation}
and
\begin{align}
  \label{eq:proof_pertub_accept_7}
  A_k &= h\ps{\Sigmabf p_k}{q_k} + h^2 \ps{\Sigmabf p_k}{p_k}/2 -h^2 \ps{\Sigmabf q_k }{\Sigmabf q_k }/2 \\
  & \qquad \qquad - h^3 \ps{\Sigmabf p_k}{\Sigmabf q_k }/2  +  h^4 \ps{\Sigmabf^2 q_k}{\Sigmabf q_k } /8  + A_{k,1}  \eqsp,
\end{align}
\begin{align}
    C_k & = -h\ps{\Sigmabf p_k}{q_k} -h^2 \ps{\Sigmabf p_k}{p_k}/2 +h^2 \ps{\Sigmabf q_k }{\Sigmabf q_k }/2\\
  & \qquad \quad  +3h^3 \ps{\Sigmabf p_k}{\Sigmabf q_k}/4  + h^4 \ps{\Sigmabf p_k }{\Sigmabf p_k}/8  
    -h^4 \ps{\Sigmabf^2 q_k }{\Sigmabf q_k} /4 \\
  &\qquad \quad -h^5 \ps{\Sigmabf^2 p_k}{q_k}/8 
    + h^6 \ps{\Sigmabf^2 q_k }{\Sigmabf^2 q_k}/32  + C_{k,1} \eqsp,
    \label{eq:proof_pertub_accept_8}
\end{align}
where
\begin{equation}
\label{eq:proof_pertub_accept_9}
  \lim_{\norm{q_0} \to \plusinfty} \sup_{\norm{p_0} \leq \norm{q_0}^\gamma} \{\abs{A_{k,1}} + \abs{C_{k,1}}\}/\norm[2]{q_0} =0 \eqsp,
\end{equation}
Using \eqref{eq:proof_pertub_accept_4}, \eqref{eq:proof_pertub_accept_7} and \eqref{eq:proof_pertub_accept_8}, we obtain that for any $q_0,p_0 \in \rset^d$,
\begin{equation}
  \label{eq:proof_pertub_accept_10}
  H(q_{k+1},p_{k+1}) - H(q_k,p_k) = D_k + A_{k,1}+  B_k + C_{k,1} \eqsp,
\end{equation}
where
\begin{align}
  D_k &= h^3 \ps{\Sigmabf p_k}{\Sigmabf q_k}/4 + h^4 \ps{\Sigmabf p_k }{\Sigmabf p_k}/8   -h^4 \ps{\Sigmabf^2 q_k }{\Sigmabf q_k} /8 \\
&  \qquad \qquad \qquad -h^5 \ps{\Sigmabf^2 p_k}{q_k}/8 + h^6 \ps{\Sigmabf^2 q_k }{\Sigmabf^2 q_k}/32 \eqsp.
\end{align}
Using that for $k \in \{1,\ldots,T\}$,  $p_k= p_0 -(h/2)\{\nabla U(q_0) + \nabla U(q_k)\} - h \sum_{i=1}^{k-1} \nabla U(q_i)$ and \eqref{eq:proof_pertub_accept_1}, we obtain that for any $k\in \{1,\ldots,T\}$ and $q_0,p_0$, $\norm{q_0} \geq \max(1,R_1)$, $\norm{p_0} \geq \norm[\gamma]{q_0}$,
\begin{align}
  D_k &\leq -h^4 k \BB_1 \norm[2]{q_k}/8 + h^6k^2 \norm{\Sigmabf}^2 C_1 \norm[2]{q_k}  - h^4 \ps{\Sigmabf^2 q_k }{\Sigmabf q_k} /8 \\
      &  \qquad \qquad\qquad \qquad+  h^6 k C_1 \norm{\Sigmabf}^2 \norm[2]{q_k}/8 +  h^6 \norm{\Sigmabf}^4 \norm[2]{q_k}/32     + D_{k,1} \eqsp,
\end{align}
where
\begin{equation}
\label{eq:proof_pertub_accept_lim_D}
    \lim_{\norm{q_0} \to \plusinfty} \sup_{\norm{p_0} \leq \norm{q_0}^\gamma} \abs{D_{k,1}}/\norm[2]{q_0} =0\eqsp.
  \end{equation}
  Define
  \begin{equation}
\label{eq:proof_pertub_accept_def_S_2}
    \bar{S}_2 = \min\defEns{ S \in \ocint{0,\bar{S}_1} \, : \, S^2 ( 2 C_1  \norm{\Sigmabf}^2 + \norm{\Sigmabf}^4)  - \BB_1/8 \geq -\BB_1/16} \eqsp.
  \end{equation}
  Then, if $Th \leq \bar{S_2}$ for any  $q_0,p_0$, $\norm{q_0} \geq \max(1,R_1)$, $\norm{p_0} \geq \norm[\gamma]{q_0}$, we get that
  \begin{equation}
\label{eq:proof_pertub_accept_bound_D}
    D_k \leq -\BB_1h^4 k \norm[2]{q_k}/16 + D_{k,1} \eqsp.
  \end{equation}
  Similarly using that $\Sigmabf$ is definite positive, we obtain that there exist $\BB_2 >0$ and  $\bar{S}_3  \in \ocint{0,\bar{S}_1}$ such that if $hT \leq \bar{S}_3$, for any  $q_0,p_0$, $\norm{q_0} \geq \max(1,R_1)$, $\norm{p_0} \geq \norm[\gamma]{q_0}$, we get that
  \begin{equation}
\label{eq:proof_pertub_accept_bound_D_0}
    D_0 \leq -\BB_2 \norm[2]{q_0} + D_{0,1} \eqsp,
 \quad \text{   where }
    \lim_{\norm{q_0} \to \plusinfty} \sup_{\norm{p_0} \leq \norm{q_0}^\gamma} \abs{D_{0,1}}/\norm[2]{q_0} =0\eqsp.
  \end{equation}
  Combining \eqref{eq:proof_pertub_accept_6}-\eqref{eq:proof_pertub_accept_9}-\eqref{eq:proof_pertub_accept_lim_D}-\eqref{eq:proof_pertub_accept_bound_D} and \eqref{eq:proof_pertub_accept_bound_D_0} in \eqref{eq:proof_pertub_accept_10}, we obtain that \eqref{eq:proof_pertub_accept_3} holds with $\bar{S} = \min(\bar{S}_2,\bar{S}_3)$ since \eqref{eq:proof_pertub_accept_1} implies that  $\norm{q_k} \geq \norm{q_0}/2$.
\end{proof}


\appendix
\section{Harris recurrence for mixture of Metropolis-Hastings type Markov kernels}
\label{sec:harr-recurr-metr}
Let $(\Xset,\Xtribu)$ be a measurable space and $\lambda$ be a
$\sigma$-finite measure on $\Xtribu$.  For all $i \in \nset^*$,
let $\alpha_i: \Xset \times \Xset \to \ccint{0,1}$ be a measurable function and   $\qker_i : \Xset
\times \Xset \to \ccint{0,\plusinfty}$ be a Markov transition density \wrt\ $\lambda$. Consider the Markov kernel
$\kernel_i$ on $\Xset \times \Xtribu$ defined by
\begin{equation}
  \label{eq:form_MH_gene}
  \kernel_i(x,\eventA) = \int_{\eventA} \alpha_i(x,y) \qker_i(x,y) \lambda(\rmd y ) + \updelta_x(\eventA) r_i(x) \eqsp, \quad  \text{$x \in \Xset$  and $\eventA \in \Xtribu$,}
\end{equation}
where for all $x
\in \Xset$
\begin{equation}
  \label{eq:def_r_i_harris_tierney}
  r_i(x) = 1 - \int_{\Xset} \alpha_i(x,y) \qker_i(x,y) \lambda(\rmd y ) \eqsp.
\end{equation}
For instance,  $\kernel_i$ may be a Markov kernel associated to the Metropolis-Hastings
algorithm, \ie
\begin{equation}
\label{eq:definition-MH-ratio}
  \alpha_i(x,y) =
  \begin{cases}
\min\parentheseDeux{1, \frac{\pi(y) \qker_i(y,x)}{\pi(x) \qker_i(x,y)}} \eqsp, & \text{ if } \pi(x) \qker_i(x,y) >0 \eqsp,\\
1\eqsp, & \text{otherwise} \eqsp,
  \end{cases}
\end{equation}
for some probability density $\pi: \Xset \to \coint{0,\plusinfty}$
with respect to $\lambda$.
We use the results below in the case
where for any $i \in \nsets$, $\kernel_i$ is a Markov kernel associated to the HMC algorithm.
\cite[Corollary 2]{tierney:1994}
considers Metropolis-Hastings kernels $\kernel_i$ with $\alpha_i$ defined by \eqref{eq:definition-MH-ratio} and shows that that if $\kernel_i$ is irreducible, then $\kernel_i$ is Harris recurrent. We extend this
result to kernels $\kernel_i$ of the form \eqref{eq:form_MH_gene} (but that do not satisfy \eqref{eq:definition-MH-ratio}) and  mixture of Markov kernels $\kernel_{\bfvarpi}$ defined on
$(\Xset, \Xtribu)$ by
\begin{equation}
\label{eq:mixture_kernel_Harris}
  \kernel_{\bfvarpi} = \sum_{i \in \nset^*} \varpi_i \kernel_i
\end{equation}
where $(\varpi_i)_{i\in
  \nset^*}$ is a sequence of non-negative numbers satisfying $\sum_{i
  \in \nset^*} \varpi_i = 1$.

\begin{proposition}
  \label{propo:harris_rec}
  Let $\kernel_{\bfvarpi}$ be the Markov kernel given by
  \eqref{eq:mixture_kernel_Harris} and associated with the sequence of
  Markov kernel $(\kernel_i)_{i \in \nset^*}$ given by
  \eqref{eq:form_MH_gene}.  Let $\pi$ be a probability measure on
  $(\Xset,\Xtribu)$. Assume that $\pi$ and $\lambda$ are mutually
  absolutely continuous and for all $i \in \nset^*$, $\pi$ is invariant for $\kernel_i$. If $\kernel_{\bfvarpi}$ is irreducible and there exists $i \in \nset^*$ such that $\varpi_i >0$ and for all $x \in \Xset$ $r_i(x) <1$, with $r_i$ defined by \eqref{eq:def_r_i_harris_tierney},  then $\kernel$ is Harris
  recurrent.
\end{proposition}

\begin{proof}
A bounded measurable function is said to be harmonic if  $\kernel_{\bfvarpi}\harmonic = \harmonic$.
By \cite[Theorem 17.1.4, Theorem
17.1.7]{meyn:tweedie:2009} a Markov kernel $\kernel_{\bfvarpi}$ is Harris recurrent if $\kernel_{\bfvarpi}$ is recurrent and any
bounded harmonic function $\harmonic : \rset^d \to \rset$ is constant.
By \cite[Theorem 10.1.1]{meyn:tweedie:2009}, since $\kernel_{\bfvarpi}$ is irreducible and admits $\pi$ as an invariant probability measure, then $\kernel_{\bfvarpi}$ is recurrent.
On the other hand, any bounded harmonic function $\phi$ is $\pi$-almost surely equal to $\pi(\phi)$ by \cite[Theorem 17.1.1, Lemma 17.1.1]{meyn:tweedie:2009}.
Using that $\pi$ and $\lambda$ are mutually
absolutely continuous, and $\pi$ is an invariant probability measure for $\kernel_i$ for all $i \in \nset^*$,  we get by \eqref{eq:form_MH_gene} that for all $x \in \Xset$
\begin{equation}
   \kernel_{\bfvarpi} \harmonic(x) =  \sum_{i\in \nset^*} \varpi_i \defEns{\pi(\harmonic) (1-r_i(x)) + \harmonic(x)r_i(x)} \eqsp.
\end{equation}
Combining this result with $\kernel_{\bfvarpi} \phi = \phi$, we get for all $x \in \Xset$
\begin{equation}
\{\harmonic(x)-\pi(\harmonic)\} \sum_{i\in \nset^*} \varpi_i \{1-r_i(x)\}= 0\eqsp.
\end{equation}
The condition that there exists $i \in \nset^*$ such that $\varpi_i >0$ and for all $x \in \Xset$ $r_i(x) <1$, implies  that for all $x \in \Xset$, $\harmonic(x) = \pi(\harmonic)$.
\end{proof}

\bibliographystyle{plain}
\bibliography{../Bibliography/bibliography}

\begin{thebibliography}{10}

\bibitem{douc:moulines:priouret:2018}


\bibitem{bakry:gentil:ledoux:2014}
D.~Bakry, I.~Gentil, and M.~Ledoux.
\newblock {\em Analysis and geometry of {M}arkov diffusion operators}, volume
  348 of {\em Grundlehren der Mathematischen Wissenschaften [Fundamental
  Principles of Mathematical Sciences]}.
\newblock Springer, Cham, 2014.

\bibitem{betancourt-bernoulli:2017}
M.~Betancourt, S.~Byrne, S.~Livingstone, and M.~Girolami.
\newblock {The geometric foundations of Hamiltonian Monte Carlo}.
\newblock {\em {BERNOULLI}}, {23}({4A}):{2257--2298}, {NOV} {2017}.

\bibitem{bou-rabee:sanz-serna:2018}
N.~Bou-Rabee and J.~M. Sanz-Serna.
\newblock Geometric integrators and the {H}amiltonian {M}onte {C}arlo method.
\newblock {\em Acta Numerica}, pages 1--92, 2018.

\bibitem{bou:sanz:2017}
N.~Bou-Rabee and J.M. Sanz-Serna.
\newblock Randomized {H}amiltonian {M}onte {C}arlo.
\newblock {\em The Annals of Applied Probability}, 27(4):2159--2194, 2017.

\bibitem{byrne:girolami:2013}
S.~Byrne and M.~Girolami.
\newblock {Geodesic Monte Carlo on Embedded Manifolds}.
\newblock {\em Scandinavian Journal of Statistics}, 40(4):825--845, December
  2013.

\bibitem{cances:legoll:stoltz}
E.~Canc\`es, F.~Legoll, and G.~Stoltz.
\newblock Theoretical and numerical comparison of some sampling methods for
  molecular dynamics.
\newblock {\em M2AN Math. Model. Numer. Anal.}, 41(2):351--389, 2007.

\bibitem{duane:1987}
S.~Duane, A.D. Kennedy, B.~J. Pendleton, and D.~Roweth.
\newblock Hybrid monte carlo.
\newblock {\em Physics Letters B}, 195(2):216 -- 222, 1987.

\bibitem{duistermaat:kolk:2004}
J.~J. Duistermaat and J.~A.~C. Kolk.
\newblock {\em Multidimensional real analysis. {I}. {D}ifferentiation},
  volume~86 of {\em Cambridge Studies in Advanced Mathematics}.
\newblock Cambridge University Press, Cambridge, 2004.
\newblock Translated from the Dutch by J. P. van Braam Houckgeest.

\bibitem{hairer:wanner:lubish:2002}
C.~Lubich E.~Hairer, G.~Wanner.
\newblock {\em Geometric Numerical Integration: Structure-Preserving Algorithms
  for Ordinary Differential Equations}.
\newblock Springer Series in Computational Mathematics 31. Springer Berlin
  Heidelberg, 2nd ed edition, 2002.

\bibitem{fang:sanz-serna:skeel:2016}
Y.~Fang, J.~M. Sanz-Serna, and R.~D. Skeel.
\newblock {Compressible generalized hybrid Monte Carlo (vol 140, 174108,
  2014)}.
\newblock {\em JOURNAL OF CHEMICAL PHYSICS}, {144}({2}), {JAN 14} {2016}.

\bibitem{girolami:calderhead:2011}
M.~Girolami and B.~Calderhead.
\newblock Riemann manifold {L}angevin and {H}amiltonian {M}onte {C}arlo
  methods.
\newblock {\em J. R. Stat. Soc. Ser. B Stat. Methodol.}, 73(2):123--214, 2011.
\newblock With discussion and a reply by the authors.

\bibitem{leimkuhler:reich:2004}
B.~Leimkuhler and S.~Reich.
\newblock {\em Simulating {H}amiltonian dynamics}, volume~14 of {\em Cambridge
  Monographs on Applied and Computational Mathematics}.
\newblock Cambridge University Press, Cambridge, 2004.

\bibitem{liu:2008}
J.~S. Liu.
\newblock {\em Monte {C}arlo strategies in scientific computing}.
\newblock Springer Series in Statistics. Springer, New York, 2008.

\bibitem{livingstone:betancourt:byrne:girolami:2016}
S.~Livingstone, M.~Betancourt, S.~Byrne, and M.~Girolami.
\newblock On the geometric ergodicity of {H}amiltonian {M}onte {C}arlo.
\newblock {\em arXiv preprint arXiv:1601.08057v2}, 2016.

\bibitem{livingstone:faulkner:roberts:2017}
S.~Livingstone, M.~F Faulkner, and G.~O. Roberts.
\newblock Kinetic energy choice in hamiltonian/hybrid monte carlo.
\newblock {\em arXiv preprint arXiv:1706.02649}, 2017.

\bibitem{lu:et:al:2016}
X.~Lu, V.~Perrone, L.~Hasenclever, Y.~W. Teh, and S.~Vollmer.
\newblock Relativistic monte carlo.
\newblock {\em arXiv preprint arXiv:1609.04388}, 2016.

\bibitem{mengersen:tweedie:1996}
K.~L. Mengersen and R.~L. Tweedie.
\newblock Rates of convergence of the hastings and metropolis algorithms.
\newblock {\em Ann. Statist.}, 24(1):101--121, 02 1996.

\bibitem{meyn:tweedie:2009}
S.~Meyn and R.~Tweedie.
\newblock {\em {M}arkov Chains and Stochastic Stability}.
\newblock Cambridge University Press, New York, NY, USA, 2nd edition, 2009.

\bibitem{neal:1993}
R.~M. Neal.
\newblock Bayesian learning via stochastic dynamics.
\newblock {\em Advances in neural information processing systems}, pages
  475--475, 1993.

\bibitem{neal:2011}
R.~M. Neal.
\newblock {MCMC} using {H}amiltonian dynamics.
\newblock {\em Handbook of Markov Chain Monte Carlo}, pages 113--162, 2011.

\bibitem{outerelo:ruiz:2009}
E.~Outerelo and J.~M. Ruiz.
\newblock {\em Mapping degree theory}, volume 108.
\newblock American Mathematical Society Providence, RI, 2009.

\bibitem{roberts:tweedie:1996}
G.~O. Roberts and R.~L. Tweedie.
\newblock Exponential convergence of {L}angevin distributions and their
  discrete approximations.
\newblock {\em Bernoulli}, 2(4):341--363, 1996.

\bibitem{roberts:tweedie:1996:biometrika}
G.~O. Roberts and R.~L. Tweedie.
\newblock Geometric convergence and central limit theorems for multidimensional
  {Hastings and Metropolis algorithms}.
\newblock {\em Biometrika}, 83(1):95--110, 1996.

\bibitem{rudin:1987}
W.~Rudin.
\newblock {\em Real and complex analysis}.
\newblock McGraw-Hill Book Co., New York, third edition, 1987.

\bibitem{sanz-serna:2014}
J.~M. Sanz-Serna.
\newblock Markov chain {M}onte {C}arlo and numerical differential equations.
\newblock In {\em Current challenges in stability issues for numerical
  differential equations}, volume 2082 of {\em Lecture Notes in Math.}, pages
  39--88. Springer, Cham, 2014.

\bibitem{schofield:barker:gelman:cook:briffa:2016}
M.~R. Schofield, R.~J. Barker, A.~Gelman, E.~R. Cook, and K.~R. Briffa.
\newblock A model-based approach to climate reconstruction using tree-ring
  data.
\newblock {\em J. Amer. Statist. Assoc.}, 111(513):93--106, 2016.

\bibitem{tang:srivastava:salakhutdinov:2014}
Y.~Tang, N.~Srivastava, and R.~R. Salakhutdinov.
\newblock Learning generative models with visual attention.
\newblock In {\em Advances in Neural Information Processing Systems}, pages
  1808--1816, 2014.

\bibitem{tierney:1994}
L.~Tierney.
\newblock {M}arkov chains for exploring posterior disiributions (with
  discussion).
\newblock {\em Ann. Statist.}, 22(4):1701--1762, 1994.

\end{thebibliography}
\end{document}